\def\isarxiv{1} 

\ifdefined\isarxiv
\documentclass[11pt]{article}

\usepackage[numbers]{natbib}

\else
\documentclass{article}
\usepackage[accepted]{icml2025}
\fi

\usepackage{amsmath}
\usepackage{amsthm}
\usepackage{amssymb}
\usepackage{algorithm}
\usepackage{subfig}
\usepackage{algpseudocode}
\usepackage{graphicx}
\usepackage{grffile}
\usepackage{wrapfig,epsfig}
\usepackage{url}
\usepackage{xcolor}
\usepackage{epstopdf}
\usepackage[utf8]{inputenc}

\usepackage{bbm}
\usepackage{dsfont}

\allowdisplaybreaks

\ifdefined\isarxiv

\usepackage{tikz}
\usepackage{hyperref}  
\hypersetup{colorlinks=true,citecolor=blue,linkcolor=blue} 
\usetikzlibrary{arrows}
\usepackage[margin=1in]{geometry}

\else

\usepackage{microtype}
\usepackage{hyperref}
\definecolor{mydarkblue}{rgb}{0,0.08,0.45}
\hypersetup{colorlinks=true, citecolor=mydarkblue,linkcolor=mydarkblue}

\fi

\newtheorem{theorem}{Theorem}[section]
\newtheorem{lemma}[theorem]{Lemma}
\newtheorem{definition}[theorem]{Definition}

\newtheorem{corollary}[theorem]{Corollary}

\newtheorem{assumption}[theorem]{Assumption}

\newtheorem{fact}[theorem]{Fact}
\newtheorem{remark}[theorem]{Remark}
\newtheorem{claim}[theorem]{Claim}

\newcommand{\wh}{\widehat}
\newcommand{\wt}{\widetilde}
\newcommand{\ov}{\overline}

\newcommand{\R}{\mathbb{R}}

\renewcommand{\d}{\mathrm{d}}
\renewcommand{\tilde}{\wt}
\renewcommand{\hat}{\wh}
\renewcommand{\epsilon}{\varepsilon}
\newcommand{\Tmat}{{\cal T}_{\mathrm{mat}}}
\newcommand{\DP}{\mathrm{DP}}
\newcommand{\rpm}{\textsc{ReportOneSidedNoisyArgMax}}

\DeclareMathOperator*{\E}{{\mathbb{E}}}

\DeclareMathOperator*{\Z}{\mathbb{Z}}

\DeclareMathOperator{\OPT}{OPT}

\DeclareMathOperator{\supp}{supp}
\DeclareMathOperator{\poly}{poly}

\DeclareMathOperator{\nnz}{nnz}

\DeclareMathOperator{\TV}{TV}

\makeatletter
\newcommand*{\RN}[1]{\expandafter\@slowromancap\romannumeral #1@}
\makeatother

\usepackage{lineno}

\begin{document}

\ifdefined\isarxiv

\date{}

\title{On Differential Privacy for Adaptively Solving Search Problems via Sketching}
\author{
Shiyuan Feng\thanks{\texttt{fsyo@stu.pku.edu.cn}. Peking University.}
\and
Ying Feng\thanks{\texttt{yingfeng@mit.edu}. MIT.}
\and
George Z. Li\thanks{\texttt{gzli@cs.cmu.edu}. Carnegie Mellon University.}
\and 
Zhao Song\thanks{\texttt{magic.linuxkde@gmail.com}. University of California, Berkeley.}
\and 
David P. Woodruff\thanks{\texttt{dwoodruf@cs.cmu.edu}. Carnegie Mellon University.}
\and
Lichen Zhang\thanks{\texttt{lichenz@mit.edu}. MIT.}
}

\else

\twocolumn[
\icmltitle{On Differential Privacy for Adaptively Solving Search Problems via Sketching}



\icmlsetsymbol{equal}{*}

\begin{icmlauthorlist}
\icmlauthor{Shiyuan Feng}{PKU}
\icmlauthor{Ying Feng}{MIT}
\icmlauthor{George Z. Li}{CMU}
\icmlauthor{Zhao Song}{Berkeley}
\icmlauthor{David P. Woodruff}{CMU}
\icmlauthor{Lichen Zhang}{MIT}
\end{icmlauthorlist}

\icmlaffiliation{PKU}{Peking University}
\icmlaffiliation{CMU}{Carnegie Mellon University}
\icmlaffiliation{MIT}{MIT}
\icmlaffiliation{Berkeley}{University of California, Berkeley}

\icmlcorrespondingauthor{Zhao Song}{magic.linuxkde@gmail.com}
\icmlcorrespondingauthor{Lichen Zhang}{lichenz@mit.edu}

\icmlkeywords{Nearest-Neighbor Search, Adaptive Data Structure, Differential Privacy}

\vskip 0.3in
]



\printAffiliationsAndNotice{}  
\fi

\ifdefined\isarxiv
\begin{titlepage}
  \maketitle
  \begin{abstract}
Recently differential privacy has been used for a number of streaming, data structure, and dynamic graph problems as a means of hiding the internal randomness of the data structure, so that multiple possibly adaptive queries can be made without sacrificing the correctness of the responses. Although these works use differential privacy to show that for some problems it is possible to tolerate $T$ queries using $\widetilde{O}(\sqrt{T})$ copies of a data structure, such results only apply to {\it numerical estimation problems}, and only return the {\it cost} of an optimization problem rather than the solution itself. In this paper, we investigate the use of differential privacy for adaptive queries to {\it search} problems, which are significantly more challenging since the responses to queries can reveal much more about the internal randomness than a single numerical query. We focus on two classical search problems: nearest neighbor queries and regression with arbitrary turnstile updates. We identify key parameters to these problems, such as the number of $c$-approximate near neighbors and the matrix condition number, and use different differential privacy techniques to design algorithms returning the solution vector with memory and time depending on these parameters. We give algorithms for each of these problems that achieve similar tradeoffs.

  \end{abstract}
  \thispagestyle{empty}
\end{titlepage}


\else

\begin{abstract}
Recently differential privacy has been used for a number of streaming, data structure, and dynamic graph problems as a means of hiding the internal randomness of the data structure, so that multiple possibly adaptive queries can be made without sacrificing the correctness of the responses. Although these works use differential privacy to show that for some problems it is possible to tolerate $T$ queries using $\widetilde{O}(\sqrt{T})$ copies of a data structure, such results only apply to {\it numerical estimation problems}, and only return the {\it cost} of an optimization problem rather than the solution itself. In this paper, we investigate the use of differential privacy for adaptive queries to {\it search} problems, which are significantly more challenging since the responses to queries can reveal much more about the internal randomness than a single numerical query. We focus on two classical search problems: nearest neighbor queries and regression with arbitrary turnstile updates. We identify key parameters to these problems, such as the number of $c$-approximate near neighbors and the matrix condition number, and use different differential privacy techniques to design algorithms returning the solution vector with memory and time depending on these parameters. We give algorithms for each of these problems that achieve similar tradeoffs.

\end{abstract}

\fi

\section{Introduction}
In modern algorithmic design, data structures are playing a more and more important role, in particular in the recent breakthroughs in optimization, machine learning and graph problems. These highly efficient data structures are usually randomized, with one crucial distinction from traditional data structures: one needs these data structures to be \emph{robust against an adaptive adversary}, in which the adversary could design the input to the data structure based on the prior outputs of the data structure. This is often because these data structures are incorporated in an iterative process, where the data structure takes an input formed by the process, returns a result, and the process utilizes the output to form a new input. Hence, the input to a data structure may be highly correlated with the \emph{internal randomness} of the data structure. Most traditional randomized data structures, unfortunately, are not robust enough to handle this type of task, as they only have probabilistic guarantees against \emph{oblivious adversaries}\footnote{In the remainder of the paper, we will refer to data structures that are robust against an adaptive adversary as adaptive data structures, and data structures that are robust against an oblivious adversary as oblivious data structures.}, in which the input sequence to the data structure is predetermined, and independent of its internal randomness. 

To address this challenge, various reductions have been proposed to build an adaptive data structure from oblivious data structures~\cite{bjwy20,cn20,cn22,bkm+21,hkm+22,bgjllps22,syyz23_dp,bdf+24}. These reductions usually proceed as follows: (1) create $k$ independent copies of the oblivious data structure, (2) during the query phase sample and query a subset of the data structures, and (3) aggregate the subsampled answers. 

We focus on the following \emph{search} data structure problem: we are given $n$ points in $\R^d$ represented by a matrix $U\in \R^{n\times d}$ that is allowed to be preprocessed. Given a sequence of adaptive queries or updates $\{v_1,\ldots,v_T\}$, we need to design a data structure that answers queries efficiently and succeeds with high probability. Throughout the paper we will consider two key examples of this: (1) the first type of problem is the approximate near-neighbor problem, where a query is a $c$-approximate near-neighbor of $v_t$ in $U$, and (2) the second type of problem is the regression problem, where $U$ is in addition augmented by a response vector $b\in \R^n$, it could be the solution to the regression problem $\arg\min_{x\in \R^d} \|Ux-b\|_2^2$ where either $U$ or $b$ can also be updated by $v_t$. Note that one can interpret regression as an unconstrained search problem over $\R^d$.

For a sequence of $T$ adaptive queries or updates, perhaps the simplest reduction is to prepare $T$ copies of an oblivious data structure and for each query, use a fresh new data structure to respond. Since a data structure is never used more than once, an adaptive adversary cannot design adversarial inputs to this data structure, and hence we obtain an adaptive data structure. In fact, this simple approach already has led to many improvements in convex optimization~\cite{lsz19,sy21,jswz21,qszz23}. However, in many applications, the overhead $T$ in both preprocessing time and space usage is prohibitive, and more time- and space-efficient alternatives are needed. If in addition the adaptive queries are vectors in $\R^d$ and each data structure succeeds with constant probability, then one can adapt an $\epsilon$-net argument and show that $\wt O(d)$ \footnote{We use $\wt O(\cdot)$ to suppress polylogarithmic factors in $n$, $d$ and $1/\delta$ where $\delta$ is the failure probability of the data structure.} oblivious data structures suffice; at query time, one can then sample a logarithmic number of these data structures and output the the best answer. This approach has been widely adapted for problems such as distance estimation~\cite{cn20,cn22} and approximate near-neighbor search~\cite{sxz22,bdf+24}, but is insufficient if $d\approx T$ or even $d\gg T$, which is the case for convex optimization. Moreover, for many machine learning applications, $T$ is often a hyperparameter that could be chosen and is much smaller than $d$, making the approach based on $\epsilon$-net arguments inefficient for these applications. 

Is it possible to improve upon the solutions that use $\min\{d, T\}$ data structures? In pioneering work of ~\cite{hkm+22} and further extended by~\cite{bkm+21}, it was shown that \emph{differential privacy} could be utilized to reduce such overhead. In particular, for numerical estimation problems, where one only needs to output a numerical value, these works show that $\wt O(\sqrt T)$ data structures are sufficient. The idea is to view the internal randomness of the oblivious data structures as the private database one would like to protect. Concretely, one prepares $\wt O(\sqrt T)$ oblivious data structures, and for each adaptive query, samples $\wt O(1)$ of the data structures, and outputs the differentially private median of the answers~\cite{bkm+21}. By the advanced composition theorem~\cite{drv10}, $\wt O(\sqrt T)$ data structures suffice to provide a correct answer to the numerical estimation problem. This leads to drastically improved algorithms when $\sqrt T\ll d$, and in~\cite{bkm+21} they apply it to a large number of graph algorithms. This technique was later extended in~\cite{csw+23} for estimating the \emph{cost} of a regression problem and the \emph{value} of kernel density estimation. 

Can these techniques be extended beyond numerical estimation, to search problems? This is a natural question, as instead of the regression cost, one could naturally be more interested in estimating the regression solution and using it to, e.g., label future examples. On the other hand, simply extending the private median framework to high dimensions seems insufficient, as the regression solution could potentially reveal much more about the internal randomness of a data structure than what only the cost would reveal. This problem becomes more evident for approximate near neighbor search: for this problem, if there exists a row $u^*$ of $U$ for which $\|u^*-v\|\leq r$ for an adaptively chosen query $v\in \R^d$ and a distance parameter $r$, then the data structure needs to output a row $\wt u$ of $U$ with $\|\wt u-v\|\leq cr$ for $c>1$. It is not clear how to extend the differential privacy framework to such a scenario, as we will have to return a row of $U$, while the private median mechanism would not be able to do this. Motivated by these problems and barriers, we ask
\begin{center}
    {\it Is it possible to design robust search data structures with fewer than $\min\{ d, T\}$ independent copies?}
\end{center}
We provide an affirmative answer to this question under mild assumptions on these problems. We start by introducing the definition of the $(c, r)$-Approximate Near Neighbor ($(c, r)$-ANN) problem and its corresponding assumption.
\begin{definition}[$(c,r)$-Approximate Near Neighbor]
\label{def:cr_ann}
Let $U\subseteq \R^d$ be a dataset, $v\in \R^d$ be a query point, $c>1$ be the approximation parameter, and $r>0$ be the distance parameter. The \emph{$(c,r)$-approximate near neighbor} problem asks, if there exists a point $u_*\in U$ with $\|u_*-v\|\leq r$, then return a point $\wt u\in U$ with $\|\wt u-v\|\leq cr$. Otherwise, the data structure can either return nothing or any point $\wt u\in U$ with $\|\wt u-v\|\leq cr$.\footnote{The standard $(c,r)$-ANN definition allows the data structure to output any point in $U$ if there is no point within distance $r$ of the query $v$. Here we restrict it to either return nothing or return a point within distance $cr$, as most popular ANN data structures, such as locality-sensitive hashing, satisfy this specification.}
\end{definition}
We now state the assumption imposed on the ANN problem, in its most general form. Given a query $v$ and a data matrix $U$, we will use $f_v(U)$ to denote the set of candidate solution of querying $v$, For example, the predicate function for $(c, r)$-ANN can be defined as $f_v(U) = U\cap B(v, cr)$
where $B(v, cr)$ is the ball centered at $v$ with radius $cr$ under some norm. The assumption is:
\begin{assumption}
\label{assumption:general}
Let $U\subseteq \R^d$ be an $n$-point dataset, and let $\{v_1,v_2,\ldots,v_T\}\subseteq \R^d$ be a sequence of (possibly adaptive) queries. We assume for all $t\in [T]$, $|f_{v_t}(U)| \leq s$.
\end{assumption}

The assumption states that the ball around $v_t$ of radius $cr$ cannot intersect with more than $s$ points in $U$, which is equivalent to that the query $v_t$ cannot have too many approximate near neighbors in $U$, when $v_t$ has an $r$-near neighbor. To facilitate a comparison with other popular assumptions for this task, such as constant expansion and bounded doubling dimension, we state a version of the assumption for ANN that solely depends on the matrix $U$:
\begin{assumption}
\label{assumption:ann_sparse}
Let $U\subseteq \R^d$, and let $\{B_u \}_{u\in U}$ be the collection of norm balls with $B_u=B(u, cr)$. We assume for each $u\in U$, that $B_u$ intersects at most $s$ other distinct balls in the collection.
\end{assumption}
Let us first see how this assumption implies the condition on $f_v(U)$. By the triangle inequality, any two points in $B(v, cr)$ have distance at most $2cr$, that is, the norm balls with these two points being the respective centers with radius $cr$ must intersect. If $v$ has more than $s+1$ approximate nearest neighbors, then there must exist some $u\in B(v,cr)$ whose norm ball intersects with more than $s$ other distinct balls, a contradiction. Hence Assumption~\ref{assumption:ann_sparse} implies Assumption~\ref{assumption:general} with $s+1$ in place of $s$. Intuitively, Assumption~\ref{assumption:ann_sparse} states that the dataset $U$ is not too dense, and in particular each point in $U$ does not have too many close neighbors. This structural assumption can also be achieved by preprocessing $U$: one could run a clustering algorithm on $U$ to group points that are close to each other into a cluster, and replace a cluster using its center. Then, the downstream ANN algorithm is performed on these centers; once a center is returned, one could return a point in the corresponding cluster. This approach has been implemented in various approximate nearest neighbor search libraries in practice, and has been a driving force for the Google ScaNN framework~\cite{p20,gsl+20}. 
For efficiency reasons, we will restrict $s\leq n^\rho$.

We also compare Assumption~\ref{assumption:ann_sparse} to two popular assumptions for nearest neighbor search: constant expansion ~\cite{kr02,bkl06} and doubling dimension \cite{gkl03,kl04_icalp,kl04_soda}. The constant (local) expansion states that if $|B(v, r)\cap U|\geq \log n$, then $|B(v, 2r)\cap U|\leq c_{\rm exp}\cdot |B(v, r)\cap U|$ for a constant $c_{\rm exp}$. The query time of nearest neighbor search data structures under the constant expansion assumption usually has a polylogarithmic dependence on $n$ but a large polynomial dependence on $c_{\rm exp}$. On the other hand, we only need $|B(v, cr)\cap U|\leq n^\rho$. If we let $|B(v, r)\cap U|=\log n$, then applying constant expansion $\lceil\log_2 c\rceil$ times gives $|B(v, cr)\cap U|\leq c_{\rm exp}^{\lceil \log_2 c\rceil}\log n\leq c_{\rm exp}\cdot c^{\log_2 c_{\rm exp}}\log n$. Thus, for any $c$ with $c_{\rm exp}\cdot c^{\log_2 c_{\rm exp}}\log n\leq n^\rho$ (for constant $c_{\rm exp}$ this allows $c=n^{\Omega(\rho)}$), constant expansion implies our assumption for every query $v$ with $|B(v, r)\cap U|=\log n$, and we show this implication only for such queries. In general the two assumptions are incomparable: $n$ points of diameter below $r$ (for example on a fine grid) have constant expansion and bounded doubling dimension, but they violate Assumption~\ref{assumption:ann_sparse} for every $s<n-1$. Moreover, we note that constant expansion restricts the growth across different distances, while our assumption only requires a relatively sparse neighborhood at the final level. A more robust version of constant expansion is the notion of doubling dimension, which asks how many balls of radius $r$ are needed to cover a ball of radius $2r$. The (global) expansion rate $c_{\rm exp}$ of~\cite{kr02} implies doubling dimension $O(\log c_{\rm exp})$~\cite{gkl03}, but not conversely~\cite{kl04_soda}, so bounded doubling dimension is a strictly weaker (more general) assumption than constant expansion, and it provides a bound of the form $|B(v, cr)\cap U|\leq \Delta^{{\rm dim}(X)}$, where $\Delta\geq 2$ is the aspect ratio of $U$ and $X$ is the metric space. If $X=\R^d$ with any norm, then ${\rm dim}(X)=\Theta(d)$, and therefore for this bound to be meaningful, one must have $d=o(\log n)$. Moreover, data structures with a doubling dimension assumption often have their preprocessing time and space exponential in ${\rm dim}(X)$.

As a final note, we would like to point out that when working with $U\subseteq \{0,1\}^d$ and when the norm is $\|\cdot \|_1$, i.e., the Hamming ANN problem, Assumption~\ref{assumption:ann_sparse} automatically holds with $s=n^\rho$ for $3d\leq n^\alpha$ and $\alpha\leq \frac{\rho}{2cr}$. In particular for Hamming LSH, $\rho=O(1/c)$~\cite{im98} and thus we only need $3d\leq n^{\Theta(1/(c^2r))}$. Our result will provide a Hamming LSH that is robust against adaptive adversaries, and in particular against the attack of~\cite{kms24} on datasets that satisfy Assumption~\ref{assumption:ann_sparse}.

Next, we consider the problem of adaptively updating the regression problem, and outputting the regression \emph{solution vector} whenever queried.

\begin{assumption}\label{assumption:regression}
Let $U\in \R^{n\times d}$ be the design matrix and $b\in \R^n$ be the response vector. Let $\{v_1,\ldots,v_T\}$ be a sequence of (possibly) adaptive updates to the problem in one of two forms: (1)  Update $U$: $v_t\in \R^{n\times d}$ and $U$ is updated via $U\leftarrow U+v_t$; (2) Update $b$: $v_t\in \R^n$ and $b$ is updated via $b\leftarrow b+v_t$.
We will use $(U_t, b_t)$ to denote the pair after being updated by $v_t$. The goal is to design a data structure such that for any $t\in [T]$, it outputs a $(1+\alpha)$-approximate solution $x_t\in \R^d$ for which $\|U_tx_t - b_t\|_2 \leq (1+\alpha)\cdot \min_{x\in \R^d} \|U_t x- b_t\|_2$ holds with high probability. For all $t\in [T]$, we assume the condition number of $U_t$ defined as $\kappa(U_t):=\frac{\sigma_{\max}(U_t)}{\sigma_{\min}(U_t)}$, is upper bounded by $\kappa$. We also assume that $T=\poly(n)$, and that for a constant $\gamma>0$, the entries of $U_t$ and $b_t$ lie in $[-n^\gamma,n^\gamma]$ for all $0\leq t\leq T$, where $(U_0,b_0)=(U,b)$.
\end{assumption}
The problem and assumption can be succinctly described as follows: an adversary could adaptively perturb entries of $U$ and $b$, and our goal is to design a data structure that outputs the regression solution in the presence of these perturbations. In addition, we assume the perturbations are bounded, in the sense that they keep all entries polynomially bounded and do not change the conditioning of the design matrix by too much. For this problem, note that one could na{\"i}vely store the matrix $U$ and solve the regression problem exactly. In this case, as the algorithm is deterministic, it is automatically adaptive. However, this approach would require $\Omega(nd)$ space and $O(nd^{\omega-1})$ time, which is both space- and time-inefficient for large $n$. Alternatively, one could prepare $T$ independent sketching matrices $S_1,\ldots,S_T$, one for each query, and store $S_t U, S_t b$ for all $t\in [T]$. When receiving an update to $U$ or $b$, one can simply update the corresponding entries of $S_t U$and $S_t b$, as sketching matrices are linear. This approach requires $\Omega(T\cdot \poly(d))$ space, a preprocessing time of $O(T\cdot (nd+\poly(d)))$, and a query time of $\poly(d)$. While the query time is much more efficient for $n\gg d$, the space and preprocessing time are prohibitive given that $T$ is large. Hence, our goal is to design algorithms that use $o(T\cdot \poly(d))$ space, have a sublinear in $T$ dependence in the  preprocessing time, and have a query time of $\poly(d)$.

\subsection{Main Results}
We state the main results for ANN and adaptive regression in this section.
\subsubsection{Adaptive ANN for Sparse Neighborhoods}
Our most general result for ANN is as follows:

\begin{theorem}
\label{thm:main_general}
Let $U\subseteq \R^d$ be an $n$-point dataset and let $f_v: (\R^d)^n\rightarrow (\R^d)^s$ be the predicate function. Let $Y_U$ be the set of yes instances. By default $Y_U=\{v: f_v(U)\neq\emptyset\}$, and for $(c,r)$-ANN, $Y_U=\{v: B(v,r)\cap U\neq\emptyset\}$. Let ${\cal A}$ be an oblivious algorithm, i.e., for a fixed query $v\in \R^d$, with probability at least $1-\delta$, where $\delta\leq 1/10$, ${\cal A}$ returns a non-empty subset of $f_v(U)$ if $v\in Y_U$, and with probability $1$ every point ${\cal A}$ returns lies in $f_v(U)$ (${\cal A}$ may return nothing). Suppose also that, given $v$ and $u\in U$, one can test whether $u\in f_v(U)$ in time $O({\cal T}_{\rm query})$, as for $(c,r)$-ANN by computing $\|u-v\|$. Moreover, ${\cal A}$ has preprocessing time ${\cal T}_{\rm prep}$, space usage ${\cal S}_{\rm space}$ and query time ${\cal T}_{\rm query}$. Then there exists an adaptive algorithm $\wt {\cal A}$ such that given a sequence of adaptive queries $\{v_1,\ldots,v_T\}$ satisfying Assumption~\ref{assumption:general}: (1) Preprocesses $U$ in time $\wt O(\sqrt T\cdot s)\cdot {\cal T}_{\rm prep}$; (2) Uses space $\wt O(\sqrt T\cdot s)\cdot {\cal S}_{\rm space}$; (3) For all $t\in [T]$, given query $v_t$, it returns a point of $f_{v_t}(U)$ if $v_t\in Y_U$, and otherwise nothing or a point of $f_{v_t}(U)$, in time $\wt O(s)\cdot {\cal T}_{\rm query}$ with probability at least $1-\delta_{\rm fail}-1/\poly(n)$, for any $\delta_{\rm fail}\in(0,1)$, where $\wt O(\cdot)$ hides $\log(1/\delta_{\rm fail})$ factors. In particular, the amortized cost per query is $\wt O(s/\sqrt T)\cdot {\cal T}_{\rm prep} + \wt O(s)\cdot {\cal T}_{\rm query}$.
\end{theorem}

Let us interpret Theorem~\ref{thm:main_general}. It states that as long as $s=o(\sqrt T)$, then we can turn an oblivious data structure into an adaptive data structure using fewer than $T$ independent copies of the oblivious data structure. Of course, this comes at a cost of a slightly worse query time by a factor of $\wt O(s)$. However, if we consider the amortized cost per query for using $T$ data structures, the cost is dominated by ${\cal T}_{\rm prep}$, as the algorithm needs to prepare a fresh data structure for each query. In contrast, Theorem~\ref{thm:main_general} has an amortized cost of $\wt O(s/\sqrt{T})\cdot {\cal T}_{\rm prep}+\wt O(s)\cdot {\cal T}_{\rm query}$ per query. Typically, ${\cal T}_{\rm query}$ is much smaller than ${\cal T}_{\rm prep}$; for example, for LSH, ${\cal T}_{\rm prep}=n^{1+\rho}d$ and ${\cal T}_{\rm query}=n^\rho d$, so it is much more important to obtain a reduction on the number of data structures. LSH under different norms meets the hypotheses of Theorem~\ref{thm:main_general}. Its success probability is a constant, which $O(1)$ independent repetitions raise to $9/10$. It reports only candidates that it has checked to be within distance $cr$ of the query, and membership in $f_v(U)=B(v,cr)\cap U$ is tested by one distance computation. So Theorem~\ref{thm:main_lsh} below, which uses fewer than $T$ data structures, is the case $f_v(U)=B(v,cr)\cap U$ with the ANN yes instances, and our analysis of the reduction for $(c,r)$-ANN proves it directly.

\begin{theorem}
\label{thm:main_lsh}
Let $U\subseteq \R^d$ be an $n$-point dataset satisfying Assumption~\ref{assumption:ann_sparse}. There exists an adaptive algorithm $\wt {\cal A}$ such that given a sequence of adaptive queries $\{v_1,\ldots,v_T\}$, it  solves the $(c,r)$-ANN problem (Definition~\ref{def:cr_ann}) and (1) Preprocesses $U$ in time $\wt O(\sqrt T\cdot s\cdot n^{1+\rho}d)$; (2) Uses space $\wt O(\sqrt T\cdot s\cdot (n^{1+\rho}+n^\rho d)+nd)$; (3) For all $t\in [T]$, given query $v_t$, it returns a point in $B(v_t, cr)\cap U$ if $B(v_t, r)\cap U\neq \emptyset$ in time $\wt O(s\cdot n^\rho d)$, with probability at least $1-\delta$.
\end{theorem}

A generic LSH data structure template restricts the number of points it looks at per query to $n^\rho$, so if $s\leq n^\rho$, we obtain the following results:
\begin{corollary}
Let $U\subseteq \R^d$ satisfy Assumption~\ref{assumption:ann_sparse} with $s\leq n^\rho$. Then there exists an adaptive algorithm $\wt {\cal A}$ such that given a sequence of adaptive queries $\{v_1,\ldots,v_T\}$, the data structure (1) Preprocesses $U$ in time $\wt O(\sqrt T\cdot n^{1+O(\rho)}d)$;
   (2) Uses space $\wt O(\sqrt T\cdot n^{O(\rho)}\cdot (n+d)+nd)$;
    (3) For all $t\in [T]$, given query $v_t$, it returns a point in $B(v_t, cr)\cap U$ if $B(v_t, r)\cap U\neq \emptyset$ in time $\wt O(n^{O(\rho)} d)$, with probability at least $1-\delta$.
\end{corollary}

Utilizing these adaptive data structures, we obtain improved runtime for online weighted matching with adversarial arrival. We also apply our framework to terminal embeddings. A $(1\pm\epsilon)$ terminal embedding into $\R^{O(\epsilon^{-2}\log n)}$ can be evaluated correctly on every query that satisfies a sparsity condition, among $T$ adaptive queries that need not all satisfy it, with $\wt O(\sqrt T\cdot s)$ LSH copies per distance scale, and for fixed $\epsilon$ the query time depends on $n$ only through a factor $n^\rho$ and polylogarithmic factors (Theorem~\ref{thm:terminal} and Corollary~\ref{cor:terminal_lsh}). This is a different trade-off from~\cite{cn21}, not an improvement of its space. In addition, matching requires deleting each returned point from the data structures. In $\ell_2$, and in $\ell_p$ for every $p\in(0,2]$, we store the sketches (for $\ell_p$, random features) of all points at preprocessing. This lets us delete a point from all copies in $\wt O(\sqrt T\cdot s\cdot n^\rho)$ time, with $\rho=1/c^p+o(1)$ for $\ell_p$ and constant $c$ (Theorems~\ref{thm:batch} and~\ref{thm:lp_delete}). In Table~\ref{tab:ann}, we compare our result with prior works that use $d$ or $T$ copies of the LSH's.

\begin{table*}[!ht]
    \centering
    \begin{tabular}{|l|l|l|l|l|}
    \hline
       {\bf Method}  & {\bf Space} & {\bf Amortized Prep Time} & {\bf Query Time} & {\bf Update Time}  \\ \hline
        $T$ copies & $Tn^{1+\rho}+nd$ & $n^{1+\rho} d$ & $n^\rho d$ & $Tn^\rho d$ \\ \hline 
        $d$ copies & $n^{1+\rho}d$ & $\frac{d}{T}n^{1+\rho}d$ & $n^\rho d$ & $n^\rho d^2$ \\ \hline
        {\bf This paper} & $\sqrt T sn^{1+\rho}+nd$ & $\frac{s}{\sqrt T}n^{1+\rho}d$ & $sn^\rho d$ & $\sqrt Tsn^\rho d$\\ \hline 
    \end{tabular}
    \caption{Specifications of data structures given a sequence of $T$ adaptive queries and updates for ANN problem. We ignore $\wt O(\cdot)$ notation for clarity. We use $s$ to denote the parameter for Assumption~\ref{assumption:ann_sparse}. Update time is the time to delete a point from all copies when its hash values are recomputed. For $\ell_2$ and $\ell_p$, $p\in(0,2]$, with stored sketches or features (Theorems~\ref{thm:batch} and~\ref{thm:lp_delete}), our deletion time is $\sqrt T s n^\rho$, with $\rho=1/c^p+o(1)$ for $\ell_p$. The Space column counts $n^{1+\rho}$ words per LSH copy. For families whose hash functions take $d$ words, such as~\cite{diim04}, each copy also stores $\wt O(n^\rho d)$ words of hash functions, which this count absorbs when $d\leq n$.}
    \label{tab:ann}
\end{table*}

\subsubsection{Adaptive Regression for Well-Conditioned Instances}
For adaptive regression, we can use fewer than $T$ copies of an oblivious data structures if the condition number upper bound $\kappa$ is small:

\begin{theorem}\label{thm:reg_well_conditioned}
Let $U\in \R^{n\times d}, b\in \R^n$ and $\{v_1,\ldots,v_T \}$ be a sequence of adaptive updates to $(U, b)$, satisfying Assumption~\ref{assumption:regression}. Suppose moreover that for all $t\in[T]$, $\|x^*_t\|_\infty\leq n^\gamma$ and $\OPT_t\geq n^{-\gamma}$, where $x^*_t:=\arg\min_{x\in\R^d}\|U_tx-b_t\|_2$, $\OPT_t:=\min_{x\in\R^d}\|U_tx-b_t\|_2$ and $\gamma$ is the constant of Assumption~\ref{assumption:regression}. Then, there exists an adaptive algorithm that
(1) Preprocesses $(U, b)$ in time $\wt O(\sqrt{Td}\cdot (\nnz(U)+\nnz(b)+d^3\kappa/\alpha+d^{2}\kappa^2/\alpha^2))$;
   (2) Uses space $\wt O(\sqrt{T}\cdot (d^{2.5}\kappa^2/\alpha^2+d^3\kappa/\alpha))$;
    (3) For all $t\in [T]$, it updates $(U_{t-1}, b_{t-1})$\footnote{We use $(U_0, b_0)$ to denote $(U, b)$, the initial design matrix and response vector.} to $(U_t, b_t)$ in time $\wt O(\sqrt{Td}\cdot (\nnz(v_t)+d^3\kappa/\alpha+d^2\kappa^2/\alpha^2))$;
    (4) For all $t\in [T]$, given update $v_t$, it returns a solution $x_t$ that is a $(1+\alpha)$-approximation to the regression problem using $(U_t, b_t)$ in time $\wt O(d^{\omega+1}\kappa^2/\alpha^2)$, with probability at least $1-\delta$.
\end{theorem}

Theorem~\ref{thm:reg_well_conditioned} offers an algorithm with extremely efficient query time. When the condition number bound $\kappa$ is small, it provides a space bound of $\sqrt T\cdot \poly(d)$, both sublinear in $T$ and with no polynomial dependence on $n$. The preprocessing time outperforms the simple algorithm which generates $T$ sketches when $d\ll T$ and $\kappa$ is small. While for ANN, one could prepare $\wt O(d)$ data structures during preprocessing to prepare for all possible queries, and at query time just sample a small number of these, this is not possible for regression. Indeed, for regression one would need to prepare $\wt O(nd)$ such data structures, given the number of possible design matrices, which would be prohibitive. We again compare our result with prior approaches that use $nd$ or $T$ copies of sketches, in Table~\ref{tab:reg}.

\begin{table*}[!ht]
    \centering
    \small
    \begin{tabular}{|l|l|l|l|l|}
    \hline
       {\bf Method}  & {\bf Space} & {\bf Amortized Prep Time} & {\bf Query Time} & {\bf Update Time}  \\ \hline
        $T$ copies & $Td^2/\alpha^2$ & $\nnz(U, b)+d^3+d^2/\alpha^2$ & $d^{\omega+1}/\alpha^2$ & $\begin{array}{@{}l@{}}T(\nnz(v_t)+d^3\\ \quad+d^2/\alpha^2)\end{array}$ \\ \hline
        $nd$ copies & $nd^3/\alpha^2$ & $\begin{array}{@{}l@{}}\frac{nd}{T}(\nnz(U, b)+d^3\\ \quad+d^2/\alpha^2)\end{array}$ & $d^{\omega+1}/\alpha^2$ & $\begin{array}{@{}l@{}}nd(\nnz(v_t)+d^3\\ \quad+d^2/\alpha^2)\end{array}$ \\ \hline
        {\bf This paper} & $\begin{array}{@{}l@{}}\sqrt T(d^{2.5}\kappa^2/\alpha^2\\ \quad+d^3\kappa/\alpha)\end{array}$ & $\begin{array}{@{}l@{}}\sqrt{\frac{d}{T}}(\nnz(U, b)+d^3\kappa/\alpha\\ \quad+d^2\kappa^2/\alpha^2)\end{array}$ & $d^{\omega+1}\kappa^2/\alpha^2$ & $\begin{array}{@{}l@{}}\sqrt{Td}(\nnz(v_t)+d^3\kappa/\alpha\\ \quad+d^2\kappa^2/\alpha^2)\end{array}$\\ \hline
    \end{tabular}
    \caption{Specifications of data structures given a sequence of $T$ adaptive queries and updates for regression. We ignore $\wt O(\cdot)$ notation for clarity. We use $\nnz(U, b)$ as a shorthand for $\nnz(U)+\nnz(b)$, and let $\kappa$ be the parameter for Assumption~\ref{assumption:regression}.}
    \label{tab:reg}
\end{table*}

%

We next study the model in~\cite{csw+23}, where only the response vector $b$ is allowed to be updated in at most $s$ positions. In~\cite{csw+23}, it is shown that one can corporate the techniques of~\cite{bkm+21} to output the \emph{cost} of the regression problems, but a key open question was whether one could output the actual solution vector to the regression problems. We resolve this question by utilizing the tools we developed in Theorem~\ref{thm:reg_well_conditioned} in conjunction with a preconditioner for $U$ to remove the dependence on $\kappa$. Below, we will use $\nnz(U)$ to denote the number of nonzero entries in $U$.

\begin{theorem}\label{thm:reg_sparse}
Let $U\in \R^{n\times d}$ have full column rank and let $b\in \R^n$. Let $\{v_1,v_2,\ldots,\}\subset \R^n$ be a sequence of adaptive updates to $b$, such that for all $t$, $\|v_t\|_0\leq s$, and let $b_t$ be the response vector after the $t$-th update, with $b_0=b$. Let $T=\poly(n)$ be a batch size parameter. Suppose that for a constant $\gamma>0$, the entries of $U$ and of every $b_t$ lie in $[-n^\gamma,n^\gamma]$, and $\min_{x\in\R^d}\|Ux-b_t\|_2\geq n^{-\gamma}$ for all $t$. There exists an adaptive algorithm that
(1) It has amortized update time $\wt O(\sqrt{d/T}\cdot (\nnz(U)+\nnz(b)+d^{3}/\alpha+d^{\omega}/\alpha^2)+\sqrt{Td}\cdot (s+d^{2}/\alpha^2+d^{2.5}/\alpha))$;
(2) It outputs a solution $x_t$ in time $\wt O(d^2)$, and in each batch of $T$ updates, with probability at least $1-\delta$, $x_t$ is a $(1+\alpha)$-approximate solution in every round of the batch;
(3) Before the first update, it computes a preconditioner once, in time $\wt O(\nnz(U)+d^3/\delta+d^{\omega})$.
\end{theorem}
The main advantage of Theorem~\ref{thm:reg_sparse} over Theorem~\ref{thm:reg_well_conditioned} is that the quadratic dependence on the condition number can be removed via choosing a proper preconditioner in this setting, since $U$ is not changing. We also have extremely fast query time, as the solution vector can be quickly updated via a matrix-vector product instead of solving the regression problem from scratch. 

\ifdefined\isarxiv
We summarize the comparison in Table~\ref{tab:label}, note that here for a net argument, we only need to union bound over the net in $\R^n$. Note that our approach offers a significant speedup in the regime $d\leq T\leq n$.

\begin{table*}[!ht]
    \centering
    \begin{tabular}{|l|l|l|l|l|}
    \hline
       {\bf Method}  & {\bf Space} & {\bf Amortized Prep Time} & {\bf Query Time} & {\bf Update Time}  \\ \hline
        $T$ copies & $Td^2/\alpha^2$ & $\nnz(U, b)+d^3+d^\omega/\alpha^2$ & $d^2$ & $T(s+d^2/\alpha^2)$ \\ \hline 
        $n$ copies & $nd^2/\alpha^2$ & $\frac{n}{T}(\nnz(U, b)+d^3+d^\omega/\alpha^2)$ & $d^2$ & $n(s+d^2/\alpha^2)$ \\ \hline
        {\bf This paper} & $\begin{array}{@{}l@{}}\sqrt T(d^{2.5}/\alpha^2\\ \quad+d^3/\alpha)\end{array}$ & $\begin{array}{@{}l@{}}\sqrt{\frac{d}{T}}(\nnz(U, b)+d^3/\alpha\\ \quad+d^\omega/\alpha^2)\end{array}$ & $d^2$ & $\begin{array}{@{}l@{}}\sqrt{Td}(s+d^2/\alpha^2\\ \quad+d^{2.5}/\alpha)\end{array}$\\ \hline
    \end{tabular}
    \caption{Specifications of data structures given a sequence of $T$ adaptive queries and updates for regression under sparse label shifts. We ignore $\wt O(\cdot)$ notation for clarity. We use $\nnz(U, b)$ as a shorthand for $\nnz(U)+\nnz(b)$ and we assume $\|v_t\|_0\leq s$ for all $t\in [T]$. In addition, our algorithm computes a preconditioner once, before the first update, in time $\nnz(U)+d^3/\delta+d^\omega$.}
    \label{tab:label}
\end{table*}
\fi

When the condition number $\kappa$ is as large as $\poly(n)$ and $U$ is dynamically changing, Theorem~\ref{thm:reg_well_conditioned} no longer gives efficient space and preprocessing time. To address this problem, we develop an algorithm with only \emph{logarithmic} dependence on the condition number $\kappa$, based on the bounded computation path technique~\cite{bjwy20}.
\vspace{-0.5mm}
\begin{theorem}\label{thm:log_kappa}
Let $U\in \R^{n\times d}, b\in \R^n$ and $\{v_1,\ldots,v_T \}$ be a sequence of adaptive updates to $(U, b)$, satisfying Assumption~\ref{assumption:regression}. Suppose moreover that for all $t\in[T]$, the entries of $U_t$ and $b_t$ are integers in $[-n^\gamma,n^\gamma]$, and $\OPT_t:=\min_{x\in\R^d}\|U_tx-b_t\|_2$ is either $0$ or at least $n^{-\gamma}$, where $\gamma$ is the constant of Assumption~\ref{assumption:regression}. Let $\mathcal{P}$ be the set of output sequences in which every output is a point of a fixed finite grid or the exact solution of the current instance, as defined in the proof of Lemma~\ref{lem:path_adaptive}. We have $|\mathcal{P}|\le (n\kappa/\alpha)^{O(dT)}$. There exists an adaptive algorithm that
(1) Preprocesses $(U, b)$ in time $\wt O(nd^{\omega-2}(d + \log |{\cal P}| + \log \frac{1}{\delta})/\alpha^2)$;
   (2) Uses space $\wt O(d(d + \log |{\cal P}| + \log \frac{1}{\delta})/\alpha^2)$; 
   (3) For all $t\in [T]$, it updates $(U_{t-1}, b_{t-1})$ to $(U_t, b_t)$ in time $\wt O(\nnz(v_t)\cdot(d + \log |{\cal P}| + \log \frac{1}{\delta})/\alpha^2)$;
    (4) For all $t\in [T]$, it returns a $(1+\alpha)$-approximate solution $x_t$ to the regression problem using $(U_t, b_t)$ in time $\wt O({d^{\omega-1}(d + \log |{\cal P}| + \log \frac{1}{\delta})}/{\alpha^2})$, with probability at least $1-\delta$.
\end{theorem}
\vspace{-0.5mm}
Compared to Theorem~\ref{thm:reg_well_conditioned}, Theorem~\ref{thm:log_kappa} offers a better dependence on the condition number $\kappa$ at the expense of a linear dependence on the length of update sequence $T$, since $\log|{\cal P}|=O(dT\log(n\kappa/\alpha))$. The algorithm rounds each output to a fixed finite grid, unless the sketch shows that the system is exactly consistent, in which case it outputs the exact solution. The analysis uses only the number of possible output sequences, so the same argument needs a sketch with fewer rows whenever the outputs are known to come from a smaller set, for example when the solutions are relatively stable. This is similar to past work where the complexity depends on a stability parameter~\cite{hkm+22}.



\section{Related Work}

\noindent\textbf{Differential Privacy.} Differential privacy is a central concept in data privacy, introduced in~\cite{dmns06}. The main idea of differential privacy is that when the inputs to the algorithm are close to each other, it would be almost impossible to differentiate the outputs. Since its introduction, differential privacy has seen rich applications in general machine learning~\cite{cm08,wm10,je19,tf20}, deep neural network~\cite{acg+16,bps19}, computer vision~\cite{zycw20,lwaff21,tzxl19}, natural language processing~\cite{ydw+21,wk18}, federated learning~\cite{syy+22,swyz23} and large language model~\cite{ynb+22,gsy23,lssz24_dp,gls+25,ncmw25}. Designing data structures with differential privacy guarantees is crucial, as it automatically ensures the privacy of any downstream tasks~\cite{cem+22,dll23,lnv23,cgk+23,aimn23,lnv24}. It is of particular interest to design differentially private data structure in the \emph{function release model}, where the quality of the output won't degrade as the data structure processes more queries~\cite{hrw13,hr14,wjf+16,ar17,cs21,wnm23,blm+24,lhr+24,hll+24,kls+25,lls+25_je}.

\vspace{2mm}
\noindent\textbf{Adaptive Data Structure.} 
In recent years, data structures have been integrated into iterative processes to speed up the algorithm. This has been the foundation for various recent breakthroughs in fast convex optimization~\cite{cls19,lsz19,blss20,bgjllps22,jswz21,qszz23}. These data structures possess the ability to answer adaptive queries, that could depend on previous outputs from the data structure, with high success probability. For streaming problems, the adversary could also be adaptive, in the sense that it would feed the streaming algorithm with updates, after observing the prior decisions of the algorithm~\cite{bjwy20,wz21,abj+22,fw23,wz24,fjw24,glw+24,gls+25}. In this work, we focus in particular on the adaptive data structures based on differential privacy~\cite{hkm+22,bkm+21,syyz23_dp,csw+23}.
\section{Technical Overview}
We divide the technical overview into two parts. For adaptive ANN, we demonstrate a novel framework based on differentially private selection over a sparse vector. For adaptive regression, we show how to upgrade the private median framework of~\cite{bkm+21} to output the solution vector, and how to obtain utility guarantees through a novel use of a per-coordinate form of the $\ell_\infty$ guarantees of the sketch-and-solve framework \cite{psw17}.

\subsection{Adaptive ANN via Differentially Private Selection}

\paragraph{Existing Results.} Before providing an overview of our techniques, we start by examining existing results for adaptive approximate nearest neighbor search data structures. The first candidates are of course deterministic data structures. Deterministic approximate nearest neighbor data structures have been a central topic of study since the 1970s \cite{dbc+08,clrs22}. 
However, without any structural assumptions on the query or dataset, these data structures suffer from the curse of dimensionality, i.e., their preprocessing time and space usage scale exponentially in $d$, making them infeasible for slightly large dimensions. If one is willing to make strong structural assumptions on the dataset, e.g., given any query point $v$, the number of points in $B(v, 2r)$ only grows by a constant factor compared to the number of points in $B(v, r)$ \cite{c97,kr02,bkl06,kl04_soda}, 
then it is possible to design a data structure with polynomial preprocessing time and space (note that the dependence on the growth constant or doubling dimension is large, but still polynomial), and logarithmic query time. However, these assumptions are strong, as they greatly limit the potential geometric structure of the dataset.

The celebrated work of Indyk and Motwani \cite{im98} 
shows that if one relaxes the problem to allow for answering $c$-approximate near neighbor queries instead, then one can achieve (slightly) super-linear preprocessing time and space, and sublinear query time. These data structures are inherently random, as they rely on partitioning the space using random directions. As this randomness is determined during the preprocessing phase, these data structures are not robust against an adaptive adversary. In fact, for Hamming approximate near neighbor search,~\cite{kms24} provides an efficient attack that can always force the data structure to output a false negative if there exists at least one point that is isolated from other points.

Since the issue of false negatives is most prevalent for Monte Carlo data structures, one might consider to use a class of \emph{Las Vegas} ANN data structures whose running times are random variables, but are guaranteed to output a correct answer with no false negatives~\cite{kor98,p16,a17,w22}. Unfortunately, the space and runtime analysis of these data structures are performed based on the assumption that the query sequence is \emph{oblivious}. An adaptive adversary could design a sequence of queries with a much higher average response time than an oblivious one.

\noindent{\bf Reduction To Differentially Private Selection.} We start by introducing the differentially private selection problem. Given $n$ categories, a collection of $s$ binary vectors over $\{0,1\}^n$ denoted by $b_{(1)},\ldots,b_{(s)}$, the goal is to find the category $j^*=\arg\max_{j\in [n]} \sum_{i=1}^s b_{(i),j}$, i.e., the most common category among all vectors. 
This can be achieved via the following mechanism: 1).\ Compute the overall count vector $B=\sum_{i=1}^s b_{(i)}$;
2).\ Add independent Laplace noise ${\rm Lap}(2/\epsilon)$ to each of the counts; 3).\ Report the index with the largest noisy count. 
The privacy parameter $\epsilon$ does not scale with either the number $n$ of categories or the number $s$ of vectors, as it only outputs a single index.

We now show how to frame approximate near neighbor search as a differentially private selection problem. We create a category for each point $u\in U$, and for each data structure, and we ask it to output \emph{all} the near neighbors it finds instead of a single one of them. This creates an indicator vector for the points: if the data structure finds $u_i$, then the corresponding vector has its $i$-th entry equal to $1$.
We can then apply the differentially private selection procedure for these indicator vectors, and output the point with maximum noisy count. To see this mechanism indeed protects the privacy of the internal randomness of the data structures, note that fixing the dataset $U$ and the adaptive query point $v$, the indicator vector is determined by the random strings of these data structures. Thus, the mechanism is $(\epsilon,0)$-DP, and we can apply the advanced composition theorem of differential privacy to reduce the number of data structures from $T$ to $\wt O(\sqrt T)$. 

This is, however, not sufficient to provide a utility guarantee, as the Laplace noise could have large magnitude, thus making all counts too noisy. While it is fortunate that we can set $\epsilon=O(1)$, and with high probability, the magnitude of the noise is $\Theta(\log n)$, this might still be too large to be useful. To see how this issue can be resolved, let us consider the case that the query point $v$ only has one approximate near neighbor, i.e., there exists only one $u$ for which $\|v-u\|\leq r$ and $\|v-u\|\leq cr$. In this case, the data structure can only output $u$ (conditioned on the data structure succeeding). If we were to sample $\omega(\log n)$ data structures and compute the count vector, we can guarantee that the entry corresponding to the near-neighbor has larger magnitude than the noise. Thus, as long as a constant fraction of the data structures succeed, we can ensure we output the correct point with high probability.

Note that our above argument assumes that a constant fraction of data structures succeed, which \emph{is not guaranteed when facing an adaptive adversary}! Fortunately, we can circumvent this problem by proving the algorithm is differentially private first; subsequently, we can apply the generalization property of DP to ensure that indeed, a constant fraction of data structures succeed. Note the two-stage nature of this argument: we can only argue about the utility if the privacy is preserved. There are two main issues left: (1) the assumption that the query has only one near-neighbor is too restrictive, and (2) the differentially private selection procedure takes $O(n)$ time to respond to each query, making the data structure very inefficient. For the first issue, we note that we can extend this to the setting when $v$ has only $s$ approximate near-neighbors. In this setting, we can sample $\omega(s\cdot \log n)$ data structures for each query. By the pigeonhole principle, there must exist at least one point whose count has magnitude larger than the noise, and we obtain the desired utility guarantee. Since the LSH data structures only output $n^\rho$ points for each query, we could pick $s=O(n^\rho)$, i.e., allow the query to have $O(n^\rho)$ near-neighbors. We can alternatively translate this assumption into a structural property on the dataset: as long as the ball centered at each point with radius $cr$ intersects at most $s$ other balls, then we are guaranteed that the query has at most $s+1$ near-neighbors. We note that this assumption is automatically satisfied for Hamming nearest-neighbor search when $3d\leq n^{\rho/(2cr)}$. This range does not include the parameters of the attack of~\cite{kms24}, where the assumption has to be imposed on the dataset.

The second issue is algorithmic: if we na\"{i}vely implement the differentially private selection mechanism for each query, we will have to generate noise for each count which inevitably takes $\Omega(n)$ time. On the other hand, the count vector $B$ is $s$-sparse, and these non-zero entries have  magnitude larger than the noise. This problem has been studied before in the context of publishing a private database for sparse data~\cite{cpst12}; however, existing solutions either require modifying the problem definition so that the privacy becomes challenging to prove, or converting a Monte Carlo algorithm into a Las Vegas one, with only expected runtime guarantees. We develop a novel algorithm to resolve these issues termed as the \emph{sparse argmax mechanism}:  given an $s$-sparse vector, 
(1) Adding $s$ exponential noises ${\rm Exp}(2/\epsilon)$ to the support of the sparse vector; (2) Generating the $X$ from the $n$-th order statistics distribution of ${\rm Exp}(2/\epsilon)$; (3) Flip a biased coin with head probability $\frac{s}{n}$, if head, generate $s-1$ i.i.d. exponential noises until none of them exceed $X$, add them including $X$ to the support, and output the maximum entry index; (4) If tail, generate $s$ i.i.d. exponential noises until none of them exceed $X$, add these noises to the support, assign $X$ a random index in the $n-s$ non-support, output the maximum index associated with the noisy entries and $X$. We prove that with high probability, generating these noises can be done $O(s\log n)$ time, and the output distribution with sparse noises is the same as the generating $n$ i.i.d. exponential noises.

\subsection{Adaptive Regression via Differentially Private Median and \texorpdfstring{$\ell_\infty$}{} Guarantee}

\paragraph{Existing Results.} We first recall the standard setup of solving the over-constrained $\ell_2$ regression problem. Given a design matrix $U\in \R^{n\times d}$ with $n\gg d$ and a response vector $b$, the goal is to compute $x^*:=\arg\min_{x\in \R^d} \|Ux-b\|_2^2$. In the static setting, $x^*$ can be computed via the normal equations: $x^*=(U^\top U)^\dagger U^\top b$, but this is usually too expensive to be directly solved. The sketch-and-solve paradigm (see, e.g., \cite{w14} for a survey)  provides a wide array of algorithmic tools to speed up this process. In particular, one can pick a random matrix $S\in \R^{r\times n}$ from a certain structured family of random matrices, so that $r$ is a small polynomial in $d$, and $S$ can be quickly applied to $U$. One can then solve the sketched $\ell_2$ regression problem $\min_{x\in \R^d} \|SUx-Sb\|_2^2$, for which the optimal solution is a good approximation to $x^*$. The $\ell_2$ regression problem has also been extensively in the streaming~\cite{cw09,bhm+21} and dynamic~\cite{jpw23,csw+23} models, where either the design matrix $U$ or the response vector $b$ can be updated, and one has to produce a high quality approximate solution after each update. These works are either not robust to adaptive adversaries~\cite{cw09}, only support inserting or removing entire rows at once~\cite{bhm+21,jpw23}, or can only return the cost instead of the solution vector~\cite{csw+23}. In the realistic settings, the design matrix needs to tolerate perturbations to its entries due to the presence of noise or updated data. We consider the most general model where $U$ can be updated by perturbation $v_t\in \R^{n\times d}$ that can be adaptively chosen.

\noindent{\bf Coordinate-wise Private Median and the $\ell_\infty$ Guarantee.} One natural idea is to extend the private median framework of~\cite{bkm+21} to outputting an approximation to the solution vector. Our algorithm follows the generic template: prepare $k$ copies of sketching matrices $S_1,\ldots,S_k$ and preprocess $(U, b)$ as $(S_i U, S_i b)$ for all $i\in [k]$. During an adaptive update, we simply update the corresponding sketches of the design matrix and of the response vector. At query time, we sample $s=\wt O(\sqrt d)$ of our sketches and solve these sketched regression problems. Let $x_{(1)},\ldots,x_{(s)}$ be the returned solution vectors. We next need to perform a private aggregation on these vectors to craft our output.

The~\cite{bkm+21} approach for private aggregation is to compute the private median of these numbers. We first consider a natural extension: for each $i\in [d]$, we compute $g_i=\textsc{PMedian}((x_{(1)})_i,\ldots,$ $(x_{(s)})_i)$ and output $g=(g_1,\ldots,g_d)$. The privacy of this approach needs care, because the $d$ medians of a round all read the same sample of sketches, so changing one sketch can move all $d$ coordinates together. We therefore treat the $d$ coordinate medians of a round as one mechanism. Advanced composition over the $d$ coordinates makes a round private with a constant parameter when each median has privacy parameter $\wt\Theta(1/\sqrt d)$, which is why we sample $s=\wt O(\sqrt d)$ sketches. Amplification by subsampling then makes the round private with respect to all $k$ sketches, and advanced composition over the $T$ rounds gives the final guarantee. This comes at a price of requiring $\wt O(\sqrt{Td})$ sketches instead of the $\sqrt{T}$ data structures of~\cite{bkm+21}, but it still offers an improvement as long as $d\ll T$. The main challenge now lies in proving the \emph{utility} of our approach, where we need to show that $g$ is in fact a good enough solution to the regression problem. 

To quantify $\|Ug-b\|_2$, note that $\|Ug-b\|_2\leq \|Ux^*-b\|_2+\|U(x^*-g)\|_2$, where the second term can be bounded by $\|U(x^*-g)\|_2\leq \sigma_{\max}(U)\cdot \|x^*-g\|_2\leq \sigma_{\max}(U)\sqrt d \cdot\|x^*-g\|_\infty$. In other words, if we can get a good handle on $\|x^*-g\|_\infty$, then we can hopefully obtain a useful bound on the error. Since $g$ is the coordinate-wise private median of $x_{(1)},\ldots,x_{(s)}$, which are solutions of sketched $\ell_2$ regression problems, the standard sketching error guarantee only ensures that the \emph{forward error} is small, i.e., $\|Ux_{(i)}-b\|_2\leq (1+\alpha) \|Ux^*-b\|_2$ for any $i\in [s]$. For a \emph{backward error} type guarantee, i.e., a bound on $\|x_{(i)}-x^*\|_2$, one could convert directly from the forward error guarantee. If $U$ is reasonably well-conditioned, then closeness in forward error implies closeness in backward error. 

However, for the private median guarantee, it is important that each {\it entry} of $x_{(i)} - x^*$ is small, rather than just the bound that $\|x_{(i)}-x^*\|_2$ is small. While the ideal scenario would be $\|x_{(i)}-x^*\|_\infty \approx \frac{1}{\sqrt d} \|x_{(i)}-x^*\|_2$, this is generally not true for most sketching matrices. Fortunately, sketched solutions can satisfy a much stronger guarantee, studied as the $\ell_\infty$ guarantee in~\cite{psw17,syyz23}. We use it for each coordinate separately: for the Subsampled Randomized Hadamard Transform (SRHT) composed with a Count Sketch, we show that for each fixed coordinate $l\in[d]$, with probability $1-\beta$, $|(x_{(i)})_l-x^*_l| \leq \frac{\alpha}{\sqrt d}\cdot \|Ux^*-b\|_2\cdot \frac{1}{\sigma_{\min}(U)}$. This per-coordinate form is all that the private median needs. We apply the generalization property of differential privacy to one predicate per round and coordinate, which shows that most of the sampled sketches are accurate in that coordinate, and a union bound over the $dT$ predicates makes this hold for all coordinates and rounds at once. By properly choosing the parameters for the private median estimator, we can show that with high probability this bound also holds for every coordinate of $g$, up to the error of rounding each coordinate to the finite grid on which the private median runs, which our bounded-precision hypotheses make small. Hence, we have $\|U(x^*-g)\|_2 \leq  \sigma_{\max}(U)\sqrt{d}\cdot \|x^*-g\|_\infty
    \leq \alpha \kappa(U) \cdot\|Ux^*-b\|_2$ plus the rounding error.
To offset the blowup in condition number, we scale down $\alpha$ by a factor of $\kappa$, and thus establish the utility of the coordinate-wise private median mechanism. We find the private median to be a surprising application of sketching with the $\ell_{\infty}$ guarantee, which is a less studied guarantee in the sketching literature.

To further speed up the preprocessing and update time, we compose the SRHT sketch with Count Sketch matrices~\cite{ccf02} so that both operations take time $\nnz$ of the input plus a term polynomial in $d$, $\kappa$ and $1/\alpha$. In addition, each sketched design matrix has $\wt O(d\kappa^2/\alpha^2+d^{1.5}\kappa/\alpha)$ rows, so the sketches and the sketched design matrices can be stored in $\wt O(d^2\kappa^2/\alpha^2+d^{2.5}\kappa/\alpha)$ words of space. Since we use $\wt O(\sqrt{Td})$ independent sketches, the space usage is $\wt O(\sqrt{T}(d^{2.5}\kappa^2/\alpha^2+d^3\kappa/\alpha))$.

\section{Applications}
\label{sec:app}
In this section, we apply our adaptive ANN data structures to online weighted matching. We also apply our framework to evaluate a terminal embedding on adaptive queries that have few data points nearby. Let $\Tmat(a,b,c)$ to be the time complexity of multiplying an $a\times b$ matrix with a $b\times c$ matrix.

\subsection{Online Weighted Matching}
\label{sec:intro:matching}
Consider the following online weighted matching problem on bipartite graphs: we are given a set of left vertices $U\subseteq \R^d$ and we will receive a sequence of online right vertices $V:=\{v_1,\ldots,v_T\}\subseteq \R^d$ with $T\leq n$. The edge weight between a left vertex $u\in U$ and a right vertex $v\in V$ is defined to be $\frac{1}{\|u-v\|}$, where $\|\cdot \|=\|\cdot\|_2$ is the $\ell_2$ norm. Our goal is to design an algorithm that given a vertex $v_t$, make an immediate decision to match with a point $u_t\in U$ with edge weight $\frac{1}{\|u_t-v_t\|}$. The goal is to compute an online matching that maximizes the total weight $\sum_{t=1}^T \frac{1}{\|u_t-v_t\|}$. The online weighted matching problem has been widely studied in the edge-arrival model~\cite{fhtz22}, but in many practical machine learning applications, one usually encounters the \emph{vertex-arrival} model~\cite{kvv90}. For example, the left vertices of the graph represent Uber/Lyft drivers, and the right vertices represent customers. The goal is to match each incoming customer with a driver that is closest in geographical distance. Another common application is movie recommendations on streaming platforms such as Netflix. Here the left vertices are movies, and the right vertices are viewers, and the goal is to find each viewer his/her favorite movie, defined by the distance between the feature embeddings of movies and viewers~\cite{kbv09,sk09,slh14}.
\footnote{We note that in movie recommendations, one movie could be matched with multiple viewers. This scenario has also been studied and our results would also apply~\cite{fhtz22}.}

We consider the vertex-arrival online weighted matching problem when the right vertices are chosen by an \emph{adaptive adversary}, i.e., the next vertex $v_{t+1}$ could depend on previous matching results $\{(u_i,v_i) \}_{i=1}^t$. This setting is particularly common for practical applications such as Ride Apps, since if a certain app raises its price due to high demand, the customers might choose to use another app, or move to a new location with smaller demand. A natural heuristic for this problem is the \emph{greedy matching approach}, i.e., each time we observe $v_t$, we simply choose $u_t$ to be the point which maximizes $\frac{1}{\|u_t-v_t\|}$. In this weighted problem greedy has no worst-case approximation guarantee, so we view it only as the heuristic that our data structure implements approximately. One could reduce this problem to a nearest neighbor search problem against an adaptive adversary. There is a caveat though: if we restrict it to be a bipartite matching, then once a vertex $u$ is matched, it should not be matched again. This means our adaptive LSH should delete the point that has been outputted by the data structure after each query. We show that it is possible to further augment our framework for the data structure to be decremental (Section~\ref{sec:batch}).

\begin{theorem}[Online Weighted Matching] \label{thm:matching}
Let $c>1$ and $r>0$ be the approximation and distance parameters, and let $U\subseteq \R^d$ be an $n$-point dataset satisfying Assumption~\ref{assumption:ann_sparse} with parameter $s$ for the $\ell_2$ norm, and assume $\log n\geq 2$ and $\log n>(c+1)/(c-1)$. Given a possibly adaptive sequence $\{v_1,\ldots,v_T\}\subseteq \R^d$ with $T\leq n$, we can design an adaptive data structure which
(1) Preprocesses $U$ in time $\wt O(\Tmat(\sqrt T\cdot s, d, n)+\sqrt T\cdot s\cdot n^{1+\rho})$;
    (2) Uses space ${\cal S}_{\rm space}=\wt O(\sqrt T\cdot s\cdot (n^{1+\rho}+d)+nd)$;
    (3) Given a point $v_t$, it returns a point of $U\setminus \{u_1,\ldots,u_{t-1}\}$ within distance $cr$ of $v_t$ whenever $B(v_t,r)\cap (U\setminus \{u_1,\ldots,u_{t-1}\})\neq\emptyset$, and otherwise NULL or such a point, in time $\wt O(s\cdot (d+n^\rho))$. This step succeeds, and runs within this time, with probability at least $1-1/\poly(n)$;
    (4) Deletes the returned point $u_t$ from the data structure in worst-case time $\wt O(\sqrt{T}\cdot s\cdot n^\rho)$.
\end{theorem}
Theorem~\ref{thm:matching} follows from Theorem~\ref{thm:batch}, and its deletion bound comes from removing the stored sketch of the deleted point from each of the $\wt O(\sqrt T\cdot s)$ structures. For constants $p\in(0,2]$ and $c>1$, the same holds for the edge weights $\frac{1}{\|u-v\|_p}$, with Assumption~\ref{assumption:ann_sparse} for $\ell_p$ balls and distances measured in $\ell_p$, by Theorem~\ref{thm:lp_delete} in place of Theorem~\ref{thm:batch}. Here $\rho=1/c^p+o(1)$, and the hypotheses on $\log n$ above become $\log n\geq 6$ and $\log n>3/(c^p-1)$.
\subsection{Terminal Embedding}
A terminal embedding concerns the following problem: given a metric space $(X, d_X)$ and a set of points $u_1,\ldots,u_n\in X$, the goal is to design an embedding to another metric space $(Y, d_Y)$ such that for any $q\in X$, $C\cdot d_X(u_i, q)\leq d_Y(u_i, q)\leq C\mu\cdot d_X(u_i, q)$ for all $i\in [n]$, where $C>0$ is a constant and $\mu\geq 1$ is the distortion factor. In contrast to metric embeddings such as the Johnson-Lindenstrauss lemma, a terminal embedding requires the distance to be preserved between a fixed set of terminals and \emph{all points} in the metric space. When both $X$ and $Y$ are the Euclidean metric, it has been shown that an embedding dimension of $O(\epsilon^{-2}\log n)$ is possible for $1+\epsilon$ distortion~\cite{nn19}. Cherapanamjeri and Nelson~\cite{cn21} give a data structure that computes such an embedding for a query point $q$ that may be chosen with knowledge of the data structure. With a suitable approximate partitioning scheme, it answers a query in time $O^*(d+n^{1-c_1\epsilon^2})$ with space $O^*(nd)$, for an absolute constant $c_1>0$, where $O^*(\cdot)$ hides factors $\poly(1/\epsilon)\cdot (nd)^{o(1)}$~\cite[Section~3]{cn21}. To be correct against such queries, it keeps $\wt O(d)$ independent copies of its randomized parts. We show that our framework evaluates a $(1\pm\epsilon)$ terminal embedding correctly on every query with few data points near it, among $T$ adaptive queries that need not all have this property, with $\wt O(\sqrt T\cdot s)$ copies of an LSH structure at each distance scale (Theorem~\ref{thm:terminal} and Corollary~\ref{cor:terminal_lsh}). This uses LSH instead of the data structures of~\cite{cn21} (Remark~\ref{rem:terminal}), and it does not improve their space, as we discuss after Corollary~\ref{cor:terminal_lsh}. We first recall a definition from~\cite{cn21}.

\begin{definition}
Given $U=\{u_1,\ldots,u_n\}\subseteq \R^d$ and $\epsilon>0$, we say a matrix $S\in \R^{k\times d}$ is a convex hull distortion for $U$ if for any $z\in {\rm conv}(\{\frac{u-v}{\|u-v\|}: u, v\in U, u\neq v \})$, we have $|\|Sz\|-\|z\||\leq \epsilon$.
\end{definition}

In the rest of this subsection, all norms are $\ell_2$ norms, and $B(q,r)$ denotes the closed $\ell_2$ ball of radius $r$ centered at $q$. For $x\in\R^k$ and $y\in\R$, we write $(x,y)\in\R^{k+1}$ for the vector whose first $k$ coordinates are those of $x$ and whose last coordinate is $y$.

\begin{definition}[Oblivious reporting structure]\label{def:reporting}
Let $c>1$, $r>0$ and $s\geq 1$. A randomized data structure ${\cal A}$, built on $U$ with a random string $\zeta$, is an \emph{oblivious $(c,r,s)$-reporting structure} if, given $q\in\R^d$, it returns a set ${\cal A}_\zeta(q)\subseteq U$ together with the distances $\|q-u\|_2$ for $u\in{\cal A}_\zeta(q)$, such that
(1) for every $\zeta$ and every $q$, ${\cal A}_\zeta(q)\subseteq B(q,cr)\cap U$ and $|{\cal A}_\zeta(q)|\leq s$, and
(2) for every fixed $q$ with $|B(q,cr)\cap U|\leq s$, $\Pr_\zeta[\,B(q,r)\cap U\subseteq{\cal A}_\zeta(q)\,]\geq 9/10$.
We write ${\cal T}_{\rm prep}$ for its preprocessing time, ${\cal S}_{\rm space}$ for its space, not counting the $nd$ words of $U$, which all structures share, and ${\cal T}_{\rm query}$ for its worst-case query time.
\end{definition}

Locality-sensitive hashing gives such structures (Lemma~\ref{lem:terminal_lsh}). For a point $q\in\R^d$, its nearest-neighbor distance is $R(q)=\min_{u\in U}\|q-u\|_2$.

\begin{assumption}[Sparse terminal queries]\label{assumption:terminal}
Let $\epsilon\in(0,1)$, $c>1$, $r_0>0$, $\Lambda\geq0$ and $s\geq1$, and let $C_0=128$. A point $v\in\R^d$ is \emph{sparse} if
\[
R(v)\leq\Lambda\qquad\text{and}\qquad\big|\,U\cap B\big(v,\ c\cdot\max\{r_0,\ C_0R(v)/\epsilon\}\big)\big|\leq s ,
\]
and a sequence of (possibly adaptive) queries $v_1,\ldots,v_T\in\R^d$ satisfies the assumption if every $v_t$ is sparse.
\end{assumption}
The assumption asks that a query have at most $s$ data points within distance $O(c/\epsilon)$ times its nearest-neighbor distance, or within distance $cr_0$ when this is larger. Its second condition is Assumption~\ref{assumption:general} with $f_v(U)=U\cap B(v,c\max\{r_0,C_0R(v)/\epsilon\})$. If $r_0<\min_{u\neq u'\in U}\|u-u'\|_2/(2c)$, a ball of radius $cr_0$ contains at most one data point, so only the radius $C_0R(v_t)/\epsilon$ matters.

\begin{theorem}[Adaptive terminal embedding]\label{thm:terminal}
Let $U\subseteq\R^d$ be a dataset of $n$ points, let $\epsilon\in(C_1n^{-1/2},1)$, where $C_1=32$, let $\delta_{\rm fail}\in(0,1)$ and $T\geq1$, and let $c,r_0,\Lambda,s$ be as in Assumption~\ref{assumption:terminal}. Let $r_i=2^ir_0$ for $i=0,\ldots,L$, where $L+1=O(\log(2+\Lambda/(\epsilon r_0)))$, and let $l=O(s\log(sT(L+1)/\delta_{\rm fail}))$ and $M=O(l\sqrt{T\log(T(L+1)/\delta_{\rm fail})})=\wt O(\sqrt T\cdot s)$, with the values fixed in Section~\ref{sec:terminal_proof}. Suppose that for each $i$ we are given an oblivious $(c,r_i,s)$-reporting structure with costs ${\cal T}_{\rm prep},{\cal S}_{\rm space},{\cal T}_{\rm query}$. Then there is a randomized algorithm that computes a data structure ${\cal D}$, which holds $M$ independent copies of the reporting structure at each of the $L+1$ radii, and a matrix $S\in\R^{k\times d}$ with $k=O(\epsilon^{-2}\log n)$, such that the following hold.
\begin{enumerate}
\item With probability at least $1-n^{-10}$ over the construction, $S$ is a convex hull distortion for $U$ with parameter $\epsilon/C_1$. If it is, a map $F:\R^d\to\R^{k+1}$ that depends only on $U$ and $S$ satisfies $F(u)=(Su,0)$ for $u\in U$ and, for all $q\in\R^d$ and $u\in U$,
\[
(1-\epsilon)\|q-u\|_2\leq\left\|F(q)-\begin{bmatrix}Su\\0\end{bmatrix}\right\|_2\leq(1+\epsilon)\|q-u\|_2 .
\]
\item Let $v_1,\ldots,v_T\in\R^d$ be queries, not necessarily sparse, where $v_t$ may depend on $U$, $S$, the answers to $v_1,\ldots,v_{t-1}$ and the sets that the query algorithm of Section~\ref{sec:terminal_proof} releases while answering them, but not on the random strings of the copies, the sampled indices or the noise. With probability at least $1-\delta_{\rm fail}-n^{-10}$, for every $t\in[T]$ such that $v_t$ is sparse (Assumption~\ref{assumption:terminal}), ${\cal D}$ outputs $z_{v_t}=F(v_t)\in\R^{k+1}$. In particular, for every such $t$,
\[
(1-\epsilon)\|v_t-u\|_2\leq\left\|z_{v_t}-\begin{bmatrix}Su\\0\end{bmatrix}\right\|_2\leq(1+\epsilon)\|v_t-u\|_2\quad\text{for all }u\in U.
\]
The answers to queries that are not sparse carry no guarantee, and they may depend on the internal randomness of ${\cal D}$.
\item In the real-RAM model with exact square roots (Section~\ref{sec:terminal_proof}), the worst-case query time, on every query, is $O\big((L+1)\,l\,({\cal T}_{\rm query}+s)+sd+k^3(k+s)\log(2/\epsilon)\big)$.
\item The space is $(L+1)M\cdot{\cal S}_{\rm space}+O(nd+nk+kd)$, and the preprocessing time is $(L+1)M\cdot{\cal T}_{\rm prep}+O(nkd)$.
\end{enumerate}
\end{theorem}

\begin{corollary}[LSH instantiation]\label{cor:terminal_lsh}
Let $c>1$ be a constant. For every $r>0$, let ${\cal H}_r$ be a family of hash functions on $\R^d$ with $\Pr_{h\sim{\cal H}_r}[h(x)=h(y)]\geq p_1$ if $\|x-y\|_2\leq r$ and $\Pr_{h\sim{\cal H}_r}[h(x)=h(y)]\leq p_2$ if $\|x-y\|_2>cr$, where $0<p_2<p_1<1$ depend only on $c$, and each function takes $O(d)$ words and $O(d)$ evaluation time, for example the family of~\cite{diim04}. Let $\rho=\ln(1/p_1)/\ln(1/p_2)$. Then Theorem~\ref{thm:terminal} holds with
\begin{itemize}
\item space $\wt O\big((L+1)\sqrt T\cdot s\cdot(n^{1+\rho}+n^{\rho}d)\big)+O\big(nd+\epsilon^{-2}(n+d)\log n\big)$,
\item preprocessing time $\wt O\big((L+1)\sqrt T\cdot s\cdot n^{1+\rho}d\big)+O(\epsilon^{-2}nd\log n)$,
\item query time $\wt O\big((L+1)\,s\,(n^{\rho}(d+s)+sd)+\epsilon^{-6}(\epsilon^{-2}+s)\big)$,
\end{itemize}
where $\wt O(\cdot)$ hides factors polylogarithmic in $n,s,T,L,1/\epsilon,1/\delta_{\rm fail}$ and factors that depend only on $c$.
\end{corollary}

Theorem~\ref{thm:terminal} evaluates one fixed terminal embedding $F$, a variant of the outer extension of~\cite{nn19} as described in~\cite[Section~2]{cn21} that keeps only the constraints of data points near the nearest neighbor. Since $F$ is a fixed function of $U$ and $S$, a repeated sparse query receives the same image. Each query visits each radius at most once, and at each radius it runs one private round, in which every sampled copy is queried at the query point itself. So the copies at one radius take part in at most $T$ private rounds, whether or not the queries are sparse, and at the radius where a sparse query stops, sparsity bounds the number of data points that its round must report. Two baselines put the result in context. First, $F(q)$ can be computed for every $q$, deterministically given $S$ and with no assumption on $q$, by finding $u^*(q)$ and $N(q)$ (Section~\ref{sec:terminal_proof}) by a linear scan and running the ellipsoid method of Lemma~\ref{lem:terminal_ellipsoid}. This takes time $O(nd+k^3(k+n)\log(2/\epsilon))$ and space $O(nd+nk+kd)$. Second, the data structure of~\cite{cn21} answers every query, with no assumption, in time $O^*(d+n^{1-c_1\epsilon^2})$ and space $O^*(nd)$, and in a stronger model: a query may be chosen with full knowledge of the data structure, and once the data structure is built successfully, it is answered correctly with probability $1-1/\poly(n)$ over the query-time randomness alone~\cite[Section~3]{cn21}. The private release matters because the adversary may also ask queries that are not sparse. If every query were sparse, every correct answer would be $F(v_t)$, which depends only on $U$, $S$ and $v_t$, and $O(\log(T(L+1)/\delta_{\rm fail}))$ copies at each radius would already suffice: query all of them, and release the points of their merged reports that lie in $B(v,r_i)$. With high probability every released set is then $B(v,r_i)\cap U$ and every answer is $F(v)$, both fixed functions of the query, so the adversary asks the same queries as against an oracle that releases $B(v,r_i)\cap U$ and answers $F(v)$, and a union bound over the at most $T(L+1)$ rounds applies. On a query that is not sparse, a ball of radius $cr_i$ can hold more than $s$ data points, each copy reports at most $s$ of them, and which ones it reports depends on its random string, so the answers can reveal the random strings of the copies. Against this adversary, Corollary~\ref{cor:terminal_lsh} uses $\wt O(\sqrt T\cdot s)$ LSH copies at each of the $L+1$ radii, instead of fresh copies for each of the up to $T(L+1)$ rounds, which is fewer when $s$ is below $\sqrt T$ by more than polylogarithmic factors (with the values of $l$, $M$ and $\beta$ fixed in Section~\ref{sec:terminal_proof}, $M\leq T$ holds once $T\geq7.2\cdot10^5\,l^2\ln(200/\beta)$, so the saving is asymptotic), and for fixed $\epsilon$ its query time depends on $n$ only through a factor $n^{\rho}$ and polylogarithmic factors. This is a different trade-off, not an improvement of the space of~\cite{cn21}: the space bound of Corollary~\ref{cor:terminal_lsh} is that of the first baseline plus the space of the LSH copies, which buy this query time.

\begin{remark}[Constant distortion under an ANN-type assumption]\label{rem:terminal_const}
Call $v\in\R^d$ \emph{good} if $R(v)\leq\Lambda$ and $|U\cap B(v,c\cdot\max\{r_0,2R(v)\})|\leq s$, and let $v_1,\ldots,v_T$ be queries as in item 2 of Theorem~\ref{thm:terminal}, not necessarily good. Use the structure of Theorem~\ref{thm:terminal} with $L_B=\lceil\log_2\max\{1,\Lambda/r_0\}\rceil$ in place of $L$, also in the parameters of Section~\ref{sec:terminal_proof}. Change the query as follows. It stops at the first radius $r_i$ at which some released point lies within distance $r_i$ of the query, and it then outputs $(Su',R'_i)$, where $u'$ is the released point of smallest index among those nearest to the query and $R'_i$ is its distance to the query. If it never stops, it outputs $(Su_1,0)$. Then with probability at least $1-\delta_{\rm fail}-n^{-10}$, for every $t\in[T]$ such that $v_t$ is good, the output is $z_{v_t}=(Su_t^*,R(v_t))$, where $u_t^*$ is a nearest neighbor of $v_t$ in $U$. Moreover, for every such $t$, $\frac{1-\epsilon}{\sqrt2}\|v_t-u\|_2\leq\|z_{v_t}-(Su,0)\|_2\leq(\sqrt5+3\epsilon)\|v_t-u\|_2$ for all $u\in U$, a distortion of $\sqrt{10}+O(\epsilon)$. The query time is $O((L_B+1)\,l\,({\cal T}_{\rm query}+s)+k)$ on every query, and no convex program is solved.
\end{remark}

\begin{remark}[The data structures of~\cite{cn21}]\label{rem:terminal}
Theorem~\ref{thm:terminal} uses oblivious reporting structures, such as LSH, and not the data structures of~\cite[Theorem~3.1]{cn21}, because our framework does not apply to the latter directly. One call to their separation oracle~\cite[Section~3]{cn21} may issue $O^*(n^{\rho_4})$ near neighbor queries, where $\rho_4$ is a parameter of their approximate partitioning data structure, and Theorem~\ref{thm:main_general} would count each of these queries as one of its $T$ queries. Also, Assumption~\ref{assumption:general} on $v_1,\ldots,v_T$ says nothing about the answer sets of these internal queries. In the model of Theorem~\ref{thm:terminal}, with $T$ queries that may depend on the earlier answers and released sets but not directly on the internal randomness, whether the $\wt\Theta(d)$ copies of~\cite{cn21} can be reduced to $O^*(\sqrt T\cdot s)$ remains open. In the model of~\cite{cn21}, where a query may depend on the data structure itself, our framework does not apply.
\end{remark}

\ifdefined\isarxiv

\else
\section*{Acknowledgement}
Ying Feng was supported by an MIT Akamai Presidential Fellowship. 
David Woodruff would like to acknowledge support from a Simons Investigator Award and Office of Naval Research (ONR) award number N000142412647. Lichen Zhang was supported in part by NSF CCF-1955217 and NSF DMS-2022448. The corrections in this version were found and written with the help of Claude Opus 5.5. The authors have verified the corrections and take full responsibility for the content of this paper.
\section*{Impact Statement}
This paper develops novel adaptive data structures to speed up machine learning tasks. It potentially would lead to efficiency improvement and reduction in carbon emission. We do not foresee any societal consequences worth highlighting.
\fi

\ifdefined\isarxiv
\else
\bibliography{ref}
\bibliographystyle{icml2025}

\fi

\ifdefined\isarxiv
\else
\newpage
\onecolumn
\appendix
\section*{Appendix}
\fi
\section{Preliminaries}

We use $\wt O(\cdot)$ to suppress $\poly\log$ factors in $n, d, 1/\delta$. We use $\E[\cdot]$ and $\Pr[\cdot]$ to denote expectation and probability. We use $\Tmat(n,m,d)$ to denote the time to multiply an $n \times m$ matrix with an $m \times d$ matrix. For a positive integer $n$, we use $[n]$ to denote the set $[n]:=\{1,2,\cdots, n\}$. For a vector $x$, we use $x^\top$ to denote the transpose of $x$. For a vector $x \in \R^n$, we use $\| x \|_2$ to denote its $\ell_2$ norm, i.e., $\| x \|_2: = (\sum_{i=1}^n x_i^2)^{1/2}$. For two real numbers $a, b>0$, we will use $a=(1\pm\alpha) b$ to denote $a\in [(1-\alpha)b, (1+\alpha)b]$ for $\alpha>0$. 

For a point $v\in \R^d$ and $r>0$, we use $B_{\|\cdot\|}(v, r)$ to denote the ball centered at $v$ of radius $r$ in some norm. When the norm $\|\cdot\|$ is clear from the context, we abbreviate this with $B(v,r)$. We use $\mathbb{S}^{d-1}$ to denote the unit sphere of dimension $d$. Let $(\Omega, {\cal F})$ be a measurable space and let probability measures $P, Q$ be defined on $(\Omega, {\cal F})$. The total variation (TV) distance between $P$ and $Q$ is defined to be $\d_{\TV}(P,Q)=\sup_{A\in {\cal F}}|P(A)-Q(A)|$.

We will be using exponential random variables extensively, and we recall its definition here.

\begin{definition}[Exponential Distribution]
The \emph{exponential distribution with parameter $\lambda>0$}, denoted by ${\rm Exp}(\lambda)$, is the distribution with mean $\lambda$ and PDF
\begin{align*}
    f(x; \lambda) = & ~ \begin{cases}
        \frac{1}{\lambda} e^{-x/\lambda}, & \text{if $x\geq 0$}, \\
        0, & \text{otherwise}.
    \end{cases}
\end{align*}
\end{definition}

We will also work with order statistics.

\begin{definition}[Order Statistics]
Let $X_1,\ldots,X_n$ be $n$ i.i.d. random samples, the \emph{$k$-th order statistics}, denoted by $X_{(k)}$, is the random variable associated with the $k$-th smallest value over the samples.
\end{definition}

For extreme order statistics such as $X_{(1)}$ and $X_{(n)}$, they could be generated efficiently via inverse CDF sampling.

\begin{fact}
\label{fact:generate_order}
Let $n\geq 1$ be a positive integer, suppose we can generate $U\sim {\rm Unif}(0, 1)$ in $O(1)$ time, then for any distribution with a closed-form CDF distribution, $X_{(1)}$ and $X_{(n)}$ can be generated in $O(1)$ time. 
\end{fact}

\section{Augmenting Oblivious ANN Data Structures with Differential Privacy}
\label{sec:dp}

We develop a generic algorithm that takes an oblivious ANN data structure, and transforms it to one that is adaptive. We will start with a simple but inefficient algorithm, and improve its efficiency afterwards.

We require these data structures to output all answers they find (within their runtime budget, which we will formalize) instead of outputting only a single answer. In the context of LSH, this means that instead of outputting one approximate near neighbor, it needs to output all approximate near neighbors it finds. As the runtime budget for a query is $n^\rho$, it can output at most $n^\rho$ approximate near neighbors.

Before introducing the algorithm, we give some background in differential privacy.

\subsection{Preliminaries on Differential Privacy}

Differential privacy refers to a class of algorithms for which, when one slightly perturbs the input, the output distribution remains relatively close. This intuitively will mean that for an adversary, it cannot learn useful information by adaptively adjusting its input to an algorithm.

 \begin{definition}[Differential Privacy]
We say a randomized algorithm ${\cal A}$ is $(\epsilon,\delta)$-differentially private if for any two databases $S$ and $S'$ that differ only by one row and any subset of outputs $T$, the following
\begin{align*}
    \Pr[ {\cal A}(S) \in T ] \leq e^{\epsilon} \cdot \Pr[ {\cal A}(S') \in T ] + \delta ,
\end{align*}
holds. Here the probability is over the randomness of ${\cal A}$.
\end{definition}

Differential privacy can be composed, and the advanced composition theorem says that $k$-fold adaptive composition only blows up the $\epsilon$ parameter by a factor of roughly $\sqrt{k}$ instead of $k$. Similar to~\cite{bkm+21}, we will utilize advanced composition to reduce the number of oblivious data structures required and gain runtime improvements.

\begin{theorem}[Advanced Composition, see \cite{drv10}] \label{thm:ada_c}
Given three parameters $\epsilon, \delta_0, \delta \geq 0$. If ${\cal A}_1, \cdots, {\cal A}_k$ are each $(\epsilon,\delta)$-$\mathrm{DP}$ algorithms, then the $k$-fold adaptive composition ${\cal A}_k \circ \cdots \circ {\cal A}_1$ is $(\epsilon_0, \delta_0 + k \delta)$-$\mathrm{DP}$ where
\begin{align*}
\epsilon_0 = \sqrt{2k \ln (1/\delta_0)} \cdot \epsilon + 2k \epsilon^2
\end{align*}
\end{theorem}

One can boost the privacy budget by subsampling the dataset first, then running the differentially private algorithm on the subsampled set.

\begin{theorem}[Amplification via Sampling (Lemma 4.12 of \cite{bnsv15})]\label{thm:amp}
Let $\epsilon\in (0,1]$ and $\delta\in[0,1]$ be parameters.
Let ${\cal A}$ denote an $(\epsilon,\delta)$-$\mathrm{DP}$ algorithm. Let $S$ denote a dataset with $n$ rows.

Suppose that ${\cal A}'$ is an algorithm that,
\begin{itemize}
    \item constructs a database $T \subset S$ by subsampling with repetition $k \leq n/2$ rows from $S$,
    \item returns ${\cal A}(T) $.
\end{itemize}
Then, we have
\begin{align*}
{\cal A}' \mathrm{~~~is~~~} \left( \frac{6k}{n} \epsilon , \exp\Big(\frac{6k}{n}\epsilon\Big)\cdot\frac{4k}{n}\cdot\delta\right)-\textnormal{DP}.
\end{align*}
In particular, for $\delta=0$, ${\cal A}'$ is $(\frac{6k}{n}\epsilon,0)$-DP.
\end{theorem}

\begin{theorem}[Generalization of Differential Privacy ($\mathrm{DP}$)] \label{thm:general}
Given two accuracy parameters $\epsilon \in (0,1/3)$ and $\delta \in (0,\epsilon/4)$,  suppose that the parameter $t$ satisfies that $t \geq \epsilon^{-2} \log(2\epsilon/\delta)$. 

We use ${\cal D}$ to represent a distribution over a domain ${\cal X}$. Suppose $S\sim {\cal D}^t$ is a database containing $t$ elements sampled independently from ${\cal D}$. Let ${\cal A}$ be an algorithm that, given any database $S$ of size $t$, outputs a predicate $h:{\cal X}\rightarrow \{0,1\}$.

If ${\cal A}$ is $(\epsilon,\delta)$-$\mathrm{DP}$, then the empirical average of $h$ on sample $S$, i.e., 
\begin{align*} 
h(S)=\frac{1}{|S|}\sum_{x\in S}h(x),
\end{align*}
and $h$'s expectation over underlying distribution ${\cal D}$, i.e.,
\begin{align*} 
h({\cal D})=\E_{x\sim {\cal D}}[h(x)]
\end{align*} 
are within $10\epsilon$ with probability at least $1-{\delta}/{\epsilon}$:
\begin{align*}
    \Pr_{S\sim {\cal D}^t, h\leftarrow {\cal A}(S)} \Big[ \big| h(S) - h({\cal D}) \big| \geq 10\epsilon \Big] \leq {\delta}/{\epsilon}.
\end{align*}
\end{theorem}

We next consider the task of differentially private selection: we have a discrete space of outcomes, and the goal is to output the largest histogram cell:

\begin{definition}[Differentially Private Selection]
Given $n$ categories, let the database $S$ be such that, each of its rows is a binary vector of length $n$, with each entry corresponding to whether the row belongs to the $i$-th category or not. We say an algorithm is an \emph{$(\epsilon,0)$-differentially private selection} if it releases the index of the approximately most popular category over the database $S$, and is $(\epsilon,0)$-DP.
\end{definition}

For any differentially private mechanism, one cares about privacy but also \emph{utility}, i.e., after adding noise, how good are the estimates compared to the algorithm without adding noise? Here, we focus on privacy, and we will later use the privacy to provide utility guarantees. {\color{black}For count queries, the most common approach is the \emph{Laplace mechanism}, where properly chosen Laplace noises are added to each count. Here, we use a similar approach that instead adds \emph{exponential} noises. This approach is more widely adopted for Exponential Mechanism, yet it still offers the same privacy budget as that of the standard Laplace Mechanism for reporting the noisy max index. See e.g.,~\cite{dr14,dks+21} for more details.

\begin{theorem}[\textsc{ReportOneSidedNoisyArgMax}, Theorem 3.13 of~\cite{dr14}]
\label{thm:rpm}
Given a database $S$ with each row being a binary vector of length $n$, consider the following algorithm:
\begin{itemize}
    \item Compute the count for each category;
    \item Add an i.i.d. exponential noise ${\rm Exp}(2/\epsilon)$ (of mean $2/\epsilon$) to each count;
    \item Output the category with the maximum noisy count.
\end{itemize}
The algorithm gives $(\epsilon,0)$-differentially private selection.
\end{theorem}
}

\subsection{Reducing Monte Carlo Search to Selection: A Slow Algorithm}

In this section, we show how to utilize the \textsc{ReportNoisyMax} procedure for adaptive data structure design. Recall the framework introduced in~\cite{hkm+22} for streaming and later extended for improving the preprocessing and update time for dynamic data structures~\cite{bkm+21,syyz23_dp,csw+23}: the idea is to set up the database as the random strings used by oblivious data structures. Suppose we have ${\cal A}_1,\ldots,{\cal A}_k$ oblivious data structures, each with their associated random strings $r_1,\ldots,r_k\in \{0,1\}^*$. The database is $R$, where the $i$-th row is the random string $r_i$. Notions such as sensitivity are in turn defined with respect to these random strings. We will now illustrate how to tie these Monte Carlo search data structures to the selection problem. We will use $(c,r)$-ANN as an example, but our reduction does not exploit any additional structure on $(c,r)$-ANN, so it can be well-extended to more general search data structures.

Given a dataset $U\in \R^d$ and a possibly adaptive query $v\in \R^d$, recall that a $(c,r)$-ANN data structure would have the following behavior:
\begin{itemize}
    \item If there exists some $u\in U$ with $\|u-v\|\leq r$, then it outputs some point $u'\in U$ with $\|u'-v\|\leq cr$. This succeeds with probability at least $\frac{9}{10}$, which is the hypothesis $\delta\leq 1/10$ of Theorem~\ref{thm:main_general} (note we do not need probability boosting for our base data structure, at this stage). We modify the data structure to output all such points it has found.
    \item If there is no such point, report NULL.
\end{itemize}
Note that ANN data structures can only produce false negatives but no false positives. This is the second hypothesis of Theorem~\ref{thm:main_general}: with probability $1$, every point the data structure returns lies in $f_v(U)=B(v,cr)\cap U$, since LSH checks the distance of every candidate it reports. We can associate the output of the data structure with a binary vector of length $n$ as follows: the $i$-th entry of the vector corresponds to the decision the data structure has made on the point $u_i$, i.e., if it outputs $u_i$, then we set the $i$-th entry to be 1, and 0 otherwise. If the data structure outputs nothing, we set the vector to be ${\bf 0}_n$. Fixing dataset $U$ and query $v$, these binary vectors are completely determined by $r_1,\ldots,r_k$. Denote these binary vectors as $b^{1},\ldots,b^k$. We then execute the \textsc{ReportOneSidedNoisyArgMax} mechanism on these vectors, and output the corresponding index. We now describe the details of our algorithm.

We will use $T$ to denote the number of queries, and $k$ to denote the total number of data structures. We assume each data structure succeeds with probability at least $1-\delta$. We let $\epsilon_{\DP}$ denote the DP parameter for \rpm. We will use $\delta_{\rm fail}$ to denote the overall failure probability of our algorithm, and $\beta=\delta_{\rm fail}/(2T)$ to denote the failure probability of a single query step. The algorithm is as follows.

\vspace{0.2cm}
\noindent{\bf Initialization.}
\begin{itemize}
    \item Prepare $k$ independent copies of oblivious data structures, denoted by ${\cal A}_1,\ldots,{\cal A}_k$, over the dataset $U$.
\end{itemize}
\vspace{0.2cm}
\noindent{\bf Query.}
\begin{itemize}
    \item Receive query vector $v\in \R^d$.
    \item Sample $l=\wt O(s)$ indices independent and uniformly from $[k]$ with replacement. Denote the corresponding data structures by ${\cal A}_{(1)},\ldots,{\cal A}_{(l)}$.
    \item Feed $v$ into ${\cal A}_{(1)},\ldots,{\cal A}_{(l)}$, receive binary vectors $b^{(1)},\ldots,b^{(l)}$.
    \item Compute $\rpm(b^{(1)},\ldots,b^{(l)})$ with parameter $\epsilon_{\DP}$. {\color{black}Let $i$ be the corresponding index.}
    \item {\color{black} Output $u_i$ if $\|u_i-v\|_2\leq cr$, otherwise output NULL.}
\end{itemize}
\vspace{0.2cm}
\noindent{\bf Update.}
\begin{itemize}
    \item Receive update vector $u\in \R^d$.
    \item Update ${\cal A}_1,\ldots,{\cal A}_k$ with $u$.
\end{itemize}
We will specify the parameters to show that the algorithm is indeed differentially private with respect to the internal randomness of the data structures.

\vspace{0.2cm}
\noindent{\bf Setting Parameters.}
We start by setting the parameter $\epsilon_{\DP}$ which affects the privacy guarantee of \rpm. We will set $\epsilon_{\DP}=1/2$ (in fact, for our analysis, $\epsilon_{\DP}$ can be any constant smaller than 1). Set the number of samples $l$ to be
\begin{align*}
    l = & ~ O(s\log^2(n/\beta))
\end{align*}
with a large enough constant, so that $l\geq 200\ln(4T/\delta_{\rm fail})$, and the total number of data structures $k$ to be
\begin{align*}
    k = & ~ 200\cdot 6l\cdot \epsilon_{\DP}\cdot \sqrt{2T\ln(100/\beta)} \\
    = & ~ \wt O(\sqrt T\cdot s).
\end{align*}

\vspace{0.2cm}
\noindent{\bf Privacy Guarantees.}
Our analysis differs from standard DP analysis, in which one could reason over the privacy and utility simultaneously. We will start by showing that the algorithm is indeed differentially private, then use this fact to argue the output of the algorithm provides good utility. 

To understand the privacy of the algorithm, we adapt the framework from~\cite{bkm+21}: let $r_1,\ldots,r_k\in \{0,1\}^*$ denote the random strings used by data structures ${\cal A}_1,\ldots,{\cal A}_k$ and let $R=\{r_1,\ldots,r_k\}$ be the database whose $i$-th row is $r_i$. We will prove the transcript of the interaction between the adaptive adversary and algorithm is differentially private with respect to $R$. Fix a time step $t$. We let ${\bf out}_t(R)$ denote the index output by the algorithm at time $t$. Even fixing $R$, the output ${\bf out}_t(R)$ is still a random variable as it depends on i.i.d. exponential noise that is  oblivious to the algorithm. Similar to~\cite{bkm+21}, we define the transcript at time $t$ as ${\cal T}_t(R)=(v_t, {\bf out}_t(R))$. Define the overall transcript as
\begin{align*}
    {\cal T}(R) = & ~ {\cal T}_1(R),\ldots,{\cal T}_T(R).
\end{align*}
We view ${\cal T}_t$ and ${\cal T}$ as algorithms that given a database $R$, output the transcripts. This enables us to reason about differential privacy with respect to $R$.
\begin{lemma}
\label{lem:DP_single_step}
\label{lem:T_i_DP}
For any time step $t$, ${\cal T}_t$ is $(\frac{6l}{k}\cdot \epsilon_{\DP},0)$-DP with respect to $R$.
\end{lemma}

\begin{proof}
Note that ${\cal T}_t(R)=(v_t,{\bf out}_t(R))$. For any fixed single step $t$, $v_t$ does not provide extra information about $R$, so we only need to consider ${\bf out}_t(R)$. 

{\color{black} We let $\wt {\bf out}_t(R):=\rpm(b^{(1)},\ldots,b^{(l)})$ with parameter $\epsilon_{\DP}$.  As we have argued in the preceding discussion, the binary vectors are determined by $\{r_{(1)},\ldots,r_{(l)}\}\subset R$ which are subsampled from $R$. As $\wt {\bf out}_t(R)$ is $(\epsilon_{\DP},0)$-DP with respect to the subset (Theorem~\ref{thm:rpm}), by Theorem~\ref{thm:amp}, it is $(\frac{6l}{k}\cdot \epsilon_{\DP},0)$-DP with respect to $R$. Finally, observe that ${\bf out}_t(R)$ is a post-processing of $\wt {\bf out}_t(R)$ as it's a deterministic mapping that only depends on the output of $\wt {\bf out}_t(R)$. Hence, we conclude that ${\bf out}_t(R)$ is $(\frac{6l}{k}\cdot \epsilon_{\DP}, 0)$-DP with respect to $R$, as desired.}
\end{proof}
Given that a single step is DP with respect to $R$, we can then perform advanced composition to prove the privacy of ${\cal T}$.

\begin{lemma}
\label{lem:T_DP}
${\cal T}$ is $(\frac{1}{100},\frac{\beta}{100})$-DP with respect to $R$.
\end{lemma}

\begin{proof}
By Lemma~\ref{lem:DP_single_step}, we know that for any fixed $t\in [T]$, ${\cal T}_t(R)$ is $(\frac{6l}{k}\cdot\epsilon_{\DP},0)$-DP with respect to $R$. Further, observe that ${\cal T}$ is an adaptive composition of ${\cal T}_T\circ\ldots\circ {\cal T}_1$, so we can apply Theorem~\ref{thm:ada_c} with $\epsilon=\frac{6l}{k}\cdot \epsilon_{\DP}$ and $\delta=0, \delta_0=\beta/100$. This yields a mechanism that is $(\epsilon_0, \delta_0)$-DP with respect to $R$, where 
\begin{align*}
    \epsilon_0 = & ~ \sqrt{2T\ln(100/\beta)}\cdot \epsilon + 2T\epsilon^2 \\
    \leq & ~  \frac{1}{200}+\frac{1}{200} \\
    = & ~ \frac{1}{100}.
\end{align*}
Thus, we conclude ${\cal T}$ is $(\frac{1}{100},\frac{\beta}{100})$-DP with respect to $R$.
\end{proof}

\vspace{0.2cm}
\noindent{\bf Accuracy: Differentiating Signals from Noise.}
Given the privacy guarantees of ${\cal T}$, we are ready to prove accuracy against an adaptive adversary. Similar to~\cite{bkm+21}, define $\ov v_t=(v_1,v_2,\ldots,v_t)$ to be the query sequence up to time $t$ and consider feeding a random string $r$ (for initialization) and $\ov v_t$ to an oblivious data structure ${\cal A}$. Let ${\cal A}(r,\ov v_t)$ denote the output. Finally, define ${\rm acc}_{\ov v_t}(r)$ to be the indicator of whether ${\cal A}(r, \ov v_t)$ succeeds, i.e., if $v_t$ is a ``yes'' instance, then ${\cal A}(r, \ov v_t)$ indeed outputs a non-empty subset of points in $U$ that are at most $cr$ away from $v_t$. We will prove that with high probability, a large constant fraction of \emph{all} data structures succeed.
\begin{lemma}
\label{lem:acc_global}
For any fixed time step $t\in [T]$, we have $\sum_{j=1}^k {\rm acc}_{\ov v_t}(r_j)\geq \frac{4}{5}k$ with probability at least $1-\beta$.
\end{lemma}
\begin{proof}
The proof will rely on the generalization property of DP. First note that ${\rm acc}_{\ov v_t}$ is determined by the transcript ${\cal T}$, as $\ov v_t$ is a substring of ${\cal T}$ and $\ov v_t$ determines ${\rm acc}_{\ov v_t}$. To invoke Theorem~\ref{thm:general}, we first note that each row of $R$ is drawn uniformly, and due to Lemma~\ref{lem:T_DP}, we know that ${\cal T}$ is $(\frac{1}{100},\frac{\beta}{100})$-DP with respect to $R$. Thus, we have that the empirical average of ${\rm acc}_{\ov v_t}$ over $R$ which is $\frac{1}{k}\sum_{j=1}^k {\rm acc}_{\ov v_t}(r_j)$ and the expectation of ${\rm acc}_{\ov v_t}$ over the uniform distribution $\E[{\rm acc}_{\ov v_t}(r)]$ are close. Setting $\epsilon=1/100$, $\delta=\beta/100$ and $k\gg \frac{1}{\epsilon^2}\log(2\epsilon/\delta)$, we conclude
\begin{align*}
    \Pr\left[ \left|\frac{1}{k}\sum_{j=1}^k {\rm acc}_{\ov v_t}(r_j)-\E[{\rm acc}_{\ov v_t}(r)] \right|\geq \frac{1}{10} \right]\leq & ~ \beta.
\end{align*}
It remains to get a good handle on the expectation. If we let ${\cal U}$ denote the uniform distribution we sample $r$ from, then we can write the expectation $\E_{r\sim {\cal U}}[{\rm acc}_{\ov v_t}(r)]=\Pr_{r\sim {\cal U}}[{\rm acc}_{\ov v_t}(r)=1]$ as ${\rm acc}_{\ov v_t}(r)$ is an indicator. The crucial point here is that the probability is taken over the randomness of $r$ when the sequence of queries ${\ov v_t}$ is \emph{fixed}. Thus, we can utilize the property of our oblivious data structure, i.e., the hypothesis $\delta\leq 1/10$ of Theorem~\ref{thm:main_general}, which gives $\Pr_{r\sim {\cal U}}[{\rm acc}_{\ov v_t}(r)=1]\geq 1-\delta\geq \frac{9}{10}$. To put it together, we have
\begin{align*}
    \frac{1}{k}\sum_{j=1}^k {\rm acc}_{\ov v_t}(r_j) \geq & ~ \frac{9}{10}-\frac{1}{10} \\
    = & ~ \frac{4}{5}. 
\end{align*}
Thus, we complete the proof.
\end{proof}

To improve the efficiency of querying, recall that we resort to sampling. We prove that sampling does not hurt the accuracy.

\begin{lemma}
\label{lem:acc_sample}
For all $t\in [T]$, let $l$ be the sampled indices of time $t$, we have $\sum_{j=1}^l {\rm acc}_{\ov v_t}(r_{(j)})\geq \frac{3}{4}l$ with probability at least $1-\delta_{\rm fail}$.
\end{lemma}

\begin{proof}
Fix any $t\in [T]$, we condition on the event of Lemma~\ref{lem:acc_global}. Therefore $\sum_{j=1}^k {\rm acc}_{\ov v_t}(r_j)\geq \frac{4}{5}k$. This means that any i.i.d. sample (with replacement) from these $k$ indices succeeds with probability at least $4/5$. Consequently $\E[\sum_{j=1}^l {\rm acc}_{\ov v_t}(r_{(j)})]\geq \frac{4}{5}l$. By Hoeffding's bound, we have
\begin{align*}
    \Pr\left[\left|\sum_{j=1}^l {\rm acc}_{\ov v_t}(r_{(j)}) - \E[\sum_{j=1}^l {\rm acc}_{\ov v_t}(r_{(j)})] \right| \geq \alpha\right] \leq & ~ 2\exp(-\frac{2\alpha^2}{l}),
\end{align*}
and setting $\alpha=\frac{1}{20}l$, we conclude that $\sum_{j=1}^l {\rm acc}_{\ov v_t}(r_{(j)})\geq \frac{4}{5}l-\frac{1}{20}l=\frac{3}{4}l$ with probability at least $1-2\exp(-l/200)$, since $2\alpha^2/l=l/200$. A union bound over $t\in[T]$, with Lemma~\ref{lem:acc_global} failing with probability at most $\beta=\delta_{\rm fail}/(2T)$ at each step, gives failure probability at most $T\beta+2T\exp(-l/200)\leq \delta_{\rm fail}/2+\delta_{\rm fail}/2=\delta_{\rm fail}$, because $l\geq 200\ln(4T/\delta_{\rm fail})$ gives $\exp(-l/200)\leq \delta_{\rm fail}/(4T)$.
\end{proof}

We present an abstract way to reason over the data structure task. Fix the preprocessed dataset $U\in \R^d$. Consider any (possibly adaptive) query $v\in \R^d$,  and let $b\in \{0,1\}^n$ denote the characteristic vector of $v$ with respect to $U$, in terms of $(c, r)$-ANN: if there exists some $u_i\in U$ with $\|u_i-v\|\leq r$, then: for any $j\in [n]$ with $\|u_j-v\|\leq cr$, we set $b_j=1$. Note that if no point in $U$ is $r$-close to $v$, then $b$ is either the all-0s vector or it has some nonzero entries that are $cr$-close to $v$. Similarly, for each data structure ${\cal A}_1,\ldots,{\cal A}_k$ and their corresponding random strings $r_1,\ldots,r_k$, we use $b^{(i)}$ to denote the characteristic vector \emph{under the random string $r_i$}, i.e., $b^{(i)}_j=1$ if and only if the LSH ${\cal A}_i$ discovers that $\|u_j-v\|\leq cr$. Note that it is completely possible that the vector $b$ has large support size, but some $b^{(i)}={\bf 0}_n$. 

Now, imagine we have an $O(n)$ time budget. Then our algorithm is essentially identical to reporting the noisy max: we first sum over all characteristic vectors: $b^{\rm sum}=\sum_{i=1}^k b^{(i)}$. Note that the $i$-th entry of $b^{\rm sum}$ counts the total number of successful data structures that report point $u_i$. Then, we sample $n$ independent exponential noise variables of scale $2/\epsilon_{\rm DP}$ (Theorem~\ref{thm:rpm}), and then add these noise variables to each entry of $b^{\rm sum}$. Then we simply report the maximum index of the noisy counts {\color{black} with a simple post-processing by directly computing the distance between the corresponding point and the query}. The algorithm is naturally $(\epsilon_{\DP},0)$-DP with respect to the random strings $r_1,\ldots,r_k$!

To prove the utility of the algorithm, i.e., that we can actually differentiate the signal from the noise, we require the following tail bound for exponential random variables:

\begin{lemma}
\label{lem:laplace_tail}
Let $Y\sim {\rm Exp}(b)$, then
\begin{align*}
    \Pr[Y\geq t\cdot b]\leq & ~ \exp(-t)
\end{align*}
\end{lemma}
Let us choose $t=C\ln n$ for some absolute constant $C>1$. Each noise variable is ${\rm Exp}(2/\epsilon_{\DP})$, so by Lemma~\ref{lem:laplace_tail} and a union bound over the $n$ noise variables, all of them are at most $(2/\epsilon_{\DP})\cdot C\ln n$ with probability at least $1-n^{1-C}$. This means that even if we choose $\epsilon_{\DP}$ to be $1/2$, the magnitude of the noise will be at most $4C\ln n=\Theta(\log n)$. Consider a simple case, where the query $v$ only has one point $u\in U$ with $\|u-v\|_2\leq r$ and $\|u-v\|_2\leq cr$. Then, for a successful data structure ${\cal A}_i$, it will have exactly one non-zero entry. Let us assume that we sample $l=O(\log^2 n)$ independent data structures to perform the mechanism, using $\epsilon_{\DP}=1/2$. Then by amplification via subsampling, the algorithm is $\frac{6l}{k}\cdot\epsilon_{\DP}$-DP. Performing an advanced composition over all $T$ queries, we conclude the algorithm is $(\frac{1}{100},\frac{1}{100\cdot\poly(n)})$-DP. We can then use the generalization property of DP to conclude that, with probability at least $1-1/\poly(n)$, at least a $3/4$-fraction of the sampled data structures succeed. 

Conditioning on these events, we argue that the noisy max is indeed the correct answer: we know that a large constant fraction of data structures succeed, therefore the corresponding point in $b^{\rm sum}$ has count $\Omega(\log^2 n)$. On the other hand, the exponential noise itself has magnitude at most $O(\log n)$. This means that adding the exponential noise, we can still differentiate the point. We formalize this argument with the following lemma.

\begin{lemma}
\label{lem:correct}
For all $t\in [T]$, the algorithm outputs a point $u_t$ that is a $c$-ANN for the adaptive query $v_t$ if $v_t$ satisfies Assumption~\ref{assumption:general}, with probability at least $1-\delta_{\rm fail}-1/\poly(n)$.
\end{lemma}

\begin{proof}
By Lemma~\ref{lem:acc_sample}, among the $l$ samples, at least a $\frac{3}{4}$-fraction of the data structures succeed, i.e., if there exists an $r$-near neighbor of $v_t$, then at least $\frac{3}{4}l$ data structures have their characteristic vectors with at least 1 non-zero entry. By the second hypothesis of Theorem~\ref{thm:main_general}, the support of each characteristic vector lies in $f_{v_t}(U)=B(v_t,cr)\cap U$, which has at most $s$ points by Assumption~\ref{assumption:general} (a dataset satisfying Assumption~\ref{assumption:ann_sparse} satisfies it with $s+1$ in place of $s$, which changes only constants). Since we set $l=O(s\log^2(n/\beta))$ with a large enough constant, by the pigeonhole principle, there exists at least one entry in $b^{\rm sum}$ with magnitude at least $\frac{3l}{4s}=\Omega(\log^2(n/\beta))$, which is larger than $4C\ln n$. Now, condition on the event that all $n$ exponential noise variables, of scale $2/\epsilon_{\DP}$, are at most $4C\ln n$, which holds with probability at least $1-n^{1-C}$. Thus, we know that there exists at least one entry of $b^{\rm sum}$ which has magnitude $\Omega(\log^2 n)$ and we can differentiate this entry from all the noise added: every index outside the support of $b^{\rm sum}$ has noisy count at most $4C\ln n$, which is below the noisy count of this entry since the noise is nonnegative. So the noisy max is an index in the support, which is a point of $B(v_t,cr)\cap U$ and passes the distance check.

{\color{black} Now, consider the case that there is no point $u\in U$ with $\|u-v_t\|_2\leq r$. Then the data structure is allowed to either return points that are $cr$-near neighbors of $v_t$, or return nothing. Since the base data structure either outputs a false negative (a $cr$-near neighbor) or outputs nothing, we analyze two cases:
\begin{itemize}
    \item Case 1: $\|b^{\rm sum}\|_\infty > 4C\ln n$, the bound on the noise. In this case, we know that the max index before adding the noise must correspond to a $cr$-near neighbor, since by the second hypothesis of Theorem~\ref{thm:main_general} every nonzero entry of $b^{\rm sum}$ is a point of $f_{v_t}(U)=B(v_t,cr)\cap U$. Hence, adding the noise and outputting the noisy max will not hide this large entry: every index outside the support of $b^{\rm sum}$ has noisy count at most $4C\ln n$, so the noisy max lies in the support, and so we can return it and satisfy the specification of an $c$-ANN data structure.
    \item Case 2: $\|b^{\rm sum}\|_\infty \leq 4C\ln n$. In this case, it must be the case there is no $r$-near neighbor. Here $\rpm$ may output any index, inside or outside the support of $b^{\rm sum}$. Correctness in this case comes only from the post-processing. In the post-processing phase, we ensure the algorithm outputs only a $cr$-near neighbor by checking the distance, that is, by testing whether $u_i\in f_{v_t}(U)=B(v_t,cr)\cap U$, which is the membership test assumed in Theorem~\ref{thm:main_general}. So the output is nothing or a point of $f_{v_t}(U)$, which item (3) of Theorem~\ref{thm:main_general} allows when $v_t\notin Y_U$.
\end{itemize}
This concludes the proof of correctness.
}

To compute the failure probability, note that a constant fraction of the data structures succeed with probability at least $1-\delta_{\rm fail}$, and all exponential random variables are at most $4C\ln n$ with probability at least $1-n^{1-C}$. A union bound concludes the bound on the failure probability.
\end{proof}

Of course, this algorithm is by no means efficient --- it takes $O(n)$ time to generate all $n$ noise variables per query, making the sublinear query time linear. In the next section, we will show it is sufficient to generate $s$ exponential variables for the support of $b^{\rm sum}$, and take the noisy max over these entries.

\subsection{Sparse Noise and Order Statistics: A Fast Algorithm}

{\color{black} We make the following observation for the \textsc{ReportOneSidedNoisyArgMax} algorithm: if $b^{\rm sum}={\bf 0}_n$, then the entry with the max noise must be the one contains $X_{(n)}$, the largest order statistics of the exponential distribution. If $b^{\rm sum}$ has $s$ nonzero entries, then there are two cases: if $X_{(n)}$ is among the nonzero entries, since those entries are all positive, the max entry must be within the $s$ nonzero entries; if $X_{(n)}$ is not among the nonzero entries, then the max entry is among the entry contains $X_{(n)}$, and the nonzero entry with the largest magnitude after adding noises. Note that the first event happens with probability $\frac{s}{n}$, and for this case, we generate $s-1$ i.i.d. noises until they do not exceed $X_{(n)}$ so that they obey the order statistics distribution. For the second case, we generate $s$ i.i.d. noises until they do not exceed $X_{(n)}$, randomly assign a zero entry to $X_{(n)}$, and output the max between $X_{(n)}$ and the largest nonzero noisy entry. Note that this algorithm is extremely efficient --- according to Fact~\ref{fact:generate_order}, $X_{(n)}$ can be generated in $O(1)$ time, and the above procedure runs in $O(s \log n)$ time with high probability, proportional to the support size. }

\begin{fact}
\label{fact:geometric_tail}
Let $Z$ be a geometric random variable with success probability $p\in (0, 1)$, then for any positive integer $k$, we have
\begin{align*}
    \Pr[Z>k] \leq & ~ \exp(-k\log(1/(1-p)))
\end{align*}
\end{fact}

\begin{proof}
Note that by definition,
\begin{align*}
    \Pr[Z>k] = & ~ (1-p)^k \\
    = & ~ \exp(k\log(1-p)) \\
    = & ~ \exp(-k\log(1/(1-p))). \qedhere
\end{align*}
\end{proof}

\begin{lemma}
Let ${\cal A}_1, {\cal A}_2$ be two algorithms, with the following behavior. Let $s>0$:
\begin{itemize}
    \item ${\cal A}_1$, at time $t$, it receives an $s$-sparse vector $u\in \mathbb{N}^n$ and let $U_t \subseteq [n]$ denote the support of $u$. We add a dense exponential noise vector with parameter $2/\epsilon_{\DP}$, denoted by $\eta \in \R^n$ and add this to $u \in \R^n$, and output index $i_*$ such that $i_* = \arg\max_{j \in [n]} (u + \eta)_j $;
    {\color{black}\item ${\cal A}_2$, at time $t$, it receives an $s$-sparse vector $u\in \mathbb{N}^n$ and let $U_t\subseteq [n]$ denote the support of $u$, with $s_t:=|U_t|\leq s$. We flip a biased coin with probability of being head $\frac{s_t}{n}$. If head, we generate a max noise $X$ from the distribution $X_{(n)}$ of parameter $2/\epsilon_{\DP}$, and then repeatedly generating $s_t-1$ i.i.d. exponential noises of parameter $2/\epsilon_{\DP}$ until none of them exceed $X$, and let $\eta_{1},\ldots,\eta_{s_t-1}$ denote these noises, form vector $\eta\in \R^n$ with the support being $U_t$, and the values being randomly assigned from $\eta_1,\ldots,\eta_{s_t-1}, X$. Then output $i_*=\arg\max_{j\in [n]}(u+\eta)_j$. If the coin flips tail, then generate $X$ from $X_{(n)}$, and generate $s_t$ i.i.d. exponential noises until none of them exceed $X$. Randomly assign $X$ to one of the non-support entry denoted by $i_1$, and let $\eta$ denote the noises for the support, output $i_1$ if $X>\max_{j\in [n]}(u+\eta)_j$ and output $\arg\max_{j\in [n]}(u+\eta)_j$ otherwise.}

\end{itemize}
Then, ${\cal A}_1$ and ${\cal A}_2$ have the same output distribution. Moreover, if $s\leq n/(2\ln n)$, then with probability at least $1-1/\poly(n)$, ${\cal A}_2$ runs in $O(s\log n)$ time. If $s>n/(2\ln n)$, then ${\cal A}_1$ runs in $O(n)=O(s\log n)$ time, and we run ${\cal A}_1$ directly.
\end{lemma}
\begin{proof}

{\color{black}  We will prove that ${\cal A}_2$ has exactly the same output distribution as ${\cal A}_1$. We do so by analyzing
an hypothetical algorithm ${\cal A}_2'$ that generates $X$ from $X_{(n)}$ first and assign it to a uniformly chosen entry $i_*$, then generate the noises for all other entries conditioning they are smaller than $X$. We will prove the distribution of the noise generated by ${\cal A}_2'$ is identical to ${\cal A}_1$, then it's easy to see that the output distribution of ${\cal A}_2$ is the same as ${\cal A}_1$.
}

{\color{black}
Suppose the noise $\eta_i$ is generated from distribution with PDF $f$ and CDF $F$, then the CDF of $X_{(n)}$ is  $G(x) = F(x) ^ n$, and  PDF 
 $g(x) = n F(x) ^ {n-1} f(x)$. Let the joint CDF of noise generated by ${\cal A}_1$ be $F_1(x_1, \ldots, x_n) = F(x_1)\cdot F(x_2) \cdot \ldots \cdot F(x_n)$, and the PDF noise generated by ${\cal A}_2'$ be $F_2(x_1, \ldots, x_n)$. Let the sorted array of $x_1, \dots, x_n$ be $x_1 ^ *, \dots, x_n^ *$, and $x_0^ * = -\infty$ we have :

\begin{align*}
    F_2(x_1, \dots,x_n) = & ~\frac 1n \sum_{i = 1} ^ n \int_0 ^ {x_i} g(x) \prod_{x_j < x} \frac{F(x_j)}{F(x)}~\mathrm{d}x
    \\ = &~ \sum_{i = 1} ^ n \sum_{j = 1} ^ i \int_{x_{j - 1} ^ *} ^ {x_j ^ *} f(x) F(x) ^ {n-1} \prod_{k = 1} ^ {j-1} \frac{F(x_k ^ *)}{F(x)}~\mathrm{d}x  \\ = & ~\sum_{j  =1} ^ n (n-j+1) \int_{x_{j-1} ^ *}^ {x_j ^ *} f(x)F(x) ^ {n-j} \prod_{k = 1} ^ {j-1} F(x_k^ *)~\mathrm{d}x \\ = &~ \sum_{j =1} ^ n \Big(\prod_{k = 1} ^ {j - 1} F(x_k ^ *) \Big) \cdot \Big(F(x_j ^ *) ^ {n-j+1} - F(x_{j -1} ^ *) ^ {n-j+1}\Big) \\ = & ~F(x_1 ^ *) \cdot \dots \cdot F(x_n ^ *) \\
    = &~ F_1(x_1, \dots, x_n)
\end{align*}

  Hence, the noise distributions of these two algorithms are identical. 
  }

{\color{black}
 To see ${\cal A}_2$ and ${\cal A}_2'$ have the same noise distribution, consider two cases. If the max noise $X$ is in the support (happens with probability $\frac{s_t}{n}$), then we know that the max index must be within the support. For the noise distribution, it's easy to see that as long as all $\eta_i$'s are smaller than $X$, then the noise distributions are identical. For the tail case, the reasoning is similar, except that our maximum must be over $X$ and the noisy entries in the support.

 Regarding the running time, if $s>n/(2\ln n)$, then ${\cal A}_1$ generates $n$ noises and takes an argmax in $O(n)$ time, and $n<2s\ln n$, so this is $O(s\log n)$. Now let $s\leq n/(2\ln n)$. We only need to examine the probability that the noises $\eta_1,\ldots,\eta_{s_t-1}$ or $\eta_1,\ldots,\eta_{s_t}$ do not exceed $X$. We without loss of generality prove for the case where we generate $s_t$ independent exponential noises. Let $\lambda$ be the rate, that is, the inverse of the parameter. If $X$ were redrawn in every round, the success probability of one round would be the probability that $Y_{(s_t)}\leq X_{(n)}$, where $Y_{(s_t)}$ is the max order statistics for another independent sequence of exponential noises. The CDF of $Y_{(s_t)}$ is $F_{Y_{(s_t)}}(y)=(1-e^{-\lambda y})^{s_t}$, and the PDF of $X_{(n)}$ is $f_{X_{(n)}}(x)=n\lambda e^{-\lambda x}(1-e^{-\lambda x})^{n-1}$, then we have
 \begin{align*}
     \Pr[Y_{(s_t)}\leq X_{(n)}] = & ~ \int_0^\infty \Pr[Y_{(s_t)}\leq x] f_{X_{(n)}}(x)~\d x \\
     = & ~ \int_{0}^\infty (1-e^{-\lambda x})^{s_t} n\lambda e^{-\lambda x}(1-e^{-\lambda x})^{n-1}~\d x \\
     = & ~ \int_0^\infty n\lambda e^{-\lambda x}(1-e^{-\lambda x})^{n+s_t-1}~\d x
 \end{align*}
 we do a change of variable $u=1-e^{-\lambda x}$, so $\d u=\lambda e^{-\lambda x} \d x$, and for $x=0$, $u=0$ and for $x=\infty$, $u=1$, therefore
 \begin{align*}
     \int_0^\infty n\lambda e^{-\lambda x}(1-e^{-\lambda x})^{n+s_t-1}~\d x = & ~ n \int_0^1 u^{s_t+n-1}~\d u \\
     = & ~ \frac{n}{n+s_t},
 \end{align*}
 similarly, for $s_t-1$ noises, the probability is $\frac{n}{n+s_t-1}$. These are the success probabilities of one round when $X$ is redrawn. Since $X$ is drawn once and kept, the rounds are independent only given $X$, so we condition on $X$. The variable $V:=F(X)^n$ is uniform on $(0,1)$, since $F^n$ is the CDF of $X_{(n)}$. Given $X$, one round succeeds with probability at least $F(X)^{s_t}\geq F(X)^{s}=V^{s/n}$. Fix a constant $C\geq 1$. With probability at least $1-n^{-C}$ we have $V\geq n^{-C}$, and then $V^{s/n}\geq e^{-Cs\ln n/n}\geq e^{-C/2}$, since $s\leq n/(2\ln n)$. Conditioned on this event, group the rounds into blocks of $b:=\lceil e^{C/2}\rceil$ consecutive rounds. Given $X$, a block contains a successful round with probability at least $1-(1-e^{-C/2})^{b}\geq 1-e^{-b e^{-C/2}}\geq 1-e^{-1}\geq \frac{1}{2}$. So the number $Z$ of blocks up to and including the first block with a successful round is a geometric random variable with success probability $p\geq \frac{1}{2}$, and by Fact~\ref{fact:geometric_tail}, we have that 
 \begin{align*}
     \Pr[Z>k] \leq & ~ \exp(-k\log(1/(1-p))) \leq 2^{-k},
 \end{align*}
 set $k=\lceil C\log_2 n\rceil$, so that this happens with probability at most $n^{-C}$. Hence, together with the event $V\geq n^{-C}$, with probability at least $1-2n^{-C}$ we only need to generate the noises for at most $bk=O(\log n)$ rounds, each in $O(s_t)$ time. Moreover, $X$ can be efficiently generated via inverse CDF sampling (Fact~\ref{fact:generate_order}), so the overall runtime for generating these noises is $O(s\log n)$, as desired. \qedhere
}

\end{proof}
\section{Proof of Theorem~\ref{thm:terminal} and Corollary~\ref{cor:terminal_lsh}}\label{sec:terminal_proof}

\paragraph{Conventions and parameters.}
We assume exact real arithmetic. Additions, multiplications, comparisons, square roots, and drawing a Gaussian or a Laplace variable each take $O(1)$ time. Norms are $\ell_2$ norms, and we use the notation $B(q,r)$, $(x,y)$ and $R(q)$ of the terminal embedding part of Section~\ref{sec:app}. For $q\in\R^d$, let $u^*(q)$ be the point of $U$ with the smallest index among those at distance $R(q)$ from $q$. Let $T_U=\{\frac{u-v}{\|u-v\|}:u,v\in U,\ u\neq v\}$. We take $C_1=32$ and $C_0=128$, and set
\[
\epsilon_0=\frac{\epsilon}{32},\quad L=\Big\lceil\log_2\max\Big\{1,\frac{64\Lambda}{\epsilon r_0}\Big\}\Big\rceil,\quad \beta=\frac{\delta_{\rm fail}}{4T(L+1)},\quad \Gamma=\ln\frac{400\,sT^2(L+1)}{\delta_{\rm fail}},
\]
\[
l=\lceil200\,s\Gamma\rceil,\qquad M=\Big\lceil600\,l\sqrt{2T\ln(200/\beta)}\Big\rceil,\qquad b=4s,\qquad \tau=1+4s\Gamma .
\]
Then $\epsilon_0\in(n^{-1/2},1/32)$, $C_0/\epsilon=4/\epsilon_0$ and $\beta<1/4$. Also $M\geq600\,l$, because $2T\ln(200/\beta)\geq2\ln800>1$. Finally $2/\beta\leq e^{\Gamma}$ and $4T/\delta_{\rm fail}\leq e^{\Gamma}$. The adversary may know $U$ and $S$, and it may even see every set released by the query algorithm below, but not the random strings $\zeta^{(i)}_j$, the sampled indices or the noise. Its queries need not be sparse in the sense of Assumption~\ref{assumption:terminal}.

\subsection{A Fixed Terminal Embedding}

\begin{lemma}[{\cite[Lemma~3.3]{cn21}}]\label{lem:terminal_chd}
If $S\in\R^{k\times d}$ has i.i.d.\ ${\cal N}(0,1/k)$ entries and $k=\lceil C\epsilon_0^{-2}\log n\rceil$ for a large enough absolute constant $C$, then $S$ is a convex hull distortion for $U$ with parameter $\epsilon_0$ with probability at least $1-n^{-10}$.
\end{lemma}
In~\cite{cn21} the vector $0$ is added to $T_U$, which only strengthens the conclusion, and we take their failure probability to be $n^{-10}$. So $k=O(\epsilon^{-2}\log n)$. When $S$ is such a convex hull distortion, every $a=u-u'$ with $u\neq u'\in U$ has $a/\|a\|\in T_U$, hence
\begin{equation}\label{eq:terminal_pair}
(1-\epsilon_0)\|a\|\leq\|Sa\|\leq(1+\epsilon_0)\|a\|,\qquad \big|\|Sa\|^2-\|a\|^2\big|\leq(2+\epsilon_0)\epsilon_0\|a\|^2\leq3\epsilon_0\|a\|^2 .
\end{equation}

\begin{lemma}[Existence, {\cite[proof of Lemma~3.5]{cn21}}]\label{lem:terminal_exist}
Let $S$ be a convex hull distortion for $U$ with parameter $\epsilon_0$, and let $q\in\R^d$, $u'\in U$ and $x=q-u'$. There is $w_0\in\R^k$ with $\|w_0\|\leq\|x\|$ and $|\langle w_0,S(u-u')\rangle-\langle x,u-u'\rangle|\leq\epsilon_0\|x\|\|u-u'\|$ for all $u\in U$.
\end{lemma}
\begin{proof}
If $x=0$ or $n=1$, take $w_0=0$. Otherwise let $W=\{\pm\frac{u-u'}{\|u-u'\|}:u\in U\setminus\{u'\}\}\subseteq T_U$. For $w\in B(0,\|x\|)$ and $\lambda$ in the probability simplex on $W$, let $\phi(w,\lambda)=\sum_{y\in W}\lambda_y(\langle x,y\rangle-\langle w,Sy\rangle)$. Both domains are convex and compact and $\phi$ is bilinear, so by von Neumann's minimax theorem $\min_w\max_\lambda\phi=\max_\lambda\min_w\phi$. For a fixed $\lambda$, let $v=\sum_y\lambda_yy\in{\rm conv}(T_U)$. Then $\min_w\phi(w,\lambda)=\langle x,v\rangle-\|x\|\|Sv\|\leq\|x\|(\|v\|-\|Sv\|)\leq\epsilon_0\|x\|$. Hence some $w_0$ with $\|w_0\|\leq\|x\|$ has $\langle x,y\rangle-\langle w_0,Sy\rangle\leq\epsilon_0\|x\|$ for all $y\in W$. Since $W=-W$, the bound holds in absolute value, and multiplying by $\|u-u'\|$ gives the claim.
\end{proof}

\begin{lemma}[Locality]\label{lem:terminal_local}
Let $S$ be a convex hull distortion for $U$ with parameter $\epsilon_0$. Let $q\in\R^d$ with $R=R(q)>0$, and let $u^*=u^*(q)$, $x=q-u^*$ and $a_u=u-u^*$ for $u\in U$. Let $w\in\R^k$ satisfy $\|w\|\leq R$ and
\begin{equation}\label{eq:terminal_constraint}
|\langle w,Sa_u\rangle-\langle x,a_u\rangle|\leq3\epsilon_0R\|a_u\|\quad\text{for every }u\in U\text{ with }0<\|a_u\|\leq R/\epsilon_0 ,
\end{equation}
and let $z=(Su^*+w,\sqrt{R^2-\|w\|^2})$. Then $\big|\|z-(Su,0)\|^2-\|q-u\|^2\big|\leq\epsilon\|q-u\|^2$, and so $(1-\epsilon)\|q-u\|\leq\|z-(Su,0)\|\leq(1+\epsilon)\|q-u\|$, for all $u\in U$.
\end{lemma}
\begin{proof}
Fix $u$ and write $a=a_u$. Expanding gives $\|z-(Su,0)\|^2=\|w-Sa\|^2+R^2-\|w\|^2=\|Sa\|^2-2\langle w,Sa\rangle+R^2$ and $\|q-u\|^2=\|x-a\|^2=R^2-2\langle x,a\rangle+\|a\|^2$. Their difference is $D=(\|Sa\|^2-\|a\|^2)-2(\langle w,Sa\rangle-\langle x,a\rangle)$. If $a=0$, then $D=0$. Otherwise, since $u^*$ is a nearest neighbor, $R\leq\|q-u\|$ and $\|a\|\leq\|q-u\|+R\leq2\|q-u\|$. By~\eqref{eq:terminal_pair}, the first term of $D$ is at most $3\epsilon_0\|a\|^2\leq12\epsilon_0\|q-u\|^2$ in absolute value. If $\|a\|\leq R/\epsilon_0$, then~\eqref{eq:terminal_constraint} gives $|\langle w,Sa\rangle-\langle x,a\rangle|\leq3\epsilon_0R\|a\|\leq6\epsilon_0\|q-u\|^2$. If $\|a\|>R/\epsilon_0$, then $R<\epsilon_0\|a\|$, and $|\langle w,Sa\rangle-\langle x,a\rangle|\leq\|w\|\|Sa\|+\|x\|\|a\|\leq(2+\epsilon_0)R\|a\|<(2+\epsilon_0)\epsilon_0\|a\|^2\leq9\epsilon_0\|q-u\|^2$. In both cases $|D|\leq12\epsilon_0\|q-u\|^2+18\epsilon_0\|q-u\|^2\leq32\epsilon_0\|q-u\|^2=\epsilon\|q-u\|^2$. The last claim follows from $\sqrt{1-\epsilon}\geq1-\epsilon$ and $\sqrt{1+\epsilon}\leq1+\epsilon$.
\end{proof}

\begin{lemma}[Ellipsoid method]\label{lem:terminal_ellipsoid}
Let $S$ be a convex hull distortion for $U$ with parameter $\epsilon_0$. Let $q$, $R>0$, $u^*$, $x$ and $a_u$ be as in Lemma~\ref{lem:terminal_local}, let $N\subseteq U\setminus\{u^*\}$, and let
\[
K=\{w\in\R^k:\ \|w\|\leq R,\ |\langle w,Sa_u\rangle-\langle x,a_u\rangle|\leq3\epsilon_0R\|a_u\|\ \text{for all }u\in N\}.
\]
Run the central-cut ellipsoid method from $E_0=B(0,R)$ as follows. At the center $w$ of the current ellipsoid, if $\|w\|>R$, cut with $\{y:\langle w,y\rangle\leq\langle w,w\rangle\}$. Otherwise check the constraints of $u\in N$ in index order. At the first violated one, cut with $\{y:\sigma\langle y,Sa_u\rangle\leq\sigma\langle w,Sa_u\rangle\}$, where $\sigma\in\{\pm1\}$ is the sign of $\langle w,Sa_u\rangle-\langle x,a_u\rangle$. If no constraint is violated, return $w$. If no point has been returned after $N_{\rm it}=\lfloor2k(k+1)\ln(4/\epsilon_0)\rfloor+1$ iterations, or if a cut would use $Sa_u=0$, stop without a point. Then one of the first $N_{\rm it}$ centers lies in $K$, so the method returns a point of $K$. Given $Su$ for all $u\in U$, and $\langle x,a_u\rangle$ and $\|a_u\|$ for $u\in N$, it runs in time $O(k^3(k+|N|)\log(2/\epsilon))$, whether or not it returns a point.
\end{lemma}
\begin{proof}
First, the cuts are valid. If $\|w\|>R$, every $y\in K$ has $\langle w,y\rangle\leq\|w\|R<\|w\|^2$. If the constraint of $u$ is violated at $w$ with sign $\sigma$, every $y\in K$ has $\sigma\langle y,Sa_u\rangle\leq\sigma\langle x,a_u\rangle+3\epsilon_0R\|a_u\|<\sigma\langle w,Sa_u\rangle$, and $Sa_u\neq0$ by~\eqref{eq:terminal_pair}. So each cut is a halfspace through the current center that contains $K$.

Second, $K$ contains a ball. Let $w_0$ be given by Lemma~\ref{lem:terminal_exist} with $u'=u^*$, so $\|x\|=R$, and let $w_1=(1-\epsilon_0/2)w_0$. For $y\in B(w_1,\epsilon_0R/4)$ we have $\|y\|\leq(1-\epsilon_0/2)R+\epsilon_0R/4<R$. For $u\in N$ and $a=a_u$,
\begin{align*}
|\langle y,Sa\rangle-\langle x,a\rangle|&\leq|\langle w_0,Sa\rangle-\langle x,a\rangle|+(\|w_1-w_0\|+\|y-w_1\|)\|Sa\|\\
&\leq\epsilon_0R\|a\|+\tfrac34\epsilon_0R(1+\epsilon_0)\|a\|<3\epsilon_0R\|a\| .
\end{align*}
So $B(w_1,\epsilon_0R/4)\subseteq K$.

Third, the method stops quickly. By~\cite[Theorem~8.1]{bt97}, the ellipsoid $E_{j+1}$ computed from $E_j$ and a halfspace $H_j$ through its center contains $E_j\cap H_j$ and has ${\rm vol}(E_{j+1})<e^{-1/(2(k+1))}{\rm vol}(E_j)$, since $k\geq2$. If the first $j$ centers lie outside $K$, induction gives $K\subseteq E_j$. Writing $V_k$ for the volume of the unit ball, this gives $(\epsilon_0R/4)^kV_k\leq{\rm vol}(K)<e^{-j/(2(k+1))}R^kV_k$, so $j<2k(k+1)\ln(4/\epsilon_0)$. Hence one of the first $N_{\rm it}$ centers lies in $K$.

Finally, the time. In one iteration, $\|w\|$ takes $O(k)$ time, and $\langle w,Sa_u\rangle=\langle w,Su\rangle-\langle w,Su^*\rangle$ takes $O(k)$ time for each $u\in N$. The ellipsoid update takes $O(k^2)$ time. Since at most $N_{\rm it}=O(k^2\log(2/\epsilon))$ iterations run, the total is $O(k^3(k+|N|)\log(2/\epsilon))$.
\end{proof}

\paragraph{The map $F$.}
If $R(q)=0$, let $F(q)=(Sq,0)$. If $R(q)>0$, let $u^*=u^*(q)$ and $N(q)=\{u\in U:0<\|u-u^*\|\leq R(q)/\epsilon_0\}$. Let $w(q)$ be the point returned by the method of Lemma~\ref{lem:terminal_ellipsoid} with $N=N(q)$, and let $F(q)=(Su^*+w(q),\sqrt{R(q)^2-\|w(q)\|^2})$. If $S$ is not a convex hull distortion with parameter $\epsilon_0$, the method may stop without returning a point, after $N_{\rm it}$ iterations or at a constraint with $Sa_u=0$. In that case let $w(q)=0$. This keeps $F$ defined and plays no role in the guarantees. In all cases $\|w(q)\|\leq R(q)$, and $F$ depends only on $U$ and $S$.

\begin{proof}[Proof of item 1 of Theorem~\ref{thm:terminal}]
By Lemma~\ref{lem:terminal_chd}, $S$ is a convex hull distortion with parameter $\epsilon_0=\epsilon/C_1$ with probability at least $1-n^{-10}$. Assume it is. For $u\in U$ we have $R(u)=0$, so $F(u)=(Su,0)$. If $R(q)=0$, then $q\in U$ and $\|F(q)-(Su,0)\|=\|S(q-u)\|$, which lies between $(1-\epsilon_0)\|q-u\|$ and $(1+\epsilon_0)\|q-u\|$ by~\eqref{eq:terminal_pair}. If $R(q)>0$, Lemma~\ref{lem:terminal_ellipsoid} puts $w(q)$ in $K$ for $N=N(q)$, which is exactly condition~\eqref{eq:terminal_constraint}. Lemma~\ref{lem:terminal_local} then gives the bounds.
\end{proof}

\subsection{The Data Structure}

\begin{lemma}[Thresholded release]\label{lem:terminal_release}
Let $s\geq1$, $\epsilon_1>0$, $b=2s/\epsilon_1$ and $\tau\geq1$. The input is a list $X^{(1)},\ldots,X^{(l)}$ of subsets of $U$, each of size at most $s$. Let $c_u=|\{a\in[l]:u\in X^{(a)}\}|$. For each $u$ with $c_u\geq1$, draw an independent $\nu_u\sim{\rm Lap}(b)$, with density $\frac1{2b}e^{-|x|/b}$, and output ${\cal O}=\{u:c_u\geq1,\ c_u+\nu_u\geq\tau\}$. Then:
\begin{enumerate}
\item if one entry $X^{(a)}$ is replaced by another set of size at most $s$, the mechanism is $(\epsilon_1,(1+e^{\epsilon_1})\delta_1)$-DP with $\delta_1=\frac s2e^{-(\tau-1)/b}$,
\item ${\cal O}\subseteq\bigcup_aX^{(a)}$ always,
\item if $c_u\geq\tau+y$ with $y\geq0$, then $\Pr[u\notin{\cal O}]\leq\frac12e^{-y/b}$.
\end{enumerate}
With an array of $n$ counters shared by all queries, it runs in time $O(ls)$.
\end{lemma}
\begin{proof}
Items 2 and 3 are immediate: $u\notin{\cal O}$ requires $\nu_u<\tau-c_u\leq-y$, and $\Pr[\nu<-y]=\frac12e^{-y/b}$.

For item 1, let $c'$ be the counts after the replacement and ${\cal O}'$ the corresponding output. Let $A=\supp(c)\cap\supp(c')$, $B=\supp(c)\setminus\supp(c')$ and $B'=\supp(c')\setminus\supp(c)$. Since $c-c'={\bf 1}_{X^{(a)}}-{\bf 1}_{X'^{(a)}}$, we have $B\subseteq X^{(a)}$ with $c_u=1$ on $B$, and $B'\subseteq X'^{(a)}$ with $c'_u=1$ on $B'$. Also $\sum_{u\in A}|c_u-c'_u|\leq|X^{(a)}|+|X'^{(a)}|\leq2s$. Let ${\cal O}_A={\cal O}\cap A$. A point of $B$ enters ${\cal O}$ only if $1+\nu_u\geq\tau$, so $\Pr[{\cal O}\cap B\neq\emptyset]\leq|B|\cdot\frac12e^{-(\tau-1)/b}\leq\delta_1$, and likewise $\Pr[{\cal O}'\cap B'\neq\emptyset]\leq\delta_1$. The set ${\cal O}_A$ is a function of $(c_u+\nu_u)_{u\in A}$, whose density differs from that of $(c'_u+\nu_u)_{u\in A}$ by a factor at most $e^{\sum_{u\in A}|c_u-c'_u|/b}\leq e^{\epsilon_1}$ at every point~\cite[Theorem~3.6]{dr14}. Hence, for every family $E$ of subsets of $U$,
\[
\Pr[{\cal O}\in E]\leq\Pr[{\cal O}_A\in E]+\delta_1\leq e^{\epsilon_1}\Pr[{\cal O}'_A\in E]+\delta_1\leq e^{\epsilon_1}(\Pr[{\cal O}'\in E]+\delta_1)+\delta_1 .
\]
For the time, counting the at most $ls$ reported points and drawing at most $ls$ noise variables take $O(ls)$ time.
\end{proof}

\paragraph{Initialization.}
\begin{itemize}
\item Draw $S$ as in Lemma~\ref{lem:terminal_chd}, and store $U$, $S$ and $SU$.
\item For each $i\in\{0,\ldots,L\}$ and $j\in[M]$, build an independent copy ${\cal A}^{(i)}_j$ of the oblivious $(c,r_i,s)$-reporting structure, with random string $\zeta^{(i)}_j$. Let $R^{(i)}=(\zeta^{(i)}_1,\ldots,\zeta^{(i)}_M)$. Write ${\cal A}^{(i)}_\zeta$ for the structure at radius $r_i$ built with string $\zeta$.
\end{itemize}

\paragraph{Query.} Given $v\in\R^d$, for $i=0,1,\ldots,L$:
\begin{enumerate}
\item[(a)] Sample $j_1,\ldots,j_l$ independently and uniformly from $[M]$, and let $X^{(a)}={\cal A}^{(i)}_{j_a}(v)$ for $a\in[l]$.
\item[(b)] Let $X_i$ be the output of Lemma~\ref{lem:terminal_release} on $X^{(1)},\ldots,X^{(l)}$ with $\epsilon_1=1/2$, $b=4s$ and $\tau=1+4s\Gamma$.
\item[(c)] Let $R'_i=\min_{u\in X_i}\|v-u\|$, read from the reported distances, with $R'_i=\infty$ if $X_i=\emptyset$. If $R'_i>\epsilon_0r_i/2$, continue with $i+1$. Otherwise go to (d) with this $i$.
\item[(d)] If $|X_i|>s$, output $(Su_1,0)$. Otherwise let $u'$ be the point of $X_i$ with the smallest index among those at distance $R'_i$ from $v$. If $R'_i=0$, output the stored $(Su',0)$. Otherwise let $N'=\{u\in X_i:0<\|u-u'\|\leq R'_i/\epsilon_0\}$, run the method of Lemma~\ref{lem:terminal_ellipsoid} with $R'_i$, $u'$, $x=v-u'$ and $N'$, setting $w=0$ if it stops without a point, and output $(Su'+w,\sqrt{R_i'^2-\|w\|^2})$.
\end{enumerate}
If the loop ends without reaching (d), output $(Su_1,0)$. Every decision and every output of the query is a function of $v$, $U$, $S$ and the released sets $X_i$.

\begin{lemma}[The path of a query]\label{lem:terminal_path}
Let $v$ be sparse (Assumption~\ref{assumption:terminal}), and let $i^*=\min\{i\geq0:r_i\geq2R(v)/\epsilon_0\}$. Whatever the random strings, samples and noise are, the following hold.
\begin{enumerate}
\item $i^*\leq L$ and $|B(v,cr_{i^*})\cap U|\leq s$.
\item The query on $v$ does not stop at any radius $i<i^*$.
\item If $B(v,r_{i^*})\cap U\subseteq X_{i^*}$, the query outputs $F(v)$.
\end{enumerate}
\end{lemma}
\begin{proof}
Write $R=R(v)$. For item 1, $r_L\geq64\Lambda/\epsilon=2\Lambda/\epsilon_0\geq2R/\epsilon_0$, so $i^*\leq L$. If $i^*\geq1$, then $r_{i^*}=2r_{i^*-1}<4R/\epsilon_0=C_0R/\epsilon$. So $r_{i^*}\leq\max\{r_0,C_0R/\epsilon\}$, and sparsity gives $|B(v,cr_{i^*})\cap U|\leq s$.

For item 2, by item 2 of Lemma~\ref{lem:terminal_release} and property (1) of Definition~\ref{def:reporting}, every released set satisfies $X_i\subseteq B(v,cr_i)\cap U$. Hence $R'_i\geq R$. For $i<i^*$ we have $\epsilon_0r_i/2<R\leq R'_i$, so the query continues.

For item 3, $|X_{i^*}|\leq|B(v,cr_{i^*})\cap U|\leq s$. Every point at distance $R$ from $v$ lies in $B(v,r_{i^*})$, because $R\leq\epsilon_0r_{i^*}/2$, and therefore in $X_{i^*}$. So $R'_{i^*}=R\leq\epsilon_0r_{i^*}/2$, the query reaches (d) at $i^*$, and $u'=u^*(v)$. If $R=0$, it outputs $(Sv,0)=F(v)$. Otherwise $N(v)\subseteq B(v,R+R/\epsilon_0)\subseteq B(v,r_{i^*})$, so $N(v)\subseteq X_{i^*}$ and $N'=N(v)$. The ellipsoid run is then the run that defines $w(v)$, and the output is $F(v)$.
\end{proof}

\begin{lemma}[Privacy]\label{lem:terminal_privacy}
Fix $i\in\{0,\ldots,L\}$ and an adversary. The sequence of all queries and all released sets $X_{t,j}$ ($t\in[T]$, $j\in\{0,\ldots,L\}$), viewed as a function of $R^{(i)}$, is $(\frac1{100},\frac\beta{100})$-DP with respect to $R^{(i)}$.
\end{lemma}
\begin{proof}
We prove the claim for every fixed value of the randomness that is independent of $R^{(i)}$: $S$, the strings $R^{(j)}$ for $j\neq i$, the samples and noise at radii $j\neq i$, and the adversary's coins. The claim for the mixture follows by averaging.

Each query visits radius $i$ at most once, so radius $i$ runs at most $T$ rounds, which we pad with trivial rounds to exactly $T$. Everything between two rounds is a function of earlier outputs and of the fixed randomness: the adversary, the rounds at other radii, and the decisions of the query algorithm. We absorb it into the adaptive adversary of Theorem~\ref{thm:ada_c}.

Consider one round at radius $i$. Replacing one row of the sampled list changes one entry $X^{(a)}$, so by Lemma~\ref{lem:terminal_release} the round is $(\frac12,(1+e^{1/2})\delta_1)$-DP with respect to the sampled list, with $\delta_1=\frac s2e^{-\Gamma}=\frac\beta{200T}$. The list is a sample with repetition of $l\leq M/2$ rows of $R^{(i)}$, so by Theorem~\ref{thm:amp} the round is $(\frac{3l}M,\delta')$-DP with respect to $R^{(i)}$. Here $\delta'=e^{3l/M}\frac{4l}M(1+e^{1/2})\delta_1\leq\delta_1$, since $M\geq600\,l$.

Theorem~\ref{thm:ada_c} with $T$ rounds, per-round $(\frac{3l}M,\delta_1)$ and $\delta_0=\beta/200$ gives privacy loss
\[
\sqrt{2T\ln(200/\beta)}\cdot\frac{3l}M+2T\Big(\frac{3l}M\Big)^2\leq\frac1{200}+\frac{9}{360000\ln(200/\beta)}<\frac1{100},
\]
and failure term $\frac\beta{200}+T\delta_1=\frac\beta{100}$.
\end{proof}

\begin{lemma}[Accuracy]\label{lem:terminal_accuracy}
Let $v_1,\ldots,v_T$ be queries as in item 2 of Theorem~\ref{thm:terminal}, not necessarily sparse. With probability at least $1-\frac34\delta_{\rm fail}$, for every $t\in[T]$ such that $v_t$ is sparse we have $B(v_t,r_{i^*_t})\cap U\subseteq X_{t,i^*_t}$, where $i^*_t$ is the radius $i^*$ of Lemma~\ref{lem:terminal_path} for $v_t$.
\end{lemma}
\begin{proof}
We use three steps, each failing with probability at most $\delta_{\rm fail}/4$.

\emph{Generalization.} For $t\in[T]$ and $i\in\{0,\ldots,L\}$, let $h_{t,i}(\zeta)={\bf 1}[B(v_t,r_i)\cap U\subseteq{\cal A}^{(i)}_\zeta(v_t)]$ if $v_t$ is sparse and $i=i^*_t$, and $h_{t,i}\equiv1$ otherwise. Whether $v_t$ is sparse, and $i^*_t$, depend only on $v_t$ and $U$, so the predicate $h_{t,i}$ is determined by $v_t$. So by Lemma~\ref{lem:terminal_privacy} and post-processing, the map $R^{(i)}\mapsto h_{t,i}$ is $(\frac1{100},\frac\beta{100})$-DP. For a fixed $v_t$, $\E_\zeta[h_{t,i}(\zeta)]\geq9/10$: this is trivial when $h_{t,i}\equiv1$, and otherwise it follows from property (2) of Definition~\ref{def:reporting} and item 1 of Lemma~\ref{lem:terminal_path}. The rows of $R^{(i)}$ are i.i.d. Theorem~\ref{thm:general} applies with $\epsilon=\frac1{100}$ and $\delta=\frac\beta{100}<\epsilon/4$, since $M\geq1.2\cdot10^5\,\Gamma\geq2\cdot10^4\ln(2/\beta)\geq10^4\log_2(2/\beta)$. So with probability at least $1-\beta$, $\frac1M\sum_{j}h_{t,i}(\zeta^{(i)}_j)>\frac9{10}-\frac1{10}=\frac45$. A union bound over the $T(L+1)$ pairs fails with probability at most $T(L+1)\beta=\delta_{\rm fail}/4$.

\emph{Sampling.} Fix $t$ such that $v_t$ is sparse, and assume the bound above holds for $(t,i^*_t)$. By item 2 of Lemma~\ref{lem:terminal_path}, the query on $v_t$ reaches radius $i^*_t$. Conditioned on everything before that round, its $l$ indices are fresh, and each one hits a good copy with probability more than $4/5$. By Hoeffding's bound, fewer than $3l/4$ of them are good with probability at most $e^{-l/200}\leq e^{-\Gamma}\leq\delta_{\rm fail}/(4T)$.

\emph{Noise.} If at least $3l/4$ sampled sets contain $B(v_t,r_{i^*_t})\cap U$, then each $u$ in it has $c_u\geq3l/4\geq150\,s\Gamma\geq\tau+y$, where $y=4s\ln(2sT/\delta_{\rm fail})\leq4s\Gamma$. By item 3 of Lemma~\ref{lem:terminal_release}, with fresh noise each such $u$ is missing from $X_{t,i^*_t}$ with probability at most $\frac12e^{-y/(4s)}=\frac{\delta_{\rm fail}}{4sT}$. There are at most $s$ such $u$, by item 1 of Lemma~\ref{lem:terminal_path}.

Union bounds over the $t$ with $v_t$ sparse, at most $T$ of them, in the last two steps give the total $\frac34\delta_{\rm fail}$.
\end{proof}

\begin{proof}[Proof of items 2--4 of Theorem~\ref{thm:terminal}]
Item 2 follows from Lemma~\ref{lem:terminal_accuracy}, item 3 of Lemma~\ref{lem:terminal_path} applied to each sparse $v_t$, and item 1 of the theorem, with total failure probability at most $\frac34\delta_{\rm fail}+n^{-10}$.

For item 3, which holds for every query, sparse or not, each visited radius costs $O(l)$ for sampling, $l\cdot{\cal T}_{\rm query}$ for the reporting queries, and $O(ls)$ for Lemma~\ref{lem:terminal_release} and for $R'_i$. There are at most $L+1$ radii. Step (d) costs $O(s)$ to find $u'$, and $O(sd)$ to form $N'$ and the values $\|a_u\|$ and $\langle x,a_u\rangle=\langle v-u',u-u'\rangle$. The ellipsoid method costs $O(k^3(k+s)\log(2/\epsilon))$, because $|N'|<|X_i|\leq s$, and the output costs $O(k)$ from the stored $Su'$.

For item 4, the $(L+1)M$ structures take $(L+1)M\cdot{\cal S}_{\rm space}$ space and $(L+1)M\cdot{\cal T}_{\rm prep}$ time. Storing $U$, $S$, $SU$ and two shared arrays of $n$ entries takes $O(nd+kd+nk)$ space. Computing $SU$ takes $O(nkd)$ time.
\end{proof}

\subsection{LSH Instantiation}

\begin{lemma}[LSH reporting structure]\label{lem:terminal_lsh}
Let $c>1$, $r>0$ and $s\geq1$, and let ${\cal H}$ be a family as in Corollary~\ref{cor:terminal_lsh} for this $r$, with $0<p_2<p_1<1$. Let $K_{\rm h}=\lceil\log_{1/p_2}n\rceil$ and $L_{\rm tab}=\lceil p_1^{-1}n^\rho\ln(20s)\rceil$. There is an oblivious $(c,r,s)$-reporting structure with ${\cal S}_{\rm space}=O(L_{\rm tab}K_{\rm h}(n+d))$, ${\cal T}_{\rm prep}=O(L_{\rm tab}nK_{\rm h}(d+\log n))$ and ${\cal T}_{\rm query}=O(L_{\rm tab}K_{\rm h}(d+\log n)+(L_{\rm tab}+s)d+L_{\rm tab}s)$. It also uses one array of $n$ entries shared by all structures.
\end{lemma}
\begin{proof}
\emph{Preprocessing.} Draw $g_j=(h_{j,1},\ldots,h_{j,K_{\rm h}})$ for $j\in[L_{\rm tab}]$ from ${\cal H}$. For each $j$, store the pairs $(g_j(u),u)$, $u\in U$, sorted by key, so that a bucket is found by binary search in $O(K_{\rm h}\log n)$ time.

\emph{Query $q$.} Use the shared array to mark each examined point as near or far, and undo the marks at the end. For $j=1,\ldots,L_{\rm tab}$, compute $g_j(q)$ and scan its bucket. A point seen for the first time costs one distance computation, and it is added to the output with its distance if $\|q-u\|\leq cr$. Every visit to a far point increments a counter. Stop when the output has $s$ points, or when the counter exceeds $20L_{\rm tab}$.

\emph{Property (1) and time.} The output lies in $B(q,cr)\cap U$ and has at most $s$ points. At most $s+20L_{\rm tab}+1$ distinct points are examined, which gives $O((s+L_{\rm tab})d)$ for distances. There are at most $sL_{\rm tab}$ visits to near points, since each of at most $s$ near points lies in one bucket per table, and at most $20L_{\rm tab}+1$ visits to far points. Hashing and lookups cost $O(L_{\rm tab}K_{\rm h}(d+\log n))$.

\emph{Property (2).} Fix $q$ with $|B(q,cr)\cap U|\leq s$. We have $p_1^{K_{\rm h}}\geq p_1n^{-\rho}$, so each $u\in B(q,r)$ shares no bucket with $q$ with probability at most $(1-p_1n^{-\rho})^{L_{\rm tab}}\leq\frac1{20s}$. A union bound over at most $s$ such points gives at most $1/20$. Each far point shares the bucket of $q$ in a given table with probability at most $p_2^{K_{\rm h}}\leq1/n$, so the expected number of far visits is at most $L_{\rm tab}$. By Markov's inequality, it exceeds $20L_{\rm tab}$ with probability at most $1/20$. Outside these two events, the query either reports $s$ points, which are then all of $B(q,cr)\cap U$, or it scans every bucket and reports every point of $B(q,r)\cap U$.
\end{proof}

\begin{proof}[Proof of Corollary~\ref{cor:terminal_lsh}]
At radius $r_i$, use Lemma~\ref{lem:terminal_lsh} with ${\cal H}_{r_i}$. Since $p_1,p_2$ depend only on $c$, ${\cal S}_{\rm space}=\wt O(n^{1+\rho}+n^\rho d)$, ${\cal T}_{\rm prep}=\wt O(n^{1+\rho}d)$ and ${\cal T}_{\rm query}=\wt O(n^\rho(d+s)+sd)$. Substitute into Theorem~\ref{thm:terminal} with $l=\wt O(s)$, $M=\wt O(\sqrt T\cdot s)$ and $k=O(\epsilon^{-2}\log n)$. Then $(L+1)l({\cal T}_{\rm query}+s)=\wt O((L+1)s(n^\rho(d+s)+sd))$, which also covers the term $sd$, and $k^3(k+s)\log(2/\epsilon)=\wt O(\epsilon^{-6}(\epsilon^{-2}+s))$. Also $(L+1)M\cdot{\cal S}_{\rm space}+O(nd+nk+kd)$ and $(L+1)M\cdot{\cal T}_{\rm prep}+O(nkd)$ give the stated space and preprocessing time.
\end{proof}

\begin{proof}[Proof of Remark~\ref{rem:terminal_const}]
Fix a good query $v$, and let $i^*_B=\min\{i\geq0:r_i\geq R(v)\}$. Then $i^*_B\leq L_B$, since $r_{L_B}\geq\Lambda$. Also $r_{i^*_B}\leq\max\{r_0,2R(v)\}$, so $|B(v,cr_{i^*_B})\cap U|\leq s$. For $i<i^*_B$, $R'_i\geq R(v)>r_i$, so the query continues. If $B(v,r_{i^*_B})\cap U\subseteq X_{i^*_B}$, then $R'_{i^*_B}=R(v)\leq r_{i^*_B}$, the query stops there, $u'$ is the nearest neighbor $u^*(v)$, and the output is $(Su^*(v),R(v))$. These three facts replace items 1--3 of Lemma~\ref{lem:terminal_path}. Every decision and output of the changed query is still a function of $v$, $U$, $S$ and the released sets. The proofs of Lemmas~\ref{lem:terminal_privacy} and~\ref{lem:terminal_accuracy} used only this property, Definition~\ref{def:reporting}, and items 1 and 2 of Lemma~\ref{lem:terminal_path} at sparse queries. So both lemmas hold verbatim with $L_B$, $i^*_B$ and good queries in place of $L$, $i^*$ and sparse queries. Together with Lemma~\ref{lem:terminal_chd}, the output is $(Su^*(v_t),R(v_t))$ for all $t\in[T]$ with $v_t$ good and $S$ is a convex hull distortion with parameter $\epsilon_0$, with probability at least $1-\frac34\delta_{\rm fail}-n^{-10}$. Each radius costs $O(l({\cal T}_{\rm query}+s))$, and the output costs $O(k)$ from the stored $Su'$.

For the distortion, fix $t$ with $v_t$ good, write $v=v_t$, $u^*=u^*(v)$ and $z=(Su^*,R(v))$, and note that $S$ is also a convex hull distortion with parameter $\epsilon\geq\epsilon_0$. For $u\in U$ and $a=u-u^*$ we have $\|z-(Su,0)\|^2=\|Sa\|^2+R(v)^2$, $\|a\|\leq2\|v-u\|$ and $R(v)\leq\|v-u\|$. On one side, $\|Sa\|^2+R(v)^2\leq4(1+\epsilon)^2\|v-u\|^2+\|v-u\|^2\leq(5+12\epsilon)\|v-u\|^2\leq(\sqrt5+3\epsilon)^2\|v-u\|^2$. On the other, $\|Sa\|^2+R(v)^2\geq(1-\epsilon)^2(\|a\|^2+R(v)^2)\geq\frac{(1-\epsilon)^2}2(\|a\|+R(v))^2\geq\frac{(1-\epsilon)^2}2\|v-u\|^2$.
\end{proof}

\section{Deletions in Adaptive \texorpdfstring{$\ell_2$}{} and \texorpdfstring{$\ell_p$}{} LSH}
\label{sec:batch}

In this section, we develop a fast deletion procedure for decremental $\ell_2$ and $\ell_p$ LSH, $p\in(0,2]$, against an adaptive adversary. This is crucial for applications such as online weighted matching (see Section~\ref{sec:intro:matching}).

\subsection{Low-Dimensional \texorpdfstring{$\ell_2$}{} LSH with Deletions}

We start by reviewing the algorithm for low-dimensional $\ell_2$ LSH, against an \emph{oblivious adversary}. Throughout, we let $\epsilon, \delta\in (0,1)$ denote the precision and failure probability of the Johnson-Lindenstrauss transform. For a matrix $S\in\R^{m\times d}$, we say that JL holds for a query $v\in\R^d$ if $(1-\epsilon)\|u-v\|_2\leq \|Su-Sv\|_2\leq (1+\epsilon)\|u-v\|_2$ for all $u\in U$. We let $c'=c(1-\epsilon)/(1+\epsilon)$.
\paragraph{Initialization.}
\begin{itemize}
    \item Prepare one Johnson-Lindenstrauss transform~\cite{jl84} $S\in\R^{m\times d}$ with $m=O(\epsilon^{-2}\log (n/\delta))$. For a fixed query $v$, JL holds for $v$ with probability at least $1-\delta$.
    \item Compute $SU$, and store $Su$ for every $u\in U$.
    \item Initialize a $(c',(1+\epsilon)r)$-$\ell_2$ LSH data structure on $SU$.
\end{itemize}
\paragraph{Query.}
\begin{itemize}
    \item Given query point $v\in \R^d$, first compute $Sv$.
    \item Query the LSH with $Sv$, output the corresponding point.
\end{itemize}
\paragraph{Deletion.}
\begin{itemize}
    \item Given a point $u\in U$ to-be-deleted, retrieve the stored sketch $Su$.
    \item Locate $Su$ in the hash buckets of the LSH, remove $Su$.
\end{itemize}
We need $c'>1$, which holds if and only if $c(1-\epsilon)>1+\epsilon$, that is, $\epsilon<(c-1)/(c+1)$. The choice of $c'$ gives
\begin{align*}
    c'\cdot\frac{1+\epsilon}{1-\epsilon} = c\cdot \frac{1-\epsilon}{1+\epsilon}\cdot\frac{1+\epsilon}{1-\epsilon} = c.
\end{align*}
Suppose JL holds for $v$. If $\|u-v\|_2\leq r$, then $\|Su-Sv\|_2\leq (1+\epsilon)r$, so $Su$ is a $(1+\epsilon)r$-near neighbor of $Sv$. Every point $u$ that the LSH reports has $\|Su-Sv\|_2\leq c'(1+\epsilon)r$, and hence $\|u-v\|_2\leq \|Su-Sv\|_2/(1-\epsilon)\leq c'(1+\epsilon)r/(1-\epsilon)=cr$. So the structure solves the $(c,r)$-ANN problem for $v$ whenever JL holds for $v$ and the LSH succeeds on $Sv$.

\paragraph{Parameters.} We set $\epsilon=1/\log n$, with $\log$ to base $2$, and assume $\log n\geq 2$ and $\log n>(c+1)/(c-1)$. Then $\epsilon\leq 1/2$ and $\epsilon<(c-1)/(c+1)$, so $c'>1$. We use the Euclidean LSH of~\cite{ai06}, whose exponent for approximation ratio $c$ is $\rho(c)=1/c^2+o(1)$. Since $(1+\epsilon)^2-(1-\epsilon)^2=4\epsilon$ and $(1-\epsilon)^2\geq 1/4$,
\begin{align*}
    \frac{1}{c'^2}-\frac{1}{c^2} = \frac{1}{c^2}\cdot\left(\frac{(1+\epsilon)^2}{(1-\epsilon)^2}-1\right) = \frac{4\epsilon}{c^2(1-\epsilon)^2}\leq \frac{16\epsilon}{c^2}.
\end{align*}
Hence $n^{1/c'^2}\leq n^{1/c^2}\cdot n^{16/(c^2\log n)}=2^{16/c^2}\cdot n^{1/c^2}\leq 2^{16}\cdot n^{1/c^2}$, and so $n^{\rho(c')}=O(n^{\rho(c)})$. We write $\rho=\rho(c)$ as in the rest of the paper. The LSH of~\cite{ai06} succeeds with constant probability for a fixed query, and we take $O(1)$ independent repetitions of it inside each structure, so that it succeeds with probability at least $19/20$. We set $\delta=1/20$, so the sketch dimension is $m=O(\epsilon^{-2}\log(n/\delta))=O(\log^3 n)$. By a union bound, for a fixed query $v$, JL holds for $v$ and the LSH succeeds on $Sv$ with probability at least $1-1/20-1/20=9/10$, which is the success probability that Theorem~\ref{thm:main_general} and Section~\ref{sec:dp} require of each structure.

\paragraph{Cost of a deletion.} Every point to be deleted is a point of $U$, and the initialization computes $Su$ for every $u\in U$. Storing these sketches takes $mn$ words, so a deletion needs no multiplication by $S$, which would take $md=O(\epsilon^{-2}d\log(n/\delta))$ time. It only removes $Su$ from its bucket in each of the $O(n^\rho)$ hash tables. If we also store the position of $Su$ in each table, this takes $O(n^\rho)$ time. In order to make such data structure adaptive, we need to create $\wt O(\sqrt T\cdot s)$ independent copies, and deleting a point from all of them takes time (up to polylogarithmic factors)
\begin{align*}
    \sqrt{T}\cdot s\cdot n^\rho,
\end{align*}
with no term that depends on $d$.

\subsection{Adaptive \texorpdfstring{$\ell_2$}{} LSH with Deletions through Stored Sketches}
We now combine $k$ copies of the structure above with the algorithm of Section~\ref{sec:dp}. Every copy deletes a point as soon as the point is output, so every copy is always up to date.
\paragraph{Initialization.}
\begin{itemize}
    \item Prepare $k$ JL transforms~\cite{jl84} $S_1,\ldots,S_k: \R^d\rightarrow \R^m$ where $m=O(\epsilon^{-2}\log(n/\delta))$. Here $k=\wt O(\sqrt T\cdot s)$ is set as in Section~\ref{sec:dp}, and $\epsilon$ and $\delta$ are as above.
    \item Set ${\cal S}=\begin{bmatrix}
        S_1 \\
        S_2 \\
        \vdots \\
        S_k
    \end{bmatrix}\in \R^{mk\times d}$.
    \item Compute ${\cal S}U$, and store $S_iu$ for every $i\in [k]$ and $u\in U$.
    \item Initialize $k$ $(c',(1+\epsilon)r)$-$\ell_2$ LSH's on $S_1U, S_2U,\ldots,S_kU$, each with the $O(1)$ repetitions above. Each bucket is a doubly linked list, and for every $i\in [k]$ and $u\in U$ we store a pointer to the list node of $S_iu$ in each hash table of the $i$-th LSH.
\end{itemize}
\paragraph{Query.}
\begin{itemize}
    \item Receive query point $v_t\in \R^d$.
    \item Sample $l=O(s\log^2 n)$ data structures, denote the sampled JLs as $S_{(1)},\ldots,S_{(l)}$.
    \item Batch compute $\begin{bmatrix}
        S_{(1)} \\
        \vdots \\
        S_{(l)}
    \end{bmatrix}v_t\in \R^{lm}$.
    \item Query corresponding LSH's with $S_{(1)}v_t,\ldots,S_{(l)}v_t$. For each $j\in [l]$, let $u$ be the first point reported by the $j$-th sampled LSH, and compute $\|u-v_t\|_2$ in $O(d)$ time. Let the $j$-th characteristic vector be the indicator vector of $u$ if $\|u-v_t\|_2\leq cr$, and ${\bf 0}_n$ otherwise, including when the LSH reports nothing.
    \item Sum these characteristic vectors into the count vector, apply sparse~\rpm~to the count vector and compute the order statistics. Let $i$ be the noisy max index.
    \item Output $u_i$ if $u_i$ has not been deleted and $\|u_i-v_t\|_2\leq cr$, and NULL otherwise.
\end{itemize}
\paragraph{Deletion.}
\begin{itemize}
    \item Given the output $u_t$ of $v_t$, for every $i\in [k]$, remove $S_iu_t$ from the $i$-th LSH, using the stored sketch $S_iu_t$ and the stored pointers to its list nodes. Mark $u_t$ as deleted, using one bit per point of $U$.
\end{itemize}
The distance checks in the query step keep every characteristic vector supported on $B(v_t,cr)\cap U$, which has at most $s+1$ points by Assumption~\ref{assumption:ann_sparse}, so the analysis of Section~\ref{sec:dp} applies. The deletions do not affect privacy. The deleted points are outputs, so the state of each structure is a function of its random string and the public transcript.

\begin{theorem}
\label{thm:batch}
Let $U\subseteq \R^d$ satisfy Assumption~\ref{assumption:ann_sparse} for the $\ell_2$ norm, assume $\log n\geq 2$ and $\log n>(c+1)/(c-1)$, and let $\{v_1,\ldots,v_T\}\subseteq \R^d$ be a sequence of adaptive queries. There exists a randomized data structure with the following guarantees:
\begin{itemize}
    \item It preprocesses $U$ in time $\wt O(\Tmat(\sqrt T\cdot s, d, n)+\sqrt T\cdot s\cdot n^{1+\rho})$;
    \item It uses space $\wt O(\sqrt T\cdot s\cdot (n^{1+\rho}+d)+nd)$;
    \item Given a query $v_t$ for $t\in [T]$, let $U'\subseteq U$ be the set of points not deleted before time $t$. It returns a point $u\in U'$ with $\|v_t-u\|_2\leq cr$ if $B(v_t, r)\cap U'\neq \emptyset$, and otherwise returns NULL or a point $u\in U'$ with $\|v_t-u\|_2\leq cr$, with probability at least $1-1/\poly(n)$. This procedure takes time $\wt O(s\cdot (d+n^\rho))$, with probability at least $1-1/\poly(n)$;
    \item If $u_t$ is the output of $v_t$, it deletes $u_t$ in worst-case time $\wt O(\sqrt T\cdot s\cdot n^\rho)$.
\end{itemize}
\end{theorem}

\begin{proof}
\emph{Preprocessing and space.} The product ${\cal S}U$ multiplies an $mk\times d$ matrix with a $d\times n$ matrix. Since $m=O(\log^3 n)$ and $k=\wt O(\sqrt T\cdot s)$, splitting the rows of ${\cal S}$ into $\wt O(1)$ blocks of at most $\sqrt T\cdot s$ rows shows that this product takes $\wt O(\Tmat(\sqrt T\cdot s, d, n))$ time. Building the $k$ LSH's on $m$-dimensional points, together with the pointers to the list nodes, takes $\wt O(k\cdot n^{1+\rho(c')})=\wt O(\sqrt T\cdot s\cdot n^{1+\rho})$ time, since $n^{\rho(c')}=O(n^{\rho})$. The $k$ LSH's and the pointers use $\wt O(\sqrt T\cdot s\cdot n^{1+\rho})$ words, the stored sketches use $mkn=\wt O(\sqrt T\cdot s\cdot n)$ words, the $k$ JL matrices, which the query step uses, take $mkd=\wt O(\sqrt T\cdot s\cdot d)$ words, $U$ uses $nd$ words, and the deletion marks use $n$ bits.

\emph{Query time.} Computing $S_{(1)}v_t,\ldots,S_{(l)}v_t$ takes $O(lmd)=\wt O(s\cdot d)$ time. Each LSH query takes $\wt O(n^\rho)$ time in dimension $m$, and each distance check takes $O(d)$ time, so the $l$ queries and checks take $\wt O(s\cdot (d+n^\rho))$ time. The sparse noisy max of Section~\ref{sec:dp} takes $\wt O(s)$ time with probability at least $1-1/\poly(n)$. Since every structure is up to date, the query has no other step.

\emph{Correctness.} Fix $t\in [T]$. The $i$-th LSH holds exactly the points of $S_iU'$, in the buckets given by its hash functions, so it is the LSH built with the same hash functions on $S_iU'$. Say that the $i$-th structure succeeds on $v_t$ if JL holds for $v_t$ under $S_i$ and its LSH succeeds on $S_iv_t$ over $S_iU'$. For a fixed query and a fixed set $U'$, this happens with probability at least $9/10$ by the choice of parameters. By the calculation in the previous subsection, every structure on which JL holds for $v_t$ reports only points of $B(v_t,cr)$, and if its LSH also succeeds, it reports a point whenever $B(v_t,r)\cap U'$ is nonempty. Its first reported point then passes the distance check, so its characteristic vector is nonzero. Because of the distance checks, every characteristic vector is supported on $B(v_t,cr)\cap U'$. The set $U'$ is a subset of $U$, so it satisfies Assumption~\ref{assumption:ann_sparse} with the same parameter $s$, and hence $|B(v_t,cr)\cap U'|\leq s+1$. Whether the $i$-th structure succeeds on $v_t$ is determined by its random string $r_i$, the query $v_t$ and the deleted points, which are outputs and hence part of the transcript. So Lemmas~\ref{lem:acc_global}, \ref{lem:acc_sample} and~\ref{lem:correct} apply with $U'$ in place of $U$. With $\delta_{\rm fail}=1/\poly(n)$, they show that, with probability at least $1-1/\poly(n)$, the noisy max index is a point of $B(v_t,cr)\cap U'$ whenever $B(v_t,r)\cap U'\neq\emptyset$. The noisy max index ranges over all of $U$, so when $B(v_t,r)\cap U'=\emptyset$ it may be a deleted point. The final check tests both that $u_i$ has not been deleted and that $\|u_i-v_t\|_2\leq cr$. The set of deleted points is part of the transcript, so this check is post-processing and does not affect privacy. In every case, it ensures that the output is NULL or a point of $B(v_t,cr)\cap U'$.

\emph{Deletion time.} For each $i\in [k]$, the deletion removes $S_iu_t$ from each of the $O(n^\rho)$ hash tables of the $i$-th LSH (over its $O(1)$ repetitions), using the stored pointer to its list node, in $O(1)$ time per table. This takes $O(n^\rho)$ time per structure, and $\wt O(\sqrt T\cdot s\cdot n^\rho)$ time over all $k=\wt O(\sqrt T\cdot s)$ structures, in the worst case. Marking $u_t$ as deleted takes $O(1)$ time. No Johnson-Lindenstrauss product is computed, so there is no term that depends on $d$.
\end{proof}

\subsection{Deletions in Adaptive \texorpdfstring{$\ell_p$}{} LSH}

The construction of Theorem~\ref{thm:batch} maps $U$ into $\ell_2^m$ with a Johnson-Lindenstrauss transform and builds an $\ell_2$ LSH there. For $p\in(0,2)$ this step has no direct analogue: a $p$-stable sketch~\cite{i06} estimates $\|x\|_p$ by a median of the entries $|(Sx)_i|$, and this median is not a norm of $Sx$. We do not need to preserve $\ell_p$ distances, however. It suffices to keep distances at most $r$ apart from distances above $cr$. We map each point to $D=\wt O(1)$ random cosine features whose frequencies are $p$-stable vectors. These are the random features of~\cite{rr07} for the kernel $e^{-(\|x-y\|_p/\sigma)^p}$, with a scale $\sigma$ fixed below. Squared Euclidean distances between features concentrate around an increasing function of the $\ell_p$ distance, and we run the Euclidean LSH of~\cite{ai06} on the features. The rest of the algorithm is that of Theorem~\ref{thm:batch}.

\paragraph{$p$-stable variables.} Let $Z$ be a standard symmetric $p$-stable random variable, that is, $\E[e^{\mathrm{i}uZ}]=e^{-|u|^p}$ for all $u\in\R$. It exists for every $p\in(0,2]$, and a copy of $Z$ can be drawn from two independent uniform variables~\cite{i06}. If $Z_1,\ldots,Z_d$ are independent copies of $Z$ and $w\in\R^d$, then $\sum_{a=1}^d w_aZ_a$ has the law of $\|w\|_pZ$~\cite{i06,diim04}. In particular, $\E[\cos(tZ)]=\mathrm{Re}\,\E[e^{\mathrm{i}tZ}]=e^{-|t|^p}$ for every $t\in\R$. We count arithmetic operations as in the rest of the paper, and we assume that drawing a copy of $Z$, evaluating a cosine and evaluating a power $y^a$ with $y>0$ and $a\in\R$ each take $O(1)$ time.

\paragraph{Parameters.} Logarithms are to base $2$, and $\ln$ is the natural logarithm. Let $\epsilon=1/\log n$. Then $\epsilon\leq 1/6$, and $\epsilon<(c^p-1)/3$ because $\log n>3/(c^p-1)$. Set
\begin{align*}
x_0=\frac{2\epsilon}{c^{2p}},\qquad \sigma=r\,x_0^{-1/p},\qquad \tau=2\epsilon x_0=\frac{4\epsilon^2}{c^{2p}},\qquad D=\left\lceil\frac{384\ln n}{\tau^2}\right\rceil=O(c^{4p}\log^5 n),
\end{align*}
so that $(r/\sigma)^p=x_0$ and $(cr/\sigma)^p=c^px_0$. Let
\begin{align*}
\psi(\Delta)=2\left(1-e^{-(\Delta/\sigma)^p}\right)\ \ (\Delta\geq 0),\qquad R_1^2=2(1+\epsilon)x_0,\qquad R_2^2=(2c^p-4\epsilon)x_0,\qquad c''=R_2/R_1.
\end{align*}

\paragraph{Random features.} Draw a matrix $G\in\R^{D\times d}$ whose entries are independent copies of $Z$, and let $g_1,\ldots,g_D$ be its rows. Independently, draw $\xi_1,\ldots,\xi_D$ independent and uniform on $[0,2\pi)$. Define $\phi:\R^d\to\R^D$ by
\begin{align*}
\phi(x)_j=\sqrt{2/D}\,\cos\left(\langle g_j,x\rangle/\sigma+\xi_j\right),\qquad j\in[D].
\end{align*}
Computing $Gx$ takes $O(Dd)$ time, and $\phi(x)$ then takes $O(D)$ time.

\begin{lemma}\label{lem:lp_features}
Fix $x,y\in\R^d$ and let $\Delta=\|x-y\|_p$. Then $\E\|\phi(x)-\phi(y)\|_2^2=\psi(\Delta)$, and
\begin{align*}
\Pr\left[\,\left|\|\phi(x)-\phi(y)\|_2^2-\psi(\Delta)\right|\geq\tau\right]\leq n^{-11}.
\end{align*}
\end{lemma}
\begin{proof}
Fix $j\in[D]$, and let $A=\langle g_j,x\rangle/\sigma+\xi_j$, $B=\langle g_j,y\rangle/\sigma+\xi_j$ and $\gamma=A-B=\langle g_j,x-y\rangle/\sigma$, which does not depend on $\xi_j$. Then $(\cos A-\cos B)^2=4\sin^2\frac{A+B}{2}\cdot\sin^2\frac{\gamma}{2}$. Given $g_j$, the angle $\frac{A+B}{2}=\langle g_j,x+y\rangle/(2\sigma)+\xi_j$ is uniform modulo $2\pi$, so $\sin^2\frac{A+B}{2}$ has conditional mean $1/2$, and
\begin{align*}
\E\left[(\cos A-\cos B)^2\mid g_j\right]=2\sin^2\frac{\gamma}{2}=1-\cos\gamma.
\end{align*}
By $p$-stability, $\gamma$ has the law of $(\Delta/\sigma)Z$, so $\E[\cos\gamma]=e^{-(\Delta/\sigma)^p}$. Summing over $j\in[D]$ with weight $2/D$ gives $\E\|\phi(x)-\phi(y)\|_2^2=\psi(\Delta)$. The $D$ summands $\frac{2}{D}(\cos A-\cos B)^2$ are independent over $j$ and lie in $[0,8/D]$. By Hoeffding's inequality, the deviation probability is at most
\begin{align*}
2\exp\left(-\frac{2\tau^2}{D\cdot(8/D)^2}\right)=2\exp\left(-\frac{\tau^2D}{32}\right)\leq 2\exp(-12\ln n)\leq n^{-11}.
\end{align*}
\end{proof}

\begin{lemma}\label{lem:lp_separation}
The function $\psi$ is strictly increasing. Moreover, $\psi(r)+\tau\leq R_1^2$ and $\psi(cr)-\tau\geq R_2^2$. Finally, $1<c''<c^{p/2}$ and $1/c''^2\leq 1/c^p+6/\log n$.
\end{lemma}
\begin{proof}
The function $\psi$ is strictly increasing because $\Delta\mapsto(\Delta/\sigma)^p$ is. Using $y-y^2/2\leq 1-e^{-y}\leq y$ for $y\geq 0$, we get
\begin{align*}
\psi(r)+\tau&\leq 2x_0+2\epsilon x_0=R_1^2,\\
\psi(cr)-\tau&\geq 2c^px_0-c^{2p}x_0^2-2\epsilon x_0=(2c^p-4\epsilon)x_0=R_2^2,
\end{align*}
where the last step uses $c^{2p}x_0=2\epsilon$. Hence $c''^2=(c^p-2\epsilon)/(1+\epsilon)$. This is less than $c^p$, and it is more than $1$ because $\epsilon<(c^p-1)/3$. Since $c^p\geq 1$, we have $c''^2\geq c^p(1-2\epsilon)/(1+\epsilon)\geq c^p(1-2\epsilon)(1-\epsilon)\geq c^p(1-3\epsilon)$. As $3\epsilon\leq 1/2$ and $1/(1-x)\leq 1+2x$ for $x\leq 1/2$, we get $1/c''^2\leq(1+6\epsilon)/c^p\leq 1/c^p+6/\log n$.
\end{proof}

\paragraph{Exponent.} For a ratio $c_0\in(1,c]$, the Euclidean LSH of~\cite{ai06} has exponent at most $1/c_0^2+f(n)$. Here $f(n)=o(1)$ can be taken the same for all ratios in $(1,c]$. In the analysis of~\cite{ai06} (see~\cite[Section~3.2]{a09}), let $P_1$ and $P_2$ be the probabilities that a hash function of the ball-partitioning family maps two points at distance at most $1$, respectively at least $c_0$, to the same bucket. The bound on $\log(1/P_1)/\log(1/P_2)$ depends on the ratio $c_0$ only through an upper bound on $c_0$, and the $n^{o(1)}$ cost of evaluating a hash function does not depend on $c_0$. We fix such an $f$, so that for every $c_0\in(1,c]$ the LSH on $n$ points in $\R^D$ has query time $D\cdot n^{1/c_0^2+f(n)}$ and preprocessing time and space $D\cdot n^{1+1/c_0^2+f(n)}$. This $f$ is the one in the definition of $\rho$ in Theorem~\ref{thm:lp_delete} below. By Lemma~\ref{lem:lp_separation}, $c''\in(1,c^{p/2})\subseteq(1,c]$, and
\begin{align*}
n^{1/c''^2+f(n)}\leq n^{6/\log n}\cdot n^{1/c^p+f(n)}=2^6\cdot n^\rho.
\end{align*}

\paragraph{One structure against an oblivious adversary.} The random string of one structure consists of two parts. The first is $(G,\xi)$. The second is the randomness of a $(c'',R_1)$ Euclidean LSH of~\cite{ai06} built on $\phi(U)\subseteq\R^D$. The LSH uses $O(1)$ independent repetitions, so that it succeeds with probability at least $19/20$ for a fixed query and a fixed point set. Success means that if the point set has a point within distance $R_1$ of the query, the LSH reports a point of the set within distance $c''R_1=R_2$. The LSH reports only candidates it has checked to be within $R_2$. Each point is stored in exactly one bucket of each hash table~\cite{ai06}. Each bucket is a doubly linked list, and for every $u\in U$ we store a pointer to the list node of $\phi(u)$ in each table.

\begin{lemma}\label{lem:lp_oblivious}
Fix $v\in\R^d$ and $U'\subseteq U$, and let the LSH hold exactly $\phi(U')$. With probability at least $9/10$, the following hold:
\begin{itemize}
    \item[(a)] every $u\in U$ with $\|\phi(u)-\phi(v)\|_2\leq R_2$ has $\|u-v\|_p\leq cr$,
    \item[(b)] if $B(v,r)\cap U'\neq\emptyset$, the LSH reports a point for the query $\phi(v)$.
\end{itemize}
\end{lemma}
\begin{proof}
Let $\Gamma$ be the event that $\left|\|\phi(u)-\phi(v)\|_2^2-\psi(\|u-v\|_p)\right|<\tau$ for all $u\in U$. By Lemma~\ref{lem:lp_features} and a union bound over the $n$ points, $\Pr[\Gamma]\geq 1-n^{-10}$. Assume $\Gamma$ holds, and use Lemma~\ref{lem:lp_separation}. If $\|u-v\|_p>cr$, then $\|\phi(u)-\phi(v)\|_2^2>\psi(\|u-v\|_p)-\tau>\psi(cr)-\tau\geq R_2^2$, which gives (a). If $\|u-v\|_p\leq r$, then $\|\phi(u)-\phi(v)\|_2^2<\psi(r)+\tau\leq R_1^2$. The randomness of the LSH is independent of $(G,\xi)$, and given $(G,\xi)$ the query $\phi(v)$ and the point set $\phi(U')$ are fixed. So, with conditional probability at least $19/20$, the LSH reports a point whenever $\phi(U')$ has a point within $R_1$ of $\phi(v)$. On $\Gamma$ this is the case when $B(v,r)\cap U'\neq\emptyset$. The failure probability is at most $n^{-10}+1/20\leq 1/10$.
\end{proof}

\paragraph{The adaptive structure.} It is the algorithm of Theorem~\ref{thm:batch}, with three changes. Each $S_i$ is replaced by an independent feature map $\phi_i$, drawn as $\phi$ above from $G_i$ and $\xi^{(i)}$. The $(c',(1+\epsilon)r)$-$\ell_2$ LSH is replaced by the $(c'',R_1)$ LSH above. Distances are measured in $\ell_p$.
\begin{itemize}
    \item \emph{Initialization.} Draw $k=\wt O(\sqrt T\cdot s)$ independent structures as above, with $k$ and $l=O(s\log^2(n/\beta))$ set as in Section~\ref{sec:dp}. Compute ${\cal G}U$, where ${\cal G}\in\R^{kD\times d}$ is the matrix whose rows are those of $G_1,\ldots,G_k$. Then compute $\phi_i(u)$ for all $i\in[k]$ and $u\in U$, and store these features. Build the $i$-th LSH on $\phi_i(U)$, with the pointers above.
    \item \emph{Query.} Given $v_t$, sample $l$ structures and compute $\phi_{(1)}(v_t),\ldots,\phi_{(l)}(v_t)$. For each $j\in[l]$, let $u$ be the point reported by the $j$-th sampled LSH on $\phi_{(j)}(v_t)$, and compute $\|u-v_t\|_p$ in $O(d)$ time. Let the $j$-th characteristic vector be the indicator vector of $u$ if $\|u-v_t\|_p\leq cr$, and ${\bf 0}_n$ otherwise, including when the LSH reports nothing. Sum these characteristic vectors and apply sparse~\rpm{} to the sum. Let $i$ be the noisy max index. Output $u_i$ if $u_i$ has not been deleted and $\|u_i-v_t\|_p\leq cr$, and NULL otherwise.
    \item \emph{Deletion.} Given the output $u_t$ of $v_t$, for every $i\in[k]$, remove $\phi_i(u_t)$ from each hash table of the $i$-th LSH, using the stored pointers. Mark $u_t$ as deleted.
\end{itemize}

\begin{theorem}[Deletions in adaptive $\ell_p$ LSH]
\label{thm:lp_delete}
Let $p\in(0,2]$ and $c>1$ be constants, and let $r>0$. For $x,y\in\R^d$ let $\|x-y\|_p=(\sum_{a=1}^d|x_a-y_a|^p)^{1/p}$ (a quasi-norm when $p<1$) and $B(v,r)=\{x\in\R^d:\|x-v\|_p\leq r\}$. Let $\rho=1/c^p+f(n)$, where $f(n)=o(1)$ is the error term, uniform over approximation ratios in $(1,c]$, in the exponent of the Euclidean LSH of~\cite{ai06} (see the paragraph on the exponent above). Let $U\subseteq\R^d$ be an $n$-point dataset satisfying Assumption~\ref{assumption:ann_sparse} for these balls, assume $\log n\geq 6$ and $\log n>3/(c^p-1)$, and let $\{v_1,\ldots,v_T\}\subseteq\R^d$ be a sequence of adaptive queries. There exists a randomized data structure with the following guarantees:
\begin{itemize}
    \item It preprocesses $U$ in time $\wt O(\Tmat(\sqrt T\cdot s,d,n)+\sqrt T\cdot s\cdot n^{1+\rho})$.
    \item It uses space $\wt O(\sqrt T\cdot s\cdot (n^{1+\rho}+d)+nd)$.
    \item Given a query $v_t$ for $t\in[T]$, let $U'\subseteq U$ be the set of points not deleted before time $t$. It returns a point $u\in U'$ with $\|v_t-u\|_p\leq cr$ if $B(v_t,r)\cap U'\neq\emptyset$, and otherwise returns NULL or a point $u\in U'$ with $\|v_t-u\|_p\leq cr$, with probability at least $1-1/\poly(n)$. This procedure takes time $\wt O(s\cdot(d+n^\rho))$, with probability at least $1-1/\poly(n)$.
    \item If $u_t$ is the output of $v_t$, it deletes $u_t$ in worst-case time $\wt O(\sqrt T\cdot s\cdot n^\rho)$.
\end{itemize}
Here $\wt O(\cdot)$ hides factors polylogarithmic in $n$ and $T$, as in Section~\ref{sec:dp}, and a factor $c^{4p}$. For $p=2$ the exponent is $1/c^2+o(1)$, as in Theorem~\ref{thm:batch}.
\end{theorem}

\begin{proof}
\emph{Preprocessing and space.} Drawing the $kDd$ entries of ${\cal G}$ takes $O(kDd)=\wt O(\sqrt T\cdot s\cdot d)$ time, which is at most $\wt O(\Tmat(\sqrt T\cdot s,d,n))$. Since $kD=\wt O(\sqrt T\cdot s)$, splitting the rows of ${\cal G}$ into $\wt O(1)$ blocks of at most $\sqrt T\cdot s$ rows shows that computing ${\cal G}U$ takes $\wt O(\Tmat(\sqrt T\cdot s,d,n))$ time. The cosines take $O(kDn)$ time. Building the $k$ LSH's on $D$-dimensional points, together with the pointers, takes $\wt O(k\cdot n^{1+1/c''^2+f(n)})=\wt O(\sqrt T\cdot s\cdot n^{1+\rho})$ time, by the exponent bound. For space:
\begin{itemize}
    \item the LSH's and the pointers use $\wt O(\sqrt T\cdot s\cdot n^{1+\rho})$ words,
    \item the features use $kDn=\wt O(\sqrt T\cdot s\cdot n)$ words,
    \item the matrices $G_i$ and the phases $\xi^{(i)}$ use $kD(d+1)=\wt O(\sqrt T\cdot s\cdot d)$ words,
    \item $U$ uses $nd$ words, and the deletion marks use $n$ bits.
\end{itemize}

\emph{Query time.} Computing the $l$ features of $v_t$ takes $O(lDd)=\wt O(s\cdot d)$ time. Each LSH query in dimension $D$ takes $\wt O(n^\rho)$ time, and each distance check takes $O(d)$ time. The sparse noisy max of Section~\ref{sec:dp} takes $\wt O(s)$ time with probability at least $1-1/\poly(n)$. The total is $\wt O(s\cdot(d+n^\rho))$.

\emph{Correctness.} Fix $t\in[T]$. The $i$-th LSH holds exactly $\phi_i(U')$, in the buckets given by its hash functions. The order inside a bucket may differ from that of a fresh build on $\phi_i(U')$. The success event of~\cite{ai06} does not depend on that order: it requires that a near point shares a bucket with the query in some table and that few far points share buckets with it. So Lemma~\ref{lem:lp_oblivious} applies.

Say that the $i$-th structure succeeds on $v_t$ if the event of Lemma~\ref{lem:lp_oblivious} holds for $(v_t,U')$. For a fixed query and a fixed $U'$, this has probability at least $9/10$. The event is determined by the random string $r_i$ of the $i$-th structure, by $v_t$ and by the deleted points, and the deleted points are outputs and hence part of the transcript. A successful structure reports a point when $B(v_t,r)\cap U'\neq\emptyset$. By (a) this point passes the distance check, so its characteristic vector is nonzero.

By the distance checks, every characteristic vector is supported on $B(v_t,cr)\cap U'$. This set has at most $s+1$ points. Indeed, if $u,u'\in B(v_t,cr)$, then $v_t\in B_u\cap B_{u'}$, because $\|u-v_t\|_p=\|v_t-u\|_p$. So $s+2$ such points would give a ball that meets $s+1$ others, contrary to Assumption~\ref{assumption:ann_sparse}, which $U'\subseteq U$ inherits. This argument uses no triangle inequality, so it also holds for $p<1$.

Hence Lemmas~\ref{lem:acc_global}, \ref{lem:acc_sample} and~\ref{lem:correct} apply with $U'$ in place of $U$. With $\delta_{\rm fail}=1/\poly(n)$, they show that, with probability at least $1-1/\poly(n)$, the noisy max index is a point of $B(v_t,cr)\cap U'$ whenever $B(v_t,r)\cap U'\neq\emptyset$. The final check depends only on the noisy max index and on the set of deleted points, which is part of the transcript. So it is post-processing and does not affect privacy. In every case, it ensures that the output is NULL or a point of $B(v_t,cr)\cap U'$.

\emph{Deletion time.} For each $i\in[k]$, the deletion unlinks $\phi_i(u_t)$ from each of the $O(n^\rho)$ hash tables of the $i$-th LSH, in $O(1)$ time per table, using the stored pointers. Over all $k$ structures this takes $\wt O(\sqrt T\cdot s\cdot n^\rho)$ time in the worst case. Marking $u_t$ as deleted takes $O(1)$ time. No feature is recomputed, so there is no term that depends on $d$.
\end{proof}

\section{Adaptive Regression with Private Median and \texorpdfstring{$\ell_\infty$}{} Guarantee}

In this section, we provide the necessary background for adaptive regression problem under turnstile updates, with a variety of algorithms with diverse guarantees.

\subsection{Preliminaries on Adpative Regression and Sketching}

The adaptive regression we will be studying is defined in Assumption~\ref{assumption:regression}, we will state here again for completeness.

Let $U\in \R^{n\times d}$ be a design matrix and $b\in \R^n$ be a response vector with $n\gg d$, the goal is to solve the $\ell_2$ regression problem:
\begin{align*}
    x^* := & ~ \arg\min_{x\in \R^d} \|Ux-b\|_2^2,
\end{align*}
with the optimal solution given by the normal equation:
\begin{align*}
    x^* = & ~ (U^\top U)^\dagger U^\top b
\end{align*}
with $M^\dagger$ being the Moore–Penrose inverse of the matrix $M$. As solving the normal equation exactly is time- and space-consuming, so one is usually interested in finding an approximate solution $\wt x\in \R^d$ such that
\begin{align*}
    \|U\wt x-b\|_2 \leq & ~ (1+\alpha) \|Ux^*-b\|_2
\end{align*}
for some $\alpha>0$. 

In the adaptive turnstile update model, we assume an adversary could curate a sequence of $T$ updates $\{v_1,\ldots,v_T\}$ adaptively, where each $v_t$ for $t\in [T]$ could take one of the two forms:
\begin{itemize}
    \item $v_t\in \R^{n\times d}$, i.e., $v_t$ is an update to the design matrix;
    \item $v_t\in \R^n$, i.e., $v_t$ is an update to the response vector. 
\end{itemize}
We need to design an algorithm that processes these adaptive updates to either $U$ or $b$, for simplicity we will use $(U_t, b_t)$ to denote the design matrix and response vector after the $t$-th update. Our algorithm needs to respond with an approximate solution $x_t$ to the $t$-th $\ell_2$ regression:
\begin{align*}
    \|U_t x_t-b_t\|_2 \leq & ~ (1+\alpha) \min_{x\in \R^d} \|U_t x-b_t\|_2.
\end{align*}
The main objective is to design an algorithm that
\begin{itemize}
    \item Uses space that is sublinear in $n$, sublinear in $T$ and small polynomial in $d$;
    \item Has efficient update that depends on the size of $v_t$ and small polynomial in $d$.
\end{itemize}

Assumption~\ref{assumption:regression} also contains two standing hypotheses: the length of the stream is $T = \poly(n)$, and in any round $0\leq t \leq T$ (with $(U_0,b_0)=(U,b)$), entries of $U_t, b_t$ are always in the range $[-n^{\gamma}, n^{\gamma}]$ for some constant $\gamma > 0$.

Sketching will be a key algorithmic tool to speed up the procedure of solving $\ell_2$ regression, we introduce them in the following.

\begin{definition}[Oblivious Subspace Embedding] \label{def:subspace-embedding}
   Fix dimension parameters $n$, $d$, approximation factor $\alpha$, and failure probability $\beta$. Suppose ${\cal D}$ is a distribution on $r \times n$ matrices $S$, where $r$ is a function of $n, d, \alpha,$ and $\beta$. Suppose that with probability at least $1 - \beta$, for any fixed $n \times d$ matrix $U$, a matrix $S$ drawn from ${\cal D}$ satisfies: for all vectors $x \in \R^d$,

   \begin{align*}
       \| SUx \|^2_2 = & ~ (1 \pm \alpha) \| Ux \|^2_2,
   \end{align*}
    
    then we call ${\cal D}$ an $(\alpha, \beta)$-Oblivious Subspace Embedding (OSE).
\end{definition}

We will mainly need the OSE property with $\alpha = O(1)$ and $\beta = O(1)$. We refer to this as $O(1)$-OSE. We list a collection of sketching matrices that will be used throughout this section.

\begin{definition}[Count Sketch~\cite{ccf02}] 
    Let $S \in \R^{r\times n}$ be constructed via the following procedure: Randomly draw $h: [n] \to [r]$ from a pairwise independent hash family, and draw $\sigma: [n] \to \{-1, 1\}$ from a 4-wise independent hash family. For each of the $n$ columns $S_{i \in [n]}$, we choose a row $h(i) \in [r]$ and an element $\sigma(i)$ from $\{-1, 1\}$. We set $S_{h(i), i} = \sigma(i)$ for all $i \in [n]$, and set all other entries of $S$ to be $0$. We call such $S$ a Count Sketch matrix.
\end{definition}

\begin{definition}[Gaussian Sketching Matrix] 
We say $S \in \R^{r \times n}$ is a Gaussian sketching matrix if all entries are independently sampled from the distribution ${\cal N}(0, 1/r)$.
\end{definition}

\begin{definition}[Subspace Randomized Hardamard Transform (SRHT)~\cite{ldfu13}] 
The Subspace Randomized Hardamard Transform (SRHT) matrix $S\in \R^{r \times n}$ is defined as the scaled matrix product $S:=\frac{1}{\sqrt{r}} P H D$, where each row of matrix $P \in \{0, 1\}^{r \times n}$ contains exactly one $1$ at a random position, $H$ is the $n \times n$ Hadamard matrix, and $D$ is a $n \times n$ diagonal matrix with each diagonal entry being a value in $\{-1, 1\}$ with equal probability. 
\end{definition}

We will use the following differentially private median procedure:

\begin{theorem}\label{thm:primed}
  There exists an $(\epsilon, 0)$-differentially private algorithm that given a database $S \in X^*$, outputs an element $x \in X$ such that with probability at least $1 - \beta$, there are at least $|S|/2 - \Gamma$ elements in $S$ that are bigger or equal to $x$, and there are at least $|S|/2 - \Gamma$ elements in $S$ that are smaller or equal to $x$, where $\Gamma = O(\frac{1}{\epsilon} \log \frac{|X|}{\beta})$. Moreover, private median runs in time
  \begin{align*}
      O(\epsilon^{-1}|S|\log^3(|X|/\beta)\cdot \poly\log |S|).
  \end{align*}
\end{theorem}
We will use $\textsc{PMedian}(x_1,\ldots,x_s)$ to denote the invocation of private median on $x_1,\ldots,x_s$. 

\subsection{Generic Algorithm with Private Median}

A critical property we will be leveraging in designing the algorithm is a coordinate-wise form of the \emph{$\ell_\infty$ guarantee} of the sketched solution, a guarantee studied in~\cite{psw17,syyz23}: the sketched solution $\wt x$ and the optimal solution $x^*$ should be close not only in terms of their \emph{costs}, i.e., $\|U\wt x-b\|_2=(1\pm\alpha) \|Ux^*-b\|_2$, but also in each coordinate. Note that a na{\"i}ve bound of $\|\wt x-x^*\|_\infty$ is just $\|\wt x-x^*\|_2$, as one could have the scenario that all the discrepancy concentrate in a few coordinates with most coordinates are the same. The $\ell_\infty$ guarantee bounds all $d$ coordinates at once, and~\cite{psw17} show that a Count Sketch matrix alone can fail it. We only need a bound for each fixed coordinate separately. We prove it directly for the composition of an SRHT matrix with a Count Sketch matrix, which improves the space usage and runtime efficiency.
\begin{lemma}\label{lem:sketch}
    Let $\alpha,\beta\in(0,1)$, let $U\in\R^{n\times d}$ have full column rank, let $b\in\R^n$, and let $x^*:=\arg\min_{x\in\R^d}\|Ux-b\|_2$. Let $S:=S_{\textsc{SRHT}} \cdot S_{\textsc{CS}}$, where $S_{\textsc{CS}} \in \R^{m \times n}$ is a Count Sketch matrix and $S_{\textsc{SRHT}} \in \R^{r \times m}$ is an independent SRHT matrix built from the $m\times m$ Hadamard matrix, with $m$ a power of $2$. There is an absolute constant $C$ such that if
    \begin{align*}
        m\geq \frac{128d^2}{\alpha\beta}+\frac{256d}{\alpha^2\beta} \quad \text{and}\quad r\geq C\Big(\frac{d}{\alpha^2}+\frac{d^{1.5}}{\alpha}\Big)\log^3\frac{m}{\beta},
    \end{align*}
    which allows $m = O(\frac{d^2}{\alpha\beta} + \frac{d}{\alpha^2\beta})$ and $r = O((\frac{d}{\alpha^2}+\frac{d^{1.5}}{\alpha})\log^3\frac{m}{\beta})$, then $\wt x:=\arg\min_{x\in\R^d}\|S(Ux-b)\|_2$ satisfies the following: for each fixed $l\in[d]$, with probability at least $1-\beta$,
    \begin{align}\label{eq:l_infty}
        |\wt x_l-x^*_l| \leq & ~ \frac{\alpha}{\sqrt d} \|Ux^*-b\|_2\cdot \frac{1}{\sigma_{\min}(U)}.
    \end{align}
    We will refer to this per-coordinate property as $(\alpha,\beta)$-accuracy.
\end{lemma}

\begin{proof}
    Write $U=Q\wt R$, where $Q\in\R^{n\times d}$ has orthonormal columns and $\wt R\in\R^{d\times d}$ is invertible, and let $\eta:=b-Ux^*$. The normal equations $U^\top\eta=0$ give $Q^\top\eta=0$. Fix $l\in[d]$, let $e_l$ be the $l$-th standard basis vector of $\R^d$, and let $w:=\wt R^{-\top}e_l$, so that $\|w\|_2\leq\|\wt R^{-1}\|=1/\sigma_{\min}(U)$. We may assume $\eta\neq 0$: otherwise $b=Ux^*$, and every sketched problem below returns $x^*$ on the events we consider.

    \vspace{0.1cm}
    \noindent{\bf A deterministic bound.} Let $\Pi$ be any matrix with $n$ columns that is an $\epsilon_0$-subspace embedding for $U$ (Definition~\ref{def:subspace-embedding}) with $\epsilon_0\leq 1/2$, i.e., $\|Q^\top\Pi^\top\Pi Q-I\|\leq\epsilon_0$. Then $\Pi U$ has full column rank, and $\hat x:=\arg\min_{x\in\R^d}\|\Pi(Ux-b)\|_2=(U^\top\Pi^\top\Pi U)^{-1}U^\top\Pi^\top\Pi b$. Substituting $b=Ux^*+\eta$ and $U=Q\wt R$ gives
    \begin{align*}
        \hat x-x^* = (U^\top\Pi^\top\Pi U)^{-1}U^\top\Pi^\top\Pi\eta = \wt R^{-1}(Q^\top\Pi^\top\Pi Q)^{-1}Q^\top\Pi^\top\Pi\eta.
    \end{align*}
    The eigenvalues of $Q^\top\Pi^\top\Pi Q$ lie in $[1-\epsilon_0,1+\epsilon_0]$, so $(Q^\top\Pi^\top\Pi Q)^{-1}=I+E$ with $\|E\|\leq\epsilon_0/(1-\epsilon_0)\leq2\epsilon_0$. Taking the $l$-th coordinate, $\hat x_l-x^*_l=e_l^\top\wt R^{-1}(I+E)Q^\top\Pi^\top\Pi\eta=w^\top(I+E)Q^\top\Pi^\top\Pi\eta$, hence
    \begin{align}\label{eq:coord_bound}
        |\hat x_l-x^*_l| \leq |(Qw)^\top\Pi^\top\Pi\eta| + 2\epsilon_0\|w\|_2\|Q^\top\Pi^\top\Pi\eta\|_2, \quad \text{where } (Qw)^\top\eta=0 \text{ and } \|Qw\|_2=\|w\|_2.
    \end{align}

    \vspace{0.1cm}
    \noindent{\bf Stage 1: Count Sketch.} Let $x':=\arg\min_{x\in\R^d}\|S_{\textsc{CS}}(Ux-b)\|_2$. We use the following moment bound (see, e.g.,~\cite{cw13,w14}): for fixed $y,z\in\R^n$,
    \begin{align}\label{eq:cs_moment}
        \E[(y^\top S_{\textsc{CS}}^\top S_{\textsc{CS}}z-y^\top z)^2]\leq\frac{2\|y\|_2^2\|z\|_2^2}{m}.
    \end{align}
    Indeed, with $h$ and $\sigma$ as in the definition of Count Sketch,
    \begin{align*}
        y^\top S_{\textsc{CS}}^\top S_{\textsc{CS}}z-y^\top z=\sum_{i\neq j}y_iz_j\sigma(i)\sigma(j)\mathbf{1}[h(i)=h(j)].
    \end{align*}
    Since $\sigma$ is $4$-wise independent, in the second moment the term $(i,j)$ correlates only with the terms $(i,j)$ and $(j,i)$, and $\Pr[h(i)=h(j)]=1/m$. So the second moment is at most $\frac{1}{m}\sum_{i\neq j}(y_i^2z_j^2+|y_iz_iy_jz_j|)$, which is at most $\frac{2}{m}\|y\|_2^2\|z\|_2^2$ by the Cauchy--Schwarz inequality. Let $\epsilon_0:=\sqrt{\alpha/8}\leq 1/2$. We show that each of the following four events fails with probability at most $\beta/8$.
    \begin{enumerate}
        \item $S_{\textsc{CS}}$ is an $\epsilon_0$-subspace embedding for $U$. By~\eqref{eq:cs_moment} applied to all pairs of columns of $Q$, $\E\|Q^\top S_{\textsc{CS}}^\top S_{\textsc{CS}}Q-I\|_F^2\leq 2d^2/m$. So by Markov's inequality, the bound $\|Q^\top S_{\textsc{CS}}^\top S_{\textsc{CS}}Q-I\|_F\leq\epsilon_0$, which implies the same bound in operator norm, fails with probability at most $2d^2/(m\epsilon_0^2)=16d^2/(\alpha m)\leq\beta/8$, as $m\geq 128d^2/(\alpha\beta)$. This is the $O(d^2)$-row subspace embedding of Count Sketch~\cite{cw13,w14} for $\epsilon_0,\beta$ fixed.
        \item $|(Qw)^\top S_{\textsc{CS}}^\top S_{\textsc{CS}}\eta|\leq\frac{\alpha}{4\sqrt d}\|w\|_2\|\eta\|_2$. By~\eqref{eq:cs_moment} with $(Qw)^\top\eta=0$ and Chebyshev's inequality, this fails with probability at most $2\|w\|_2^2\|\eta\|_2^2\cdot\frac{16d}{m\alpha^2\|w\|_2^2\|\eta\|_2^2}=\frac{32d}{\alpha^2m}\leq\beta/8$, as $m\geq 256d/(\alpha^2\beta)$.
        \item $\|Q^\top S_{\textsc{CS}}^\top S_{\textsc{CS}}\eta\|_2\leq 4\sqrt{d/(m\beta)}\,\|\eta\|_2$. By~\eqref{eq:cs_moment} applied to each column of $Q$ (which is orthogonal to $\eta$), $\E\|Q^\top S_{\textsc{CS}}^\top S_{\textsc{CS}}\eta\|_2^2\leq 2d\|\eta\|_2^2/m$, so by Markov's inequality this fails with probability at most $\beta/8$.
        \item $\|S_{\textsc{CS}}\eta\|_2\leq 2\|\eta\|_2$. By~\eqref{eq:cs_moment} with $y=z=\eta$ and Chebyshev's inequality, $\Pr[\|S_{\textsc{CS}}\eta\|_2^2-\|\eta\|_2^2\geq 3\|\eta\|_2^2]\leq\frac{2}{9m}\leq\beta/8$.
    \end{enumerate}
    On events 1 to 3, the bound~\eqref{eq:coord_bound} with $\Pi=S_{\textsc{CS}}$ gives
    \begin{align*}
        |x'_l-x^*_l| \leq & ~ \frac{\alpha}{4\sqrt d}\|w\|_2\|\eta\|_2 + 8\epsilon_0\sqrt{\frac{d}{m\beta}}\,\|w\|_2\|\eta\|_2 \\
        \leq & ~ \Big(\frac{\alpha}{4\sqrt d}+\frac{\alpha}{4\sqrt d}\Big)\frac{\|\eta\|_2}{\sigma_{\min}(U)} = \frac{\alpha}{2\sqrt d}\cdot\frac{\|\eta\|_2}{\sigma_{\min}(U)},
    \end{align*}
    because $m\geq 128d^2/(\alpha\beta)$ gives $\frac{d}{m\beta}\leq\frac{\alpha}{128d}$, hence $8\epsilon_0\sqrt{\frac{d}{m\beta}}\leq 8\sqrt{\frac{\alpha}{8}}\cdot\sqrt{\frac{\alpha}{128d}}=\frac{8\alpha}{32\sqrt d}=\frac{\alpha}{4\sqrt d}$.

    \vspace{0.1cm}
    \noindent{\bf Stage 2: SRHT.} Condition on a fixed $S_{\textsc{CS}}$ for which events 1 to 4 hold. Let $U':=S_{\textsc{CS}}U$ and $b':=S_{\textsc{CS}}b$, so that $x'=\arg\min_x\|U'x-b'\|_2$ and $\wt x=\arg\min_x\|S_{\textsc{SRHT}}(U'x-b')\|_2$. By event 1, $\sigma_{\min}(U')\geq\sqrt{1-\epsilon_0}\,\sigma_{\min}(U)\geq\sigma_{\min}(U)/2$, and by event 4 the residual $\eta':=b'-U'x'$ satisfies $\|\eta'\|_2\leq\|S_{\textsc{CS}}(b-Ux^*)\|_2=\|S_{\textsc{CS}}\eta\|_2\leq 2\|\eta\|_2$. Write $U'=Q'\wt R'$ as above and let $w':=\wt R'^{-\top}e_l$, so that $\|w'\|_2\leq 1/\sigma_{\min}(U')\leq 2/\sigma_{\min}(U)$ and $Q'^\top\eta'=0$. These quantities depend only on $S_{\textsc{CS}}$, and $S_{\textsc{SRHT}}$ is independent of $S_{\textsc{CS}}$. We use two facts about an SRHT matrix $\Pi\in\R^{r\times m}$.
    \begin{enumerate}
        \item[(S1)] There is an absolute constant $C_0$ such that for all $\epsilon\in(0,1/2]$ and $\beta_0\in(0,1)$, if $r\geq C_0\epsilon^{-2}(d+\log(m/\beta_0))\log(d/\beta_0)$, then $\Pi$ is an $\epsilon$-subspace embedding for any fixed $m\times d$ matrix with probability at least $1-\beta_0$ (see, e.g.,~\cite{ldfu13,w14}).
        \item[(S2)] For fixed $y,z\in\R^m$, $\epsilon\in(0,1)$ and $\beta_0\in(0,1)$, if $r\geq 8\epsilon^{-2}\ln^2(8m/\beta_0)\ln(4/\beta_0)$, then $\Pr[|y^\top\Pi^\top\Pi z-y^\top z|>\epsilon\|y\|_2\|z\|_2]\leq\beta_0$.
    \end{enumerate}
    To see (S2), write $\Pi=\frac{1}{\sqrt r}PHD$ and let $\bar y:=HDy/\sqrt m$ and $\bar z:=HDz/\sqrt m$. Since $HD/\sqrt m$ is orthogonal, $\|\bar y\|_2=\|y\|_2$, $\|\bar z\|_2=\|z\|_2$ and $\bar y^\top\bar z=y^\top z$. Each $\bar y_j$ is a Rademacher sum whose coefficients have total squared norm $\|y\|_2^2/m$, so by Hoeffding's inequality and a union bound over the $2m$ coordinates of $\bar y$ and $\bar z$, with probability at least $1-\beta_0/2$ we have $|\bar y_j|\leq\sqrt{2L/m}\,\|y\|_2$ and $|\bar z_j|\leq\sqrt{2L/m}\,\|z\|_2$ for all $j$, where $L:=\ln(8m/\beta_0)$. Fix $D$ in this event. Then $y^\top\Pi^\top\Pi z=\frac1r\sum_{i=1}^r Y_i$ with $Y_i:=m\bar y_{p_i}\bar z_{p_i}$, where $p_i$ is the position of the $1$ in row $i$ of $P$. The $Y_i$ are independent with $\E[Y_i]=\bar y^\top\bar z=y^\top z$ and $|Y_i|\leq 2L\|y\|_2\|z\|_2$, so Hoeffding's inequality bounds the probability that $|\frac1r\sum_iY_i-y^\top z|>\epsilon\|y\|_2\|z\|_2$ by $2\exp(-r\epsilon^2/(8L^2))\leq\beta_0/2$.

    Let $\epsilon_1:=\sqrt{\alpha/(32\sqrt d)}\leq 1/2$. We show that each of the following three events fails with probability at most $\beta/6$.
    \begin{enumerate}
        \setcounter{enumi}{4}
        \item $S_{\textsc{SRHT}}$ is an $\epsilon_1$-subspace embedding for $U'$. By (S1) with $\epsilon=\epsilon_1$ and $\beta_0=\beta/6$, this needs $r\geq C_0\cdot\frac{32\sqrt d}{\alpha}(d+\log(6m/\beta))\log(6d/\beta)$.
        \item $|(Q'w')^\top S_{\textsc{SRHT}}^\top S_{\textsc{SRHT}}\eta'|\leq\frac{\alpha}{16\sqrt d}\|w'\|_2\|\eta'\|_2$. By (S2) with $y=Q'w'$, $z=\eta'$, $\epsilon=\frac{\alpha}{16\sqrt d}$ and $\beta_0=\beta/6$, using $(Q'w')^\top\eta'=0$, this needs $r\geq 2048\frac{d}{\alpha^2}\ln^2(48m/\beta)\ln(24/\beta)$.
        \item $\|Q'^\top S_{\textsc{SRHT}}^\top S_{\textsc{SRHT}}\eta'\|_2\leq\epsilon_1\|\eta'\|_2$. By (S2) applied to each column $q'$ of $Q'$ and $z=\eta'$, with $\epsilon=\epsilon_1/\sqrt d$ and $\beta_0=\beta/(6d)$, using $q'^\top\eta'=0$, and a union bound over the $d$ columns, we get $\|Q'^\top S_{\textsc{SRHT}}^\top S_{\textsc{SRHT}}\eta'\|_2\leq\sqrt d\cdot\frac{\epsilon_1}{\sqrt d}\|\eta'\|_2$. This needs $r\geq 256\frac{d^{1.5}}{\alpha}\ln^2(48md/\beta)\ln(24d/\beta)$.
    \end{enumerate}
    Since $m\geq d$ and $m/\beta\geq 128$, every logarithm above is $O(\log(m/\beta))$ and $d+\log(6m/\beta)=O(d\log(m/\beta))$, so all three requirements hold when $r\geq C(\frac{d}{\alpha^2}+\frac{d^{1.5}}{\alpha})\log^3\frac{m}{\beta}$ for a large enough absolute constant $C$. On events 5 to 7, the bound~\eqref{eq:coord_bound} applied to $(U',b')$ in place of $(U,b)$, with $x',Q',w',\eta'$ in place of $x^*,Q,w,\eta$, with $\epsilon_1$ in place of $\epsilon_0$ and with $\Pi=S_{\textsc{SRHT}}$, gives
    \begin{align*}
        |\wt x_l-x'_l| \leq & ~ \frac{\alpha}{16\sqrt d}\|w'\|_2\|\eta'\|_2 + 2\epsilon_1^2\|w'\|_2\|\eta'\|_2 = \frac{\alpha}{8\sqrt d}\|w'\|_2\|\eta'\|_2 \\
        \leq & ~ \frac{\alpha}{8\sqrt d}\cdot\frac{2}{\sigma_{\min}(U)}\cdot 2\|\eta\|_2 = \frac{\alpha}{2\sqrt d}\cdot\frac{\|\eta\|_2}{\sigma_{\min}(U)},
    \end{align*}
    using $2\epsilon_1^2=\frac{\alpha}{16\sqrt d}$.

    \vspace{0.1cm}
    \noindent{\bf Conclusion.} By the triangle inequality, on events 1 to 7 we have $|\wt x_l-x^*_l|\leq|\wt x_l-x'_l|+|x'_l-x^*_l|\leq\frac{\alpha}{\sqrt d}\cdot\frac{\|Ux^*-b\|_2}{\sigma_{\min}(U)}$, and the probability that one of the events fails is at most $4\cdot\frac{\beta}{8}+3\cdot\frac{\beta}{6}=\beta$.
    \end{proof}

We are now ready to describe our algorithm. Throughout, we will let $\alpha\in (0,1)$ be the approximation factor, $\beta\in (0, 1)$ be the failure probability and we require an extra parameter $\kappa$ which is an upper bound on the condition number of $U_t$ throughout the update sequence. Since our algorithm would utilize private median, we let $\epsilon_{\DP}$ to denote the privacy parameter which will be specified below. We will refer this algorithm to $\textsc{AdaptiveRegDP}$. We write $x^*_t:=\arg\min_{x\in\R^d}\|U_tx-b_t\|_2$ and $\OPT_t:=\min_{x\in\R^d}\|U_tx-b_t\|_2$, and $\gamma$ is the constant of Assumption~\ref{assumption:regression}.

\vspace{0.2cm}
\noindent{\bf Initialization.}
\begin{itemize}
    \item (Parameters for sketches): Let $\alpha'=\frac{\alpha}{2\kappa}$ and $\beta'=\frac{1}{100}$, and let $m=O(\frac{d^2}{\alpha'\beta'}+\frac{d}{\alpha'^2\beta'})=O(\frac{d^2\kappa}{\alpha}+\frac{d\kappa^2}{\alpha^2})$ (a power of $2$) and $r=O((\frac{d}{\alpha'^2}+\frac{d^{1.5}}{\alpha'})\log^3\frac{m}{\beta'})=\wt O(\frac{d\kappa^2}{\alpha^2}+\frac{d^{1.5}\kappa}{\alpha})$ be as in Lemma~\ref{lem:sketch} with parameters $(\alpha',\beta')$.
    \item (Parameters for rounding): Let $R=n^\gamma$, let $N=\lceil 8d\sqrt n\, n^{3\gamma}/\alpha\rceil$, let $h=2R/N$ and let $X=h\Z\cap[-2R,2R]$. Then $h\leq\frac{\alpha n^{-2\gamma}}{4d\sqrt n}$, $|X|=2N+1$ and $\log|X|=O(\log(nd/\alpha))$. For $z\in\R$, let ${\rm rd}_X(z)$ be the point of $X$ nearest to $\min\{\max\{z,-2R\},2R\}$, with ties broken toward the smaller point.
    \item (Parameters for privacy): Let $\delta_{\DP}=\frac{\beta}{400dT}$, $\delta_1=\frac{\delta_{\DP}}{2T}$ and $\epsilon_{\DP}=\frac{1}{2\sqrt{2d\ln(1/\delta_1)}}$. Let $\Gamma=O(\frac{1}{\epsilon_{\DP}}\log \frac{4dT |X|}{\beta})$ be the bound of Theorem~\ref{thm:primed} for privacy parameter $\epsilon_{\DP}$ and failure probability $\frac{\beta}{4dT}$, with its constant taken at least $1$ and the natural logarithm. Set $s=\lceil 100\Gamma\rceil$ and $k=\lceil 1200s\sqrt{2T\ln(2/\delta_{\DP})}\rceil$, so that $s=\wt O(\sqrt d)$ and $k=\wt O(\sqrt{Td})$.
    \item Prepare $k$ independent copies $S_1,\ldots,S_k\in \R^{r\times n}$ according to Lemma~\ref{lem:sketch}, and let $\rho_i$ be the random string that generates $S_i$.
    \item Let ${\rm sk}_U^i=S_i U$ and ${\rm sk}_b^i=S_i b$ for all $i\in [k]$.
\end{itemize}

\vspace{0.2cm}
\noindent{\bf Update.}
\begin{itemize}
    \item Receive update $v_t$;
    \item If $v_t$ is an update to $U$, then update ${\rm sk}_U^i\leftarrow {\rm sk}_U^i+S_i v_t$ for all $i\in [k]$, otherwise update ${\rm sk}_b^i\leftarrow {\rm sk}_b^i+S_i v_t$.
\end{itemize}

\vspace{0.2cm}
\noindent{\bf Query.}
\begin{itemize}
    \item Sample with replacement $s$ indices from $[k]$, let $j_1,\ldots,j_s$ denote the sampled indices.
    \item Compute $x_{j_i}=\arg\min_{x\in \R^d} \|{\rm sk}_U^{j_i} x - {\rm sk}_b^{j_i} \|_2$ for all $i\in [s]$.
    \item Compute $g_l=\textsc{PMedian}({\rm rd}_X((x_{j_1})_l),\ldots,{\rm rd}_X((x_{j_s})_l))$ over the domain $X$, with privacy budget $\epsilon_{\DP}$ and failure probability $\frac{\beta}{4dT}$ (Theorem~\ref{thm:primed}), for all $l\in [d]$.
    \item Output $g=(g_1,g_2,\ldots,g_d)$.
\end{itemize}

We now prove the privacy and utility of our algorithm.

\begin{lemma}\label{lem:reg_privacy}
    The transcript of \textsc{AdaptiveRegDP}, i.e., the sequence of updates and outputs, satisfies $(\epsilon, \delta)$-differential privacy w.r.t. the collection of random strings ${\cal R}=(\rho_1,\ldots,\rho_k)$ for $\epsilon := 1/100$ and $\delta := \delta_{\DP}=\frac{\beta}{400dT}$.
\end{lemma}

\begin{proof}
    Two collections are neighboring if they differ in one string $\rho_i$. The update $v_t$ is chosen by the adversary from the outputs of the earlier rounds, so the transcript is the adaptive composition of the $T$ rounds, and it suffices to bound the privacy of the output of one round $t$ for every fixed pair $(U_t,b_t)$.

    Fix a round $t$ and the pair $(U_t,b_t)$. Let $\wt{\cal M}$ be the mechanism that, given a database $(\rho'_1,\ldots,\rho'_s)$ of $s$ random strings, computes for each $i\in[s]$ the solution $x'_{(i)}$ of the regression problem $(U_t,b_t)$ sketched by the matrix generated from $\rho'_i$, and outputs $(g_1,\ldots,g_d)$ with
    \begin{align*}
        g_l=\textsc{PMedian}({\rm rd}_X((x'_{(1)})_l),\ldots,{\rm rd}_X((x'_{(s)})_l)).
    \end{align*}
    Replacing one string $\rho'_i$ changes only the $i$-th input of each of the $d$ medians, since ${\rm rd}_X$ acts on each copy separately. So by Theorem~\ref{thm:primed}, each median is $(\epsilon_{\DP},0)$-DP w.r.t.\ $(\rho'_1,\ldots,\rho'_s)$. Theorem~\ref{thm:ada_c}, applied to these $d$ mechanisms with $\delta=0$ and $\delta_0=\delta_1$, shows that $\wt{\cal M}$ is $(\epsilon_c,\delta_1)$-DP, where
    \begin{align*}
        \epsilon_c = \sqrt{2d\ln(1/\delta_1)}\cdot\epsilon_{\DP}+2d\epsilon_{\DP}^2 = \frac12+\frac{1}{4\ln(1/\delta_1)} \leq \frac34,
    \end{align*}
    by the choice $\epsilon_{\DP}=\frac{1}{2\sqrt{2d\ln(1/\delta_1)}}$ and since $\ln(1/\delta_1)\geq1$.

    Round $t$ of the algorithm runs $\wt{\cal M}$ on the strings of $s$ indices sampled with replacement from $[k]$. Since $\epsilon_c\leq1$ and $s\leq k/2$, Theorem~\ref{thm:amp} shows that round $t$ is $(\epsilon_r,\delta_r)$-DP w.r.t.\ ${\cal R}$, where
    \begin{align*}
        \epsilon_r:=\frac{6\epsilon_cs}{k}\leq\frac{4.5s}{k}\leq\frac{4.5}{1200\sqrt{2T\ln(2/\delta_{\DP})}} \quad\text{and}\quad \delta_r:=e^{\epsilon_r}\cdot\frac{4s}{k}\cdot\delta_1\leq\delta_1.
    \end{align*}
    For the last inequality, $s/k\leq1/1200$ gives $\epsilon_r\leq0.004$ and $e^{\epsilon_r}\cdot4s/k\leq e^{0.004}/300<1$. Finally, Theorem~\ref{thm:ada_c} over the $T$ rounds, with $\delta_0=\delta_{\DP}/2$, shows that the transcript is $(\epsilon_0,\delta_{\DP}/2+T\delta_1)$-DP w.r.t.\ ${\cal R}$, where
    \begin{align*}
        \epsilon_0 = \sqrt{2T\ln(2/\delta_{\DP})}\cdot\epsilon_r+2T\epsilon_r^2 \leq \frac{4.5}{1200}+\frac{4.5^2}{1200^2\ln(2/\delta_{\DP})} < 0.0038 < \frac{1}{100},
    \end{align*}
    and $\delta_{\DP}/2+T\delta_1=\delta_{\DP}/2+\delta_{\DP}/2=\delta_{\DP}$.
\end{proof}

\begin{lemma}\label{lem:utility}
    Suppose that, in addition to Assumption~\ref{assumption:regression}, for all $t\in[T]$ we have $\|x^*_t\|_\infty\leq R=n^\gamma$ and $\OPT_t\geq n^{-\gamma}$. Then with probability at least $1 - \beta$, in all rounds $t \in [T]$ during the update sequence, the algorithm \textsc{AdaptiveRegDP} outputs $g$ that satisfies

    \[\| U_t g - b_t\|_2 \leq (1+\alpha) \min_{x \in \R^{d}}\| U_t x - b_t\|_2, \]
     for underlying design matrix $U_t$ and vector $b_t$.
\end{lemma}

\begin{proof}
Consider any round $t \in [T]$ during the stream, let $(U_t, b_t)$ be the underlying pair defined by the stream up to round $t$, let $\sigma_{\min}$ denote the minimum singular value of $U_t$ and $\sigma_{\max}$ denote the maximum, and let $\tau_t:=\frac{\alpha'}{\sqrt d}\cdot\frac{\OPT_t}{\sigma_{\min}}$. For $l\in[d]$, let $f_{t,l}(\rho)$ be the indicator of the event
\begin{align*}
    |(x_\rho)_l-(x^*_t)_l|\leq\tau_t, \quad\text{where } x_\rho := \arg\min_{x \in \R^{d}} \| S(U_tx - b_t)\|_2,
\end{align*}
and $S$ is generated as in Lemma \ref{lem:sketch}, using random string $\rho$.
That is, $f_{t,l}(\rho)=1$ if the copy with string $\rho$ satisfies Equation~\eqref{eq:l_infty}, with $\alpha'$ in place of $\alpha$, for $(U_t,b_t)$ and coordinate $l$. Since $\kappa(U_t)\leq\kappa$, $U_t$ has full column rank, and for every fixed pair $(U_t,b_t)$, the $(\alpha',\beta')$-accuracy of Lemma~\ref{lem:sketch} gives $\E_\rho[f_{t,l}(\rho)]\geq1-\beta'=\frac{99}{100}$. We show that three events hold, except with probability $\frac{\beta}{4}$ each.

\vspace{0.1cm}
\noindent{\bf Most of the $k$ copies are accurate.} The predicate $f_{t,l}$ is determined by $(U_t,b_t)$, hence by the updates $v_1,\ldots,v_t$, which are part of the transcript. By Lemma~\ref{lem:reg_privacy} and post-processing, the map from ${\cal R}$ to $f_{t,l}$ is $(\frac{1}{100},\delta_{\DP})$-DP. The strings $\rho_1,\ldots,\rho_k$ are drawn independently from the same distribution. The database size satisfies $k\geq1200s\geq240000\ln(4dT/\beta)$, which exceeds $\epsilon^{-2}\log(2\epsilon/\delta_{\DP})=10^4\log(8dT/\beta)$ for $\epsilon=\frac{1}{100}$ (here $s\geq200\ln(4dT/\beta)$, shown below), and $\delta_{\DP}<\epsilon/4$. So Theorem~\ref{thm:general}, applied with database size $k$, $\epsilon=\frac{1}{100}$ and $\delta=\delta_{\DP}$, shows that each of the $dT$ predicates $f_{t,l}$ satisfies
\begin{align*}
    \frac1k\sum_{j=1}^k f_{t,l}(\rho_j) \geq \E_\rho[f_{t,l}(\rho)]-10\epsilon \geq \frac{99}{100}-\frac{1}{10} \geq \frac45
\end{align*}
except with probability $\delta_{\DP}/\epsilon=100\delta_{\DP}=\frac{\beta}{4dT}$. A union bound over $t\in[T]$ and $l\in[d]$ bounds the failure probability by $\frac{\beta}{4}$.

\vspace{0.1cm}
\noindent{\bf Most of the $s$ sampled copies are accurate.} In round $t$, the indices $j_1,\ldots,j_s$ are drawn after $(U_t,b_t)$ is fixed, independently of ${\cal R}$ and of the earlier rounds. So given ${\cal R}$ and the history up to the update $v_t$, the values $f_{t,l}(\rho_{j_1}),\ldots,f_{t,l}(\rho_{j_s})$ are independent Bernoulli variables with mean $\frac1k\sum_{j=1}^kf_{t,l}(\rho_j)$. When this mean is at least $\frac45$, Hoeffding's inequality gives
\begin{align*}
    \Pr\Big[\sum_{i=1}^s f_{t,l}(\rho_{j_i})<\frac{3s}{4}\Big] \leq \exp\Big(-\frac{2(s/20)^2}{s}\Big) = \exp(-s/200) \leq \frac{\beta}{4dT},
\end{align*}
because $s\geq100\Gamma\geq200\ln(4dT/\beta)$, as $\epsilon_{\DP}^{-1}\geq2$ and the constant in $\Gamma$ is at least $1$. A union bound over $t$ and $l$ bounds the failure probability of this step by $\frac\beta4$.

\vspace{0.1cm}
\noindent{\bf The private medians succeed.} Each of the $dT$ calls of $\textsc{PMedian}$ fails with probability at most $\frac{\beta}{4dT}$, so with probability at least $1-\frac\beta4$ all of them satisfy the guarantee of Theorem~\ref{thm:primed}.

\vspace{0.1cm}
In total, all three events hold with probability at least $1-\frac{3\beta}{4}\geq1-\beta$, and we condition on them. Fix $t\in[T]$ and $l\in[d]$. At least $\frac{3s}{4}$ of the sampled copies are good for coordinate $l$, i.e., $|(x_{j_i})_l-(x^*_t)_l|\leq\tau_t$. Since $|(x^*_t)_l|\leq R$, clipping to $[-2R,2R]$ does not increase the distance to $(x^*_t)_l$, and rounding a point of $[-2R,2R]$ to $X$ moves it by at most $\frac h2$, since $\pm2R\in X$. So every good copy satisfies $|{\rm rd}_X((x_{j_i})_l)-(x^*_t)_l|\leq\tau_t+\frac h2$. Theorem~\ref{thm:primed} with $s\geq100\Gamma$ gives that at least $\frac s2-\Gamma\geq0.49s$ of the inputs ${\rm rd}_X((x_{j_i})_l)$ are at least $g_l$, and at least $0.49s$ of them are at most $g_l$. At most $\frac s4$ copies are not good, and $0.49s-0.25s>0$, so there are good copies $p,q$ with ${\rm rd}_X((x_{j_p})_l)\leq g_l\leq{\rm rd}_X((x_{j_q})_l)$. Hence $|g_l-(x^*_t)_l|\leq\tau_t+\frac h2$, for all $l\in[d]$ and $t\in[T]$.

We now bound the cost. Since the entries of $U_t$ lie in $[-n^\gamma,n^\gamma]$, $\sigma_{\max}\leq\|U_t\|_F\leq\sqrt{nd}\,n^\gamma$, and with $h\leq\frac{\alpha n^{-2\gamma}}{4d\sqrt n}$ we get $\sigma_{\max}\sqrt d\cdot\frac h2\leq d\sqrt n\,n^\gamma\cdot\frac{\alpha n^{-2\gamma}}{8d\sqrt n}=\frac\alpha8 n^{-\gamma}\leq\frac\alpha8\OPT_t$. Therefore
\begin{align*}
    \| U_tg - b_t\|_2 &\leq \| U_tx^*_t- b_t\|_2 + \| U_t(g-x^*_t)\|_2\\
    &\leq \OPT_t + \sigma_{\max} \| g-x^*_t\|_2\\
    &\leq \OPT_t + \sigma_{\max} (\sqrt{d}\cdot\| g-x^*_t\|_\infty)\\
    &\leq \OPT_t + \sigma_{\max}\sqrt d\Big(\frac{\alpha'}{\sqrt d}\cdot\frac{\OPT_t}{\sigma_{\min}}+\frac h2\Big)\\
    &\leq \Big(1+\frac{\sigma_{\max} }{\sigma_{\min}}\alpha'+\frac\alpha8\Big) \OPT_t.
\end{align*}
Since $\frac{\sigma_{\max}}{\sigma_{\min}}\leq\kappa$ and $\alpha'=\frac{\alpha}{2\kappa}$, this is at most $(1+\frac\alpha2+\frac\alpha8)\OPT_t\leq(1+\alpha)\OPT_t$. Overall, with probability at least $1-\beta$, all the output vectors are accurate to within a multiplicative error of $(1+\alpha)$.
\end{proof}

It remains to show that \textsc{AdaptiveRegDP} is both time- and space-efficient.

\begin{theorem}
Suppose that, in addition to Assumption~\ref{assumption:regression}, for all $t\in[T]$ we have $\|x^*_t\|_\infty\leq n^\gamma$ and $\OPT_t\geq n^{-\gamma}$. Given a sequence of $T$ adaptive updates $\{v_1,\ldots,v_T\}$, algorithm \textsc{AdaptiveRegDP} has the following specifications:
\begin{itemize}
    \item It preprocesses $(U, b)$ in time $\wt O(\sqrt{Td}\cdot (\nnz(U)+\nnz(b)+d^3\kappa/\alpha+d^2\kappa^2/\alpha^2))$.
    \item It uses space $\wt O(\sqrt{T}\cdot (d^{2.5}\kappa^2/\alpha^2+d^3\kappa/\alpha))$.
    \item Given an update $v_t$ for $t\in [T]$, it takes time $\wt O(\sqrt{Td}\cdot (\nnz(v_t)+d^3\kappa/\alpha+d^2\kappa^2/\alpha^2))$ to update.
    \item It outputs a solution $x_t$ in time $\wt O(d^{\omega+1}\kappa^2/\alpha^2)$, and with probability at least $1-\beta$, $x_t$ is a $(1+\alpha)$-approximate solution in all rounds $t\in[T]$.
\end{itemize}
\end{theorem}

\begin{proof}
The approximation guarantee is Lemma~\ref{lem:utility}. We prove the remaining items one by one. Recall that $m=O(d^2\kappa/\alpha+d\kappa^2/\alpha^2)$, $r=\wt O(d\kappa^2/\alpha^2+d^{1.5}\kappa/\alpha)$, $s=\wt O(\sqrt d)$ and $k=\wt O(\sqrt{Td})$.

\vspace{0.2cm}
\noindent{\bf Preprocessing time.} During preprocessing, we prepare $k=\wt O(\sqrt{Td})$ independent copies of sketching matrices due to Lemma~\ref{lem:sketch} and compute $S_iU, S_ib$ for all $i\in [k]$. Note that $S_i$ is a composition of a Count Sketch matrix and an SRHT matrix, so $S_i U$ takes time $O(\nnz(U))$ for the Count Sketch, and this results in a matrix of size $m\times d$. Applying SRHT to this matrix takes time $O(md\log m)$ with the fast Hadamard transform, and $md=O(d^3\kappa/\alpha+d^2\kappa^2/\alpha^2)$. Similarly, $S_ib$ takes time $O(\nnz(b)+m\log m)$. Hence, the overall time for preprocessing is
\begin{align*}
    \wt O(\sqrt{Td}\cdot (\nnz(U)+\nnz(b)+d^3\kappa/\alpha+d^2\kappa^2/\alpha^2)).
\end{align*}

\vspace{0.2cm}
\noindent{\bf Space usage.} The algorithm stores ${\cal R}$ and $\mathrm{sk}^i_U, \mathrm{sk}^i_b$ for all $i \in [k]$. We start by considering the random strings to generate $k$ pairs of SRHT and Count Sketch matrices. To store the pairwise and $4$-wise independent hash functions for one Count Sketch, we only need $O(1)$ words of space. To store an $r \times m$ SRHT matrix $S_{\textsc{SRHT}} = \frac{1}{\sqrt{r}} P H D$, since $H$ is deterministic, we only consider $P \in \R^{r \times m}$ and $D \in \R^{m \times m}$. $P$ contains exactly one $1$ in each row (and contains $0$ everywhere else), thus can be stored in $O(r\log m)$ bits of space. $D$ is a diagonal matrix with diagonal entries in $\{-1, 1\}$, so it can be stored in $O(m)$ bits of space. Therefore, each random string $\rho_i$ takes $O(r+m)$ words, assuming a word size of $O(\log n)$ bits. The sketches $\mathrm{sk}^i_U\in\R^{r\times d}$ and $\mathrm{sk}^i_b\in\R^r$ take $O(rd)$ words. Since $m=O(d^2\kappa/\alpha+d\kappa^2/\alpha^2)$ and $rd=\Omega(d^2\kappa^2/\alpha^2+d^{2.5}\kappa/\alpha)$, we have $m=O(rd)$, so the total space is
\begin{align*}
    O(k\cdot(rd+m)) = O(k\cdot r\cdot d) = \wt O\Big(\sqrt{Td}\cdot\Big(\frac{d^2\kappa^2}{\alpha^2}+\frac{d^{2.5}\kappa}{\alpha}\Big)\Big) = \wt O\Big(\sqrt{T}\cdot\Big(\frac{d^{2.5}\kappa^2}{\alpha^2}+\frac{d^3\kappa}{\alpha}\Big)\Big)
\end{align*}
words.

\vspace{0.2cm}
\noindent{\bf Update time.} For update, we simply apply the sketch $S_i$ to the update $v_t$ for all $i\in [k]$. If $v_t$ updates $U$, the Count Sketch takes time $O(\nnz(v_t))$ and gives an $m\times d$ matrix, and the SRHT then takes time $O(md\log m)$. If $v_t$ updates $b$, this takes time $O(\nnz(v_t)+m\log m)$. Hence, the total update time is $\wt O(\sqrt{Td}\cdot (\nnz(v_t)+d^3\kappa/\alpha+d^2\kappa^2/\alpha^2))$.

\vspace{0.2cm}
\noindent{\bf Query time.} To answer the query, we need to solve a sketched regression for $s$ instances, each with an $r\times d$ matrix. Since $r\geq d$, each regression could be solved in time $O(rd^{\omega-1})=\wt O(d^{\omega}\kappa^2/\alpha^2+d^{\omega+1/2}\kappa/\alpha)$, by forming the $d\times d$ Gram matrix with fast matrix multiplication and solving the resulting $d\times d$ system. With $s=\wt O(\sqrt d)$, the $s$ regressions take time
\begin{align*}
    \wt O(d^{\omega+1/2}\kappa^2/\alpha^2+d^{\omega+1}\kappa/\alpha) \leq \wt O(d^{\omega+1}\kappa^2/\alpha^2),
\end{align*}
because $d^{1/2}\leq d$ and $\kappa/\alpha\geq1$. Rounding the $sd$ coordinates to $X$ takes time $O(sd)=\wt O(d^{1.5})$. The private median is then performed over $s$ numbers of $X$, and by Theorem~\ref{thm:primed} each private median runs in time $O(\epsilon_{\DP}^{-1}s\log^3(4dT|X|/\beta)\cdot\poly\log s)=\wt O(d)$, since $\epsilon_{\DP}^{-1}=\wt O(\sqrt d)$ and $s=\wt O(\sqrt d)$. Repeating this procedure for all $d$ coordinates, the overall time for forming $g$ is then $\wt O(d^2)$. Hence, the query time is
\begin{align*}
& ~ \wt O(d^{\omega+1}\kappa^2/\alpha^2). \qedhere
\end{align*}
\end{proof}

\subsection{Improved Algorithm with Bounded Computation Path Technique}

When the condition number $\kappa$ is as large as $\poly n$, a quadratic dependence on $\kappa$ is prohibitive. We provide an algorithm for such case based on the \emph{bounded computation path technique}~\cite{hkm+22}. Let ${\cal P}$ denote the set of all possible output sequences of the algorithm against a deterministic adversary, which we define precisely in the proof of Lemma~\ref{lem:path_adaptive}. The algorithm outputs the exact solution when the sketch shows that the system is exactly consistent, and otherwise it rounds the sketched solution to a fixed finite grid. So the total number of possible computation paths $|{\cal P}|$ is at most $(n\kappa/\alpha)^{O(dT)}$, hence we could prepare a large Gaussian sketching matrix with size roughly $\log |{\cal P}|$, that is guaranteed to succeed against all such paths. This leads to a more efficient alternative and a better dependence on the condition number $\kappa$, which enters only through $\log|{\cal P}|$. We will refer to this algorithm as \textsc{AdaptiveRegPath}.

Throughout this subsection, in addition to Assumption~\ref{assumption:regression}, we assume that for all $t\in[T]$ the entries of $U_t$ and $b_t$ are integers, and that $\OPT_t$ is either $0$ or at least $n^{-\gamma}$. Since $\kappa(U_t)\leq\kappa$, every $U_t$ has full column rank, so $d\leq n$ and $x^*_t=U_t^\dagger b_t$.

\vspace{0.2cm}
\noindent{\bf Initialization.}
\begin{itemize}
    \item (Parameters for rounding): Let $\bar h=\frac{\alpha n^{-2\gamma}}{8d\sqrt n}$, $\bar R=3\sqrt n\,n^\gamma\kappa$ and $G=\bar h\Z\cap[-\bar R,\bar R]$. For $z\in\R$, let ${\rm rd}_G(z)$ be the point of $G$ nearest to $z$, with ties broken toward the smaller point. Then $|G|\leq2\bar R/\bar h+1=48dn^{1+3\gamma}\kappa/\alpha+1$.
    \item (Parameters for the union bound): Let $\beta_0=\beta/(C\cdot T\cdot(|G|^d+1)^T)$ for a large constant $C$. We show in the proof of Lemma~\ref{lem:path_adaptive} that $|{\cal P}|\leq(|G|^d+1)^T$, so $\beta_0\leq\beta/(C\cdot T\cdot|{\cal P}|)$.
    \item (Parameters for sketches): Let $r=O((d+\log (1/\beta_0))/\alpha^2)$ be large enough that the distribution of $r\times n$ Gaussian sketching matrices is an $(\frac\alpha4,\beta_0)$-OSE for $n\times(d+1)$ matrices~\cite{w14}. Generate the Gaussian sketching matrix $S\in \R^{r\times n}$ by a length-$\tilde{O}(r)$ bits random seed $\texttt{r}_S$ (Lemma~\ref{lemma:pseudorandom}).
    \item Let ${\rm sk}_U=SU$ and ${\rm sk}_b=Sb$.
\end{itemize}

\vspace{0.2cm}
\noindent{\bf Update.}
\begin{itemize}
    \item Receive update $v_t$;
    \item Using $\texttt{r}_S$, generate the column $S_{:,i}$ of $S$ for every row $i\in[n]$ in which $v_t$ has a nonzero entry (Lemma~\ref{lem:generate_gaussian}). If $v_t$ is an update to $U$, then for every nonzero entry $(v_t)_{ij}$, add $(v_t)_{ij}S_{:,i}$ to column $j$ of ${\rm sk}_U$. Otherwise, for every nonzero entry $(v_t)_i$, add $(v_t)_iS_{:,i}$ to ${\rm sk}_b$. Since $Sv_t=\sum_{i,j}(v_t)_{ij}S_{:,i}e_j^\top$ in the first case and $Sv_t=\sum_i(v_t)_iS_{:,i}$ in the second, this computes ${\rm sk}_U\leftarrow {\rm sk}_U+Sv_t$ or ${\rm sk}_b\leftarrow {\rm sk}_b+Sv_t$.
\end{itemize}

\vspace{0.2cm}
\noindent{\bf Query.}
\begin{itemize}
    \item Compute $x_t=\arg\min_{x\in \R^d} \|{\rm sk}_{U_t}x-{\rm sk}_{b_t} \|_2$, where ${\rm sk}_{U_t}=SU_t$ and ${\rm sk}_{b_t}=Sb_t$ are the current sketches.
    \item If $\|{\rm sk}_{U_t}x_t-{\rm sk}_{b_t}\|_2=0$, output $g=x_t$.
    \item Otherwise, round each coordinate of $x_t$ to $G$, and output $g=({\rm rd}_G((x_t)_1),\ldots,{\rm rd}_G((x_t)_d))\in G^d$.
\end{itemize}

We start by showing the utility guarantee of \textsc{AdaptiveRegPath} when the updates are \emph{oblivious}.

\begin{claim}\label{claim:gaussian}
    Fix a round $t\in[T]$.
    \begin{enumerate}
        \item If the update sequence is oblivious, so that $(U_t,b_t)$ does not depend on $S$, then with probability at least $1-\beta_0$, $S$ is an $\frac\alpha4$-subspace embedding for the column span of $[U_t,b_t]\in\R^{n\times(d+1)}$, i.e., $\|Sy\|_2^2=(1\pm\frac\alpha4)\|y\|_2^2$ for every $y$ in this span.
        \item If $S$ is an $\frac\alpha4$-subspace embedding for the column span of $[U_t,b_t]$, then the vector $g$ that \textsc{AdaptiveRegPath} outputs in round $t$ satisfies \[\| U_tg - b_t\|_2 \leq (1 + \alpha)\min_{x \in \R^d} \| U_tx - b_t\|_2\]
    for underlying input matrix $U_t$ and vector $b_t$. Moreover, $\|{\rm sk}_{U_t}x_t-{\rm sk}_{b_t}\|_2=0$ if and only if $\OPT_t=0$, in which case $g=U_t^\dagger b_t$, and otherwise $g\in G^d$.
    \end{enumerate}
\end{claim}

\begin{proof}
Part 1 follows from our setting of $r = O((d + \log(1/\beta_0))/\alpha^2)$, which is sufficient for the distribution of Gaussian matrices to be an $(\frac\alpha4, \beta_0)$-OSE for $n\times(d+1)$ matrices \cite{w14}, applied to the fixed matrix $[U_t,b_t]$.

For part 2, recall that $U_t$ has full column rank. For every $x\in\R^d$, the vector $U_tx-b_t$ lies in the column span of $[U_t,b_t]$, so
\begin{align}\label{eq:path_embed}
    \Big(1-\frac\alpha4\Big)\|U_tx-b_t\|_2^2\leq\|S(U_tx-b_t)\|_2^2\leq\Big(1+\frac\alpha4\Big)\|U_tx-b_t\|_2^2.
\end{align}
Similarly, $\|SU_tx\|_2^2\geq(1-\frac\alpha4)\|U_tx\|_2^2>0$ for every $x\neq0$, so $SU_t$ has full column rank and $x_t$ is the unique minimizer of $\|S(U_tx-b_t)\|_2$. Applying~\eqref{eq:path_embed} to $x_t$ and to $x^*_t$ gives
\begin{align*}
    \|U_tx_t-b_t\|_2^2 \leq \frac{\|S(U_tx_t-b_t)\|_2^2}{1-\alpha/4} \leq \frac{\|S(U_tx^*_t-b_t)\|_2^2}{1-\alpha/4} \leq \frac{1+\alpha/4}{1-\alpha/4}\cdot\OPT_t^2 \leq \Big(1+\frac\alpha2\Big)^2\OPT_t^2,
\end{align*}
where the last step uses $\frac{1+\alpha/4}{1-\alpha/4}=1+\frac{\alpha/2}{1-\alpha/4}\leq1+\frac{2\alpha}{3}\leq(1+\frac\alpha2)^2$ for $\alpha\in(0,1)$. So $\|U_tx_t-b_t\|_2\leq(1+\frac\alpha2)\OPT_t$.

\vspace{0.1cm}
\noindent{\bf Exactly consistent systems.} By~\eqref{eq:path_embed}, $\|{\rm sk}_{U_t}x_t-{\rm sk}_{b_t}\|_2=\|S(U_tx_t-b_t)\|_2$ is $0$ if and only if $U_tx_t=b_t$. If $U_tx_t=b_t$, then $\OPT_t=0$, and $x_t=U_t^\dagger b_t$ because $U_t$ has full column rank. Conversely, if $\OPT_t=0$, then $\|U_tx_t-b_t\|_2\leq(1+\frac\alpha2)\OPT_t=0$. So the algorithm outputs $g=x_t$ exactly when $\OPT_t=0$, and then $g=U_t^\dagger b_t$ and $\|U_tg-b_t\|_2=0=(1+\alpha)\OPT_t$.

\vspace{0.1cm}
\noindent{\bf Rounding.} Otherwise $\OPT_t\neq0$, so $\OPT_t\geq n^{-\gamma}$ by hypothesis. Since $\OPT_t\leq\|U_t\cdot0-b_t\|_2=\|b_t\|_2$ and the $n$ entries of $b_t$ lie in $[-n^\gamma,n^\gamma]$,
\begin{align*}
    \|U_tx_t\|_2\leq\|U_tx_t-b_t\|_2+\|b_t\|_2\leq\Big(1+\frac\alpha2\Big)\OPT_t+\|b_t\|_2\leq\Big(2+\frac\alpha2\Big)\|b_t\|_2\leq\frac52\sqrt n\,n^\gamma.
\end{align*}
Since $U_t$ has full column rank, $U_t\neq0$. Its entries are integers, so some entry has absolute value at least $1$, and the column containing it shows that $\sigma_{\max}(U_t)\geq1$. Hence $\sigma_{\min}(U_t)\geq\sigma_{\max}(U_t)/\kappa\geq1/\kappa$, and
\begin{align*}
    \|x_t\|_\infty\leq\|x_t\|_2\leq\frac{\|U_tx_t\|_2}{\sigma_{\min}(U_t)}\leq\frac52\sqrt n\,n^\gamma\kappa=\bar R-\frac12\sqrt n\,n^\gamma\kappa\leq\bar R-\bar h,
\end{align*}
since $\bar h\leq\frac18<\frac12\leq\frac12\sqrt n\,n^\gamma\kappa$. So for each $l\in[d]$, the multiple of $\bar h$ nearest to $(x_t)_l$ (with ties broken toward the smaller one) is within $\frac{\bar h}{2}$ of $(x_t)_l$, hence lies in $[-\bar R,\bar R]$ and in $G$, and it equals ${\rm rd}_G((x_t)_l)$. That is, the clipping to $[-\bar R,\bar R]$ is inactive, and $\|g-x_t\|_\infty\leq\frac{\bar h}{2}$. Since the entries of $U_t$ lie in $[-n^\gamma,n^\gamma]$, we have $\sigma_{\max}(U_t)\leq\|U_t\|_F\leq\sqrt{nd}\,n^\gamma$, so
\begin{align*}
    \|U_t(g-x_t)\|_2&\leq\sigma_{\max}(U_t)\sqrt d\,\|g-x_t\|_\infty\leq\sqrt{nd}\,n^\gamma\cdot\sqrt d\cdot\frac{\bar h}{2}\\&=d\sqrt n\,n^\gamma\cdot\frac{\alpha n^{-2\gamma}}{16d\sqrt n}=\frac{\alpha}{16}n^{-\gamma}\leq\frac{\alpha}{16}\OPT_t.
\end{align*}
Therefore $\|U_tg-b_t\|_2\leq\|U_tx_t-b_t\|_2+\|U_t(g-x_t)\|_2\leq(1+\frac\alpha2+\frac\alpha{16})\OPT_t\leq(1+\alpha)\OPT_t$, and $g\in G^d$.
\end{proof}

\begin{lemma}\label{lemma:pseudorandom}
    The Gaussian sketching matrix $S$ in \textsc{AdaptiveRegPath} can be pseudo-randomly generated using a random seed of length $\tilde{O}(r)=\tilde{O}((d+\log(1/\beta_0))/\alpha^2)$ bits, where $\beta_0=\beta/(C\cdot T\cdot(|G|^d+1)^T)$, while still satisfying Claim \ref{claim:gaussian}.
\end{lemma}

\begin{proof}[Proof Sketch]
This can be seen by opening up the proof e.g., in Section 2.1 of \cite{w14} that the Gaussian distribution satisfies OSE. The i.i.d. property is used for showing that the random Gaussian matrix is a JL transform (Lemma 2.12 in \cite{w12}). However, there one needs to consider at most an $O(\log(9^{d+1}/\beta_0)/\alpha^2)$-tuple of Gaussian variables, where $9^{d+1}$ is the size of the $1/2$-net for a unit sphere in the $(d+1)$-dimensional subspace that we hope to embed in. Therefore the proof goes through as long as the entries of $S$ are generated using at least $r = O((d+\log(1/\beta_0))/\alpha^2)$-wise independence. It is well known that such an $r$-wise independent hash function can be compressed as a length $\tilde{O}(r)$-bit seed, which concludes the proof.
\end{proof}

\begin{lemma}\label{lem:path_adaptive}
     With probability at least $1-\beta$, in all rounds $t \in [T]$, \textsc{AdaptiveRegPath} given an adaptive update $v_t$ outputs a vector $g$ that satisfies:

     \[\| U_t g - b_t\|_2 \leq (1+\alpha) \min_{x \in \R^{d}}\| U_t x - b_t\|_2, \]
     for underlying input matrix $U_t$ and vector $b_t$.
\end{lemma}

\begin{proof}
    As argued in \cite{bjwy20}, we may assume the adversary to be deterministic. This means, in particular, that the output sequence we provide to the adversary fully determines its stream of updates $v_1, v_2, \cdots, v_T$: for each $t$, the update $v_t$, and hence the instance $(U_t,b_t)$, is a fixed function of the outputs $g_1,\ldots,g_{t-1}$ of the earlier rounds. Let ${\cal P}$ be the set of sequences $(g_1,\ldots,g_T)$ in which, for every $t\in[T]$, $g_t$ is either a vector of $G^d$ or the exact solution $U_t^\dagger b_t$ of the instance $(U_t,b_t)$ determined by $g_1,\ldots,g_{t-1}$. Given $g_1,\ldots,g_{t-1}$, there are at most $|G|^d+1$ choices of $g_t$, so
    \begin{align*}
        |{\cal P}|\leq(|G|^d+1)^T\leq(2|G|^d)^T\leq\Big(\frac{98n^{2+3\gamma}\kappa}{\alpha}\Big)^{dT}=\Big(\frac{n\kappa}{\alpha}\Big)^{O(dT)},
    \end{align*}
    where we used $|G|\leq48dn^{1+3\gamma}\kappa/\alpha+1\leq49n^{2+3\gamma}\kappa/\alpha$, which holds as $d\leq n$, $\kappa\geq1$ and $\alpha<1$, and that $\gamma$ is a constant. In particular, $\beta_0\leq\beta/(C\cdot T\cdot|{\cal P}|)$.

    Each $p=(g_1,\ldots,g_T)\in{\cal P}$ determines the updates that the adversary produces in response to the prefixes of $p$, and hence a stream of instances $(U^p_1,b^p_1),\ldots,(U^p_T,b^p_T)$ that does not depend on $S$. Let ${\cal E}_p$ be the event that for every $t\in[T]$, $S$ is an $\frac\alpha4$-subspace embedding for the column span of $[U^p_t,b^p_t]$. By part 1 of Claim~\ref{claim:gaussian}, this fails with probability at most $\beta_0$ for each of the $T$ rounds, so ${\cal E}_p$ fails with probability at most $T\beta_0$. A union bound over these $T|{\cal P}|$ events, one for each round of each of the $|{\cal P}|$ streams, shows that all the events ${\cal E}_p$ hold with probability at least $1-T|{\cal P}|\beta_0\geq1-\beta/C\geq1-\beta$.

    Condition on this. We show by induction on $t$ that the outputs $g_1,\ldots,g_t$ of the algorithm are $(1+\alpha)$-approximate solutions and form a prefix of a sequence in ${\cal P}$. For $t=0$ there is nothing to show. Suppose that this holds for $t-1$. The instance $(U_t,b_t)$ of round $t$ is the one determined by $g_1,\ldots,g_{t-1}$. Extending $g_1,\ldots,g_{t-1}$ by arbitrary vectors of $G^d$ gives a sequence $p\in{\cal P}$ with $(U^p_t,b^p_t)=(U_t,b_t)$, so the event ${\cal E}_p$ shows that $S$ is an $\frac\alpha4$-subspace embedding for the column span of $[U_t,b_t]$. By part 2 of Claim~\ref{claim:gaussian}, $g_t$ is a $(1+\alpha)$-approximate solution, and $g_t$ is either $U_t^\dagger b_t$ or a vector of $G^d$, so $g_1,\ldots,g_t$ is a prefix of a sequence in ${\cal P}$. This completes the induction. Hence, with probability at least $1-\beta$, \textsc{AdaptiveRegPath} outputs a $(1+\alpha)$-approximation in every round against any adversary.
\end{proof}

\begin{lemma}\label{lem:generate_gaussian}
    The time to generate $S$ via $\texttt{r}_S$ is $\tilde{O}(rn)$, and any single column of $S$ can be generated in $\tilde{O}(r)$ time.
\end{lemma}

\begin{proof}[Proof Sketch]
    Note that to generate one column $S_{:,i}$ of the sketching matrix, we need $r$ evaluations of the $r$-wise independent hash function. Using fast multipoint evaluation of polynomials \cite{G13}, these evaluations can be done in $\tilde{O}(r)$ time. Hence, generating $S$ takes $\wt O(rn)$ time, as desired.
\end{proof}

\begin{theorem}
Suppose that, in addition to Assumption~\ref{assumption:regression}, for all $t\in[T]$ the entries of $U_t$ and $b_t$ are integers and $\OPT_t$ is either $0$ or at least $n^{-\gamma}$. Let ${\cal P}$ be the set of output sequences defined in the proof of Lemma~\ref{lem:path_adaptive}, so that $|{\cal P}|\leq(n\kappa/\alpha)^{O(dT)}$. Given a sequence of $T$ adaptive updates $\{v_1,\ldots,v_T\}$, algorithm \textsc{AdaptiveRegPath} has the following specifications:
\begin{itemize}
    \item It preprocesses $(U, b)$ in time $\wt O({nd^{\omega-2}(d + \log |{\cal P}| + \log \frac{1}{\beta})}/{\alpha^2})$;
    \item It uses space $\wt O({d(d + \log |{\cal P}| + \log \frac{1}{\beta})}/{\alpha^2})$;
    \item Given an update $v_t$ for $t\in [T]$, it takes time $\wt O( \nnz(v_t)\cdot( d + \log |{\cal P}| + \log \frac{1}{\beta} ) / \alpha^2)$ to update. In particular, an update with $O(1)$ nonzero entries takes time $\wt O( ( d + \log |{\cal P}| + \log \frac{1}{\beta} ) / \alpha^2)$.
    \item It outputs a $(1+\alpha)$-approximate solution $x_t$ in time $\wt O({d^{\omega-1}(d + \log |{\cal P}| + \log \frac{1}{\beta})}/{\alpha^2} )$.
\end{itemize}
Here $\beta$ denotes the failure probability.
\end{theorem}

\begin{proof}
We prove the theorem item by item. The approximation guarantee is Lemma~\ref{lem:path_adaptive}. Since $\bar h<\bar R$, the set $G$ contains $0$ and $\pm\bar h$, so $|G|\geq3$. In the definition of ${\cal P}$, every $g_t$ can be any of the $|G|^d$ vectors of $G^d$, so $|{\cal P}|\geq|G|^{dT}\geq2^{dT}$, and $(|G|^d+1)^T\leq(2|G|^d)^T\leq|G|^{2dT}\leq|{\cal P}|^2$. Hence $\log T\leq\log|{\cal P}|$ and
\begin{align*}
    \log\frac{1}{\beta_0}=\log\frac{C\cdot T\cdot(|G|^d+1)^T}{\beta}=O\Big(\log|{\cal P}|+\log\frac1\beta\Big),
\end{align*}
so $r=O((d + \log |{\cal P}| + \log \frac{1}{\beta})/\alpha^2)$. The parameter $r$ is expanded in this way to obtain the theorem.

\vspace{0.2cm}
\noindent{\bf Preprocessing time.} During preprocessing, we generate the sketching matrix $S$ using our random seed, which by Lemma~\ref{lem:generate_gaussian} takes $\wt O(rn)$ time. Afterward, computing the sketches is dominated by computing the matrix product of $r$-by-$n$ matrix $S$ and $n$-by-$d$ matrix $U$, which is $\wt O(rnd^{\omega-2})$.

\vspace{0.2cm}
\noindent{\bf Space usage.} The sketches $\mathrm{sk}_U$ and  $\mathrm{sk}_b$ together take $\wt O(d r)$ bits of space, and the random seed takes $\wt O(r)$ bits of space.

\vspace{0.2cm}
\noindent{\bf Update time.} Let $I\subseteq[n]$ be the set of rows in which $v_t$ has a nonzero entry, so $|I|\leq\nnz(v_t)$. By Lemma~\ref{lem:generate_gaussian}, we generate the column $S_{:,i}$ for each $i\in I$ in $\wt O(r)$ time, which is $\wt O(r\cdot\nnz(v_t))$ time in total. Then each nonzero entry of $v_t$ contributes one scaled column, $(v_t)_{ij}S_{:,i}$ added to column $j$ of ${\rm sk}_U$ or $(v_t)_iS_{:,i}$ added to ${\rm sk}_b$, which takes $O(r)$ time. Hence, the update takes $\wt O(r\cdot\nnz(v_t))$ time.

\vspace{0.2cm}
\noindent{\bf Query time.} To answer the query, we solve a sketched regression instance in time $\wt O(rd^{\omega-1})$, and we compute the sketched residual ${\rm sk}_{U_t}x_t-{\rm sk}_{b_t}$ in $O(rd)$ time. If the residual is nonzero, we then round each of the $d$ coordinates of $x_t$ to $G$. The point of $G$ nearest to $z\in\R$ is obtained by rounding $z$ to the nearest multiple of $\bar h$ and clipping the result to the smallest or the largest point of $G$, in $O(1)$ time, so the rounding takes $O(d)$ time.
\end{proof}

\subsection{Handling Sparse Label Shifts with Preconditioner}

In machine learning, it is common that the design or data matrix remains unchanged, as it might be generated via embedding raw data into high dimensional vectors, and these embeddings are expensive to compute. The response or label vector, in contrary, might constantly drift. Hence~\cite{csw+23} studies a model where only $b$ receives a sparse update $v_t$ with $\|v_t\|_0\leq s$. For example, $b$ could be the count of a large class of images, and only a few popular classes get their counts incremented frequently. In~\cite{csw+23}, they are only able to output regression cost against sparse, adaptive updates to $b$. Here, we develop an algorithm based on \textsc{AdaptiveRegDP} to output the solution vector. We will use \textsc{AdaptiveRegPreconditioner} to denote this algorithm.

Throughout this subsection, $U$ has full column rank, $b_t$ denotes the response vector after the $t$-th update (with $b_0=b$), $x^*_t:=\arg\min_{x\in\R^d}\|Ux-b_t\|_2$ and $\OPT_t:=\min_{x\in\R^d}\|Ux-b_t\|_2$. For a constant $\gamma>0$, we assume that the entries of $U$ and of every $b_t$ lie in $[-n^\gamma,n^\gamma]$, that $\OPT_t\geq n^{-\gamma}$ for all $t$, and that the batch size $T$ below is $\poly(n)$. Here $s$ denotes the sparsity bound $\|v_t\|_0\leq s$, so we write $\lambda$ for the number of copies the algorithm samples in each round.

\vspace{0.2cm}
\noindent{\bf Initialization.}
\begin{itemize}
    \item (Parameters for sketches): Let $\alpha'=\frac{\alpha}{8}$ and $\beta'=\frac{1}{100}$, and let $m=O(\frac{d^2}{\alpha}+\frac{d}{\alpha^2})$ (a power of $2$) and $r=O((\frac{d}{\alpha'^2}+\frac{d^{1.5}}{\alpha'})\log^3\frac{m}{\beta'})=\wt O(\frac{d}{\alpha^2}+\frac{d^{1.5}}{\alpha})$ be as in Lemma~\ref{lem:sketch} with parameters $(\alpha',\beta')$.
    \item (Parameters for rounding): Let $R'=2\sqrt n\,n^\gamma$, $N'=\lceil16\sqrt{nd}\,n^{2\gamma}/\alpha\rceil$, $h'=2R'/N'$ and $X'=h'\Z\cap[-2R',2R']$, so that $h'\leq\frac{\alpha n^{-\gamma}}{4\sqrt d}$ and $\log|X'|=O(\log(nd/\alpha))$. Let ${\rm rd}_{X'}$ be defined as ${\rm rd}_X$, with $R'$ and $X'$ in place of $R$ and $X$.
    \item (Parameters for privacy): Let $\delta_{\DP}$, $\delta_1$ and $\epsilon_{\DP}$ be as in \textsc{AdaptiveRegDP}, with $T$ the batch size. Let $\Gamma=O(\frac{1}{\epsilon_{\DP}}\log \frac{4dT |X'|}{\beta})$ be the bound of Theorem~\ref{thm:primed} for privacy parameter $\epsilon_{\DP}$ and failure probability $\frac{\beta}{4dT}$, with its constant taken at least $1$. Set $\lambda=\lceil100\Gamma\rceil$ and $k=\lceil1200\lambda\sqrt{2T\ln(2/\delta_{\DP})}\rceil$, so that $\lambda=\wt O(\sqrt d)$ and $k=\wt O(\sqrt{Td})$.
    \item (Parameters for batch size): Let $T$ be a parameter denote the batch size depending on $\nnz(U), d$ and $\alpha$. Let ${\rm counter}=1$ be the batch counter parameter.
    \item Prepare $k$ independent copies $S_1,\ldots,S_k\in \R^{r\times n}$ according to Lemma~\ref{lem:sketch}.
    \item Compute a preconditioner $P\in \R^{d\times d}$ such that all singular values of $UP$ lie in $[1/2,2]$, so that $\kappa(UP)\leq 4$, as follows. Draw $S_0=S'_{\textsc{SRHT}}S'_{\textsc{CS}}$, where $S'_{\textsc{CS}}$ is a Count Sketch matrix with $m_0\geq256d^2/\beta$ rows ($m_0$ a power of $2$) and $S'_{\textsc{SRHT}}$ is an SRHT matrix with $O((d+\log(d/\beta))\log(d/\beta))$ rows, compute a QR factorization $S_0U=Q_0R_0$, and set $P=R_0^{-1}$.
    \item Prepare $M_i=(S_i UP)^\dagger S_i$ and ${\rm sk}^i=M_i b$ for all $i\in [k]$, where $M_i$ is kept implicitly through $(S_iUP)^\dagger$ and the random string of $S_i$.
\end{itemize}

\vspace{0.2cm}
\noindent{\bf Update.}
\begin{itemize}
    \item If ${\rm counter}=T$, regenerate sketching matrices $S_1,\ldots,S_k\in \R^{r\times n}$, set $M_i=(S_i UP)^\dagger S_i$ and ${\rm sk}^i=M_i b$ for all $i\in [k]$. Reset the ${\rm counter}=1$;
    \item Receive update $v_t$;
    \item Update $b\leftarrow b+v_t$ and ${\rm sk}^i\leftarrow {\rm sk}^i+M_iv_t$ for all $i\in [k]$.
\end{itemize}

\vspace{0.2cm}
\noindent{\bf Query.}
\begin{itemize}
    \item Sample with replacement $\lambda$ indices from $[k]$, let $j_1,\ldots,j_\lambda$ denote the sampled indices.
    \item Compute $g_l=\textsc{PMedian}({\rm rd}_{X'}(({\rm sk}^{j_1})_l),\ldots,{\rm rd}_{X'}(({\rm sk}^{j_\lambda})_l))$ over the domain $X'$, with privacy budget $\epsilon_{\DP}$ and failure probability $\frac{\beta}{4dT}$ (Theorem~\ref{thm:primed}), for all $l\in [d]$.
    \item Let $\wt g=(g_1,g_2,\ldots,g_d)$, output $g=P\wt g$.
\end{itemize}

Instead of fixing the length of update sequence in advance, we allow for arbitrary length of update sequence, and we will restart the data structure every $T$ updates. Each batch uses fresh sketches, so privacy and utility follow from Lemmas~\ref{lem:reg_privacy} and~\ref{lem:utility}, applied to each batch as to a stream of length $T$, with $\lambda$ in place of $s$ and $X'$ in place of $X$. We point out the three differences.
\begin{itemize}
    \item By the proof of Lemma~\ref{lem:sketch} (event 1 and fact (S1), each with distortion $\frac14$ and failure probability $\frac\beta8$), $S_0$ satisfies $\frac{9}{16}\|Ux\|_2^2\leq\|S_0Ux\|_2^2\leq\frac{25}{16}\|Ux\|_2^2$ for all $x\in\R^d$, except with probability $\frac\beta4$. Since $\|S_0UPy\|_2=\|Q_0y\|_2=\|y\|_2$, all singular values of $UP$ then lie in $[\frac45,\frac43]\subseteq[\frac12,2]$. This event depends only on $U$ and $S_0$, which are fixed before the first update, so the adversary cannot influence it, and we condition on it.
    \item By the definition of $M_i$, ${\rm sk}^i=M_ib_t=\arg\min_{y\in\R^d}\|S_i(UPy-b_t)\|_2$ is the sketched solution of the preconditioned problem $\min_{y}\|UPy-b_t\|_2$, whose minimizer is $y^*_t:=P^{-1}x^*_t$ and whose optimal cost is $\OPT_t$. So Lemma~\ref{lem:sketch} applies to $UP$ with parameters $(\alpha',\beta')$. Moreover, $UPy^*_t$ is the orthogonal projection of $b_t$ onto the column space of $U$, so $\|y^*_t\|_\infty\leq\|y^*_t\|_2\leq\|UPy^*_t\|_2/\sigma_{\min}(UP)\leq2\|b_t\|_2\leq2\sqrt n\,n^\gamma=R'$. Hence the proofs of Lemmas~\ref{lem:reg_privacy} and~\ref{lem:utility} show that within a batch, except with probability $\frac{3\beta}{4}$, every round and every $l\in[d]$ satisfy $|\wt g_l-(y^*_t)_l|\leq\frac{\alpha'}{\sqrt d}\cdot\frac{\OPT_t}{\sigma_{\min}(UP)}+\frac{h'}{2}$.
    \item As in the proof of Lemma~\ref{lem:utility}, using $\kappa(UP)\leq4$, $\alpha'=\frac\alpha8$, $\sigma_{\max}(UP)\leq2$, $h'\leq\frac{\alpha n^{-\gamma}}{4\sqrt d}$ and $\OPT_t\geq n^{-\gamma}$,
    \begin{align*}
        \|Ug-b_t\|_2=\|UP\wt g-b_t\|_2 &\leq \OPT_t+\sigma_{\max}(UP)\sqrt d\,\|\wt g-y^*_t\|_\infty \\
        &\leq \OPT_t+\kappa(UP)\alpha'\OPT_t+\sigma_{\max}(UP)\sqrt d\cdot\frac{h'}{2} \\
        &\leq \OPT_t+\frac\alpha2\OPT_t+\frac\alpha4n^{-\gamma} \leq (1+\alpha)\OPT_t.
    \end{align*}
\end{itemize}
So in each batch, with probability at least $1-\beta$, every output is a $(1+\alpha)$-approximate solution. We now turn to the complexity.

\begin{theorem}
Suppose that $U$ has full column rank, that $T=\poly(n)$, and that for a constant $\gamma>0$ the entries of $U$ and of every $b_t$ lie in $[-n^\gamma,n^\gamma]$ and $\OPT_t\geq n^{-\gamma}$ for all $t$. Given a sequence of adaptive updates $\{v_1,v_2,\ldots,\}$ with $\|v_t\|_0\leq s$, algorithm \textsc{AdaptiveRegPreconditioner} has the following specifications:
\begin{itemize}
    \item It has amortized update time $\wt O(\sqrt{d/T}\cdot(\nnz(U)+\nnz(b)+d^{3}/\alpha+d^{\omega}/\alpha^2)+\sqrt{Td}\cdot (s+d^{2}/\alpha^2+d^{2.5}/\alpha))$.
    \item It outputs a solution $x_t$ in time $\wt O(d^2)$, and in each batch, with probability at least $1-\beta$, $x_t$ is a $(1+\alpha)$-approximate solution in every round of the batch.
    \item Before the first update, it computes the preconditioner $P$ once, in time $\wt O(\nnz(U)+d^3/\beta+d^{\omega})$.
\end{itemize}
\end{theorem}

\begin{proof}
The approximation guarantee was shown above. Recall that $m=O(d^2/\alpha+d/\alpha^2)$, $r=\wt O(d/\alpha^2+d^{1.5}/\alpha)$, $\lambda=\wt O(\sqrt d)$ and $k=\wt O(\sqrt{Td})$. We analyze over batches. At the start of each batch, we generate $k$ sketching matrices, and for each $i\in[k]$ we do the following. The Count Sketch part of $S_i$ is applied to $U$ and $b$ in $O(\nnz(U)+\nnz(b))$ time. The SRHT part is applied to the resulting $m\times d$ matrix and $m$-dimensional vector in $O(md\log m)=\wt O(d^3/\alpha+d^2/\alpha^2)$ time. Right multiplying by $P$ and computing the pseudoinverse $(S_iUP)^\dagger=((S_iUP)^\top S_iUP)^{-1}(S_iUP)^\top$ takes $O(rd^{\omega-1})=\wt O(d^\omega/\alpha^2+d^{\omega+1/2}/\alpha)$ time, since $r\geq d$. Finally, note that we don't need to form $M_i$, rather we multiply $(S_iUP)^\dagger$ with the vector $S_ib$ in $O(rd)=\wt O(d^2/\alpha^2+d^{2.5}/\alpha)$ time. Using $d^2/\alpha^2\leq d^\omega/\alpha^2$, $d^{2.5}/\alpha\leq d^3/\alpha$, and $d^{\omega+1/2}\leq d^3$ (as $\omega\leq5/2$), the cost per copy is $\wt O(\nnz(U)+\nnz(b)+d^3/\alpha+d^\omega/\alpha^2)$. Over the $k=\wt O(\sqrt{Td})$ copies, the amortized time for this process over $T$ steps is
\begin{align*}
    \wt O(\sqrt{d/T}\cdot(\nnz(U)+\nnz(b)+d^3/\alpha+d^\omega/\alpha^2)).
\end{align*}
When receiving $v_t$, it takes $O(s)$ time to update the response vector. For each copy, computing $S_iv_t$ takes $O(s)$ time for the Count Sketch part, which gives an $m$-dimensional vector, and $O(m\log m)$ time for the SRHT part, and multiplying the result by $(S_iUP)^\dagger$ takes $O(rd)$ time. Since $m=O(d^2/\alpha+d/\alpha^2)$ and $rd=\Omega(d^2/\alpha^2+d^{2.5}/\alpha)$, we have $m=O(rd)$, so this is $\wt O(s+d^2/\alpha^2+d^{2.5}/\alpha)$ per copy. Hence, the overall time is
\begin{align*}
    \wt O(\sqrt{Td}\cdot (s+d^2/\alpha^2+d^{2.5}/\alpha)).
\end{align*}
Computing $P$ is a one-time cost at initialization: the Count Sketch $S'_{\textsc{CS}}$ takes $O(\nnz(U))$ time, the SRHT of the resulting $m_0\times d$ matrix takes $\wt O(m_0d)=\wt O(d^3/\beta)$ time, and the QR factorization and the inverse take $\wt O(d^\omega)$ time.

Now, to output a query, we round the $\lambda d$ sampled coordinates to $X'$ in $O(\lambda d)=\wt O(d^{1.5})$ time, and compute $d$ private medians, each over $\lambda=\wt O(\sqrt d)$ numbers with privacy parameter $\epsilon_{\DP}$, $\epsilon_{\DP}^{-1}=\wt O(\sqrt d)$. By Theorem~\ref{thm:primed}, each takes $\wt O(\epsilon_{\DP}^{-1}\lambda)=\wt O(d)$ time, so $\wt O(d^2)$ in total. Finally, we output $P\wt g$ in $O(d^2)$ time, so the query time is $\wt O(d^2)$.
\end{proof}

We next choose the batch size $T$.

\begin{claim}
    Let $\Phi:=\nnz(U)+\nnz(b)+d^3/\alpha+d^\omega/\alpha^2$ and $\Psi:=s+d^2/\alpha^2+d^{2.5}/\alpha$ be the two brackets in the amortized update time above, and choose $T=\max\{1,\lceil\Phi/\Psi\rceil\}$. Then the amortized update time of the algorithm is $\wt O(\sqrt{d\Phi\Psi}+\sqrt d\cdot\Psi)$.
\end{claim}

\begin{proof}
If $\Phi\geq\Psi$, then $\Phi/\Psi\leq T\leq2\Phi/\Psi$, so $\sqrt{d/T}\cdot\Phi\leq\sqrt{d\Psi/\Phi}\cdot\Phi=\sqrt{d\Phi\Psi}$ and $\sqrt{Td}\cdot\Psi\leq\sqrt{2d\Phi/\Psi}\cdot\Psi=\sqrt{2d\Phi\Psi}$. If $\Phi<\Psi$, then $T=1$, and the bound is $\sqrt d\cdot\Phi+\sqrt d\cdot\Psi\leq2\sqrt d\cdot\Psi$.
\end{proof}

\section{Improved Algorithm Against Adversarial Attack for Hamming Space LSH}

In a recent work,~\cite{kms24} shows that for Hamming space, there exists a query efficient algorithm that can quickly compute a point $v$ adaptively, such that there exists $u\in U$ such that $\|u-v\|\leq r$ (here $\| \cdot \|$ is the Hamming distance) but the data structure returns nothing. We state their result here.

\begin{definition}
Let $U\subseteq \R^d$ be an $n$-point dataset, we say a point $z\in U$ is an isolated point if for all points $u\in U$ with $u\neq z$, $\|u-z\|\geq 2cr$. Here $\| \cdot \|$ is the Hamming distance.
\end{definition}

\begin{theorem}[\cite{kms24}]
\label{thm:kms24_main}
Let $n$ denote the size of dataset, $d$ be the dimension, $(c, r)$ be the parameters for ANN, $\lambda$ be a complexity parameter. Suppose they satisfy the following relations:
\begin{itemize}
    \item $cr\leq d$;
    \item $\ln^3 n\leq r \leq d/\ln n$;
    \item $\lambda \leq \min\{\frac{r}{\ln n},n^{1/8} \}$;
    \item $c\geq 1+\frac{\ln \lambda}{\ln n}$.
\end{itemize}
Then with probability at least $1/4-1/n$, there exists an algorithm that finds a point $q$ such that an isolated point $z$ in the dataset is at most $r$ away from $q$, but the algorithm returns nothing. The algorithm makes $O(\log(cr)\cdot \lambda)$ queries to the LSH.
\end{theorem}

In their original theorem statement, they also require $c\leq \ln n$. This is unnecessary, as it is derived using the upper bound $r\leq d/\ln n$ and $cr\leq d$. If $r$ is much smaller, then $c$ can be larger. Their algorithm starts from the isolated point $z$, then gradually moves away from $z$ while ensuring no other points could collide with it. In the end, it moves at most $r$-away from $z$, but with high probability, no point in $U$ has been hashed to the same bucket as it.

To combat such an adversarial attack, one could either repeat the data structures $d$ times or $T$ times, which are both relatively inefficient. We will show that whenever the dataset satisfies Assumption~\ref{assumption:ann_sparse}, which holds automatically when the dimension $d$ is small enough, we could use Theorem~\ref{thm:main_lsh} with $\sqrt T\cdot s$ data structures against the attack of~\cite{kms24}.

\begin{lemma}
Let $n, d$ be positive integers and $c\geq 1$ and $r>0$ be parameters. Let $\rho\in (0,1)$ be a parameter. If $3d=s^\alpha$ and $\alpha \leq \frac{1}{2cr}$, then Assumption~\ref{assumption:ann_sparse} holds for any $n$-point set $U\subseteq \{0,1\}^d$. In particular, if $3d\leq n^{\frac{\rho}{2cr}}$ then Assumption~\ref{assumption:ann_sparse} holds with $s=n^\rho$.
\end{lemma}

\begin{proof}
Two balls of radius $cr$ can intersect only if their centers are within distance $2cr$. So Assumption~\ref{assumption:ann_sparse} holds if for each $u\in U$, at most $s$ other points of $U$ lie within distance $2cr$ of $u$. In the Hamming space, it means that they could differ by at most $2cr$ bits. Fix a vector $u$ and let $k=\lfloor 2cr\rfloor$. The total number of points, including $u$, that differ by at most $2cr$ bits from $u$ is at most
\begin{align*}
    \sum_{i=0}^{k} \binom{d}{i} \leq & ~ \Big(\frac{e d}{k}\Big)^{k} \leq (3d)^{2cr}
\end{align*}
if $1\leq k\leq d$. If $k=0$ the count is $1$, and if $k>d$ it is $2^d\leq (3d)^{2cr}$. As long as $2cr \log_s (3d)\leq 1$, we have $(3d)^{2cr}\leq s$, so we are guaranteed the total number of possible near neighbors is at most $s$. Taking $s=n^\rho$ gives the second claim.
\end{proof}

For Hamming space LSH, it is known that $\rho=O(1/c)$, and thus the bound on $d$ becomes $3d\leq n^{\Theta(1/(c^2r))}$. This range does not cover the parameters of Theorem~\ref{thm:kms24_main}. There $r\geq \ln^3 n$, so $n^{\rho/(2cr)}\leq e^{1/(2\ln^2 n)}<2$, while $3d\geq 3$. In the setting of Theorem~\ref{thm:kms24_main}, we therefore impose Assumption~\ref{assumption:ann_sparse} on the dataset directly, and obtain the following corollary.

\begin{corollary}
If $U\subseteq \{0,1\}^d$ satisfies Assumption~\ref{assumption:ann_sparse} for Hamming balls with $s\leq n^\rho$, and the parameters are as in Theorem~\ref{thm:kms24_main}, there exists an adaptive LSH data structure for Hamming space with
\begin{itemize}
    \item Preprocessing time $\wt O(\sqrt T\cdot n^{1+O(\rho)}d)$;
    \item Space usage $\wt O(\sqrt T \cdot n^{1+O(\rho)}+nd)$;
    \item For all $t\in [T]$, with high probability, if $B(v_t,r)\cap U\neq \emptyset$ then it returns a point in $B(v_t, cr)\cap U$ in time $\wt O(n^{O(\rho)}d)$. Here $B(\cdot, \cdot)$ is the Hamming ball.
\end{itemize}
In particular, to be robust against the attack of Theorem~\ref{thm:kms24_main}, it suffices to pick $T=O(\log (cr)\cdot \lambda)$.
\end{corollary}

The above corollary essentially asserts that, for datasets satisfying Assumption~\ref{assumption:ann_sparse}, to be robust against the attack of~\cite{kms24}, it is sufficient to blow up the number of data structures by a $\sqrt{\log(cr)\cdot \lambda}$ factor, even in the presence of an isolated point. 
\section{Locating the Thin Level via Approximate Counting}

If one considers Assumption~\ref{assumption:general} for $(c,r)$-ANN, it states that for an adaptive query $v$, if $B(v, r)\cap U\neq \emptyset$ then $|B(v, cr)\cap U|\leq s$. In an iterative process, one usually does not care for a particular $r$ but only wants to find an approximate nearest neighbor, and this is often achieved via binary search over the choice of $r$. We also note that since our assumption constrains only one scale, unlike constant expansion and doubling dimension, it is completely possible that for some small $r$, Assumption~\ref{assumption:general} holds for query $v$ and it no longer holds for $2r$. In fact, even if $|B(v, cr)\cap U|$ is small, $|B(v, 2cr)\cap U|$ could be the entirety of $U$ if the diameter of $U$ is between $cr$ and $2cr$. Thus, in order to deploy our data structure, it is crucial that we have the ability to identify the \emph{thin level}, i.e., given $v$, the largest possible $r$ for which the assumption holds. Recall that due to efficiency concerns and the nature of LSH data structures, we set $s\leq n^\rho$.

To do so, we will utilize an approximate near neighbor counting procedure for $\ell_2$ norm, developed in~\cite{aimn23}. In essence, they develop an algorithm based on the LSH scheme of~\cite{alrw17} that, given an oblivious query $v$, it can with high probability return an approximate count ${\tt ans}$ such that
\begin{align*}
    (1-o(1)) \cdot |B(v, r)\cap U| \leq {\tt ans} \leq |B(v, cr)\cap U|+O(n^{\rho+o(1)}).
\end{align*}
We first note this provides a tool for us to do binary search on $r$, to locate the thin level: we could start with small $r$, and run the approximate counting procedure of~\cite{aimn23}. Conditioned on these counts being accurate, whenever we are still below or at the thin level, ${\tt ans}$ must be below $O(n^{\rho+o(1)})$. If we are at a level where $|B(v,r)\cap U|\geq \omega(n^{\rho+o(1)})$, then the counting data structure could successfully detect it, so the only troublesome case is when $|B(v,r)\cap U|$ is small but $|B(v,cr)\cap U|\geq \omega(n^{\rho+o(1)})$. We can indeed detect this case by searching over the range $(cr, c^2r)$, and the data structure should indeed report that the count is already too large. We could then refine our search by looking at the range $(r/c, r)$, in which the correct thin level is guaranteed to be in this interval.

One last issue remains: the~\cite{aimn23} data structure only works for oblivious queries; for adaptive queries, they run a net argument with $d$ independent copies. To improve upon that, we note that the data structure only outputs a real number, and thus it falls into the category of \emph{estimation} data structures, which can be augmented via the framework of~\cite{bkm+21}.

\begin{theorem}[Adaptive Approximate Near Neighbor Counting]
Let $U\subseteq \R^d$ be an $n$-point dataset and $\{v_1,\ldots,v_T\}\subseteq \R^d$ be a sequence of adaptive queries. Let $c\geq 1$, $r\geq 0$ be parameters and $\rho\in (0,1)$ be a parameter that  depends on $c$. There exists a randomized data structure with
\begin{itemize}
    \item Preprocessing time $\wt O(\sqrt T\cdot n^{1+o(1)}d)$;
    \item Space usage $\wt O(\sqrt T\cdot n^{1+o(1)}d)$;
    \item Given $v_t$, it outputs a number $\wt {\tt ans}$ such that
    \begin{align*}
        (1-o(1))\cdot |B_{\|\cdot\|_2}(v_t,r)\cap U| \leq \wt {\tt ans} \leq 1.01 |B_{\|\cdot\|_2}(v_t, cr)\cap U| + O(n^{\rho+o(1)})
    \end{align*}
    holds with probability at least $1-1/\poly(n)$. The time to output $\wt {\tt ans}$ is $\wt O(n^{\rho+o(1)}d)$.
\end{itemize}
\end{theorem}

\begin{proof}
The proof is a combination of the data structure of~\cite{aimn23} and augmentation of~\cite{bkm+21}. We start from the oblivious data structure. In the proof of Lemma 3.3 from~\cite{aimn23}, we note that their argument is to first provide a bound on the noiseless count, and then bound the error incurred by Laplacian noise. The noiseless count ${\tt ans}$ indeed has a bound
\begin{align*}
    (1-o(1))\cdot |B_{\|\cdot\|_2}(v_t,r)\cap U|\leq {\tt ans} \leq  |B_{\|\cdot\|_2}(v_t, cr)\cap U| + O(n^{\rho+o(1)})
\end{align*}
with high probability, albeit for oblivious $v_t$. To augment it for an adaptive adversary, we apply Theorem 3.1 of~\cite{bkm+21}, which states one could use $\sqrt T$ copies of the data structure when the output is a real number. Regarding the output quality, only the upper bound is increased by a factor of $1+\alpha$, if we set $\alpha=0.01$, we obtain the desired result.
\end{proof}

Given such an adaptive data structure for approximate near neighbor counting, we could then instantiate a meta algorithm for $\ell_2$ LSH under Assumption~\ref{assumption:general}. Without loss of generality, assume $U\subseteq \mathbb{S}^{d-1}$ and all queries lie on the unit sphere. We first pick a small discretization value $\tau$ and apply the transformation $x\mapsto \begin{bmatrix}
    \lfloor x_1/\tau\rfloor\cdot \tau \\
    \vdots \\
    \lfloor x_d/\tau\rfloor \cdot \tau
\end{bmatrix}$ to all points in $U$ and all queries. We use $\ov U$ and $\ov v$ to denote transformed points. Note that these points have their entries being a multiple of $\tau$, and the difference of the norm $\|\ov x-x\|_2 \leq \sqrt{d}\tau$, if we choose $\tau=(nd)^{-C}$ for a large enough constant $C$, then the difference is $1/\poly(n)$, which is negligible. The advantage of this discretization framework is that the difference of norms for points over the unit sphere could also be discretized into $O(1/\tau)$ levels. Hence we could perform binary search over these discrete $O(1/\tau)$ levels in $O(\log(1/\tau))=O(\log(nd))$ steps. We will prepare $O(\log (1/\tau))$ data structures by initializing an LSH for $r=2^i \tau$ where $i\in \{0,1,\ldots,\lceil\log_2(2/\tau)\rceil\}$ for $c\geq 2$. This ensures that for any $r\in [\tau,2]$, there exists a node in between $r/c$ and $r$, as desired.

\ifdefined\isarxiv
\section*{Acknowledgement}
Ying Feng was supported by an MIT Akamai Presidential Fellowship. 
David Woodruff would like to acknowledge support from a Simons Investigator Award and Office of Naval Research (ONR) award number N000142412647. Lichen Zhang was supported in part by NSF CCF-1955217 and NSF DMS-2022448. The corrections in this version were found and written with the help of Claude Opus 5.5. The authors have verified the corrections and take full responsibility for the content of this paper.
\bibliographystyle{alpha}
\bibliography{ref}
\else
\fi





\end{document}